\documentclass[a4paper,fleqn]{cas-dc}

\usepackage[numbers]{natbib}
\usepackage{caption,subcaption}
\usepackage{stackengine}
\usepackage{arydshln}

\graphicspath{{./figures/}} 

\def\tsc#1{\csdef{#1}{\textsc{\lowercase{#1}}\xspace}}
\tsc{WGM}
\tsc{QE}

\newcommand{\R}{\mathbb{R}}

\newcommand{\II}{{\mathcal I}}
\DeclareMathOperator{\adj}{adj}

\newtheorem{theorem}{Theorem}[section]
\newtheorem{lemma}[theorem]{Lemma}
\newtheorem{corollary}[theorem]{Corollary}
\newtheorem{example}{Example}[section]
\newdefinition{definition}{Definition}[section]
\newdefinition{remark}[definition]{Remark}
\newproof{proof}{Proof}
\newproof{pot}{Proof of Theorem \ref{thm}}

\newcommand{\END}{\hfill\mbox{\raggedright \mbox{\raisebox{2pt}{\scalebox{0.6}{$\Diamond$}}}}\medskip} 
\newcommand{\etal}{\textrm{et al.}}
\newcommand{\ignore}[1]{}

\newcommand{\pathto}{\rightsquigarrow}

\newcommand{\flatrightarrow}{\raisebox{-.2ex}{\mbox{\;\rotatebox{90}{\scalebox{1}[1.2]{$\bot$}}}\;\;}}
\newcommand{\biarrow}{\;\smash{{}^{\displaystyle \rightarrow}_{\displaystyle\leftarrow}}\;}

\begin{document}
\let\WriteBookmarks\relax
\def\floatpagepagefraction{1}
\def\textpagefraction{.001}
\let\printorcid\relax

\shorttitle{Homeostasis in Input-Output Networks}    

\shortauthors{Antoneli et al.}  

\title[mode = title]{Homeostasis in Input-Output Networks: \texorpdfstring{\\} {}
Structure, Classification  and Applications}  



\author[1]{Fernando Antoneli}

\cormark[1]


\ead{fernando.antoneli@unifesp.br}



\affiliation[1]{organization={Centro de Bioinform\'atica M\'edica, Universidade Federal de S\~ao Paulo},
            addressline={Edif\'{\i}cio de Pesquisas 2}, 
            city={S\~ao Paulo},
            postcode={04039-032}, 
            state={SP},
            country={Brazil}}


\cortext[cor1]{Corresponding author}


\author[2]{Martin Golubitsky}


\ead{golubitsky.4@osu.edu}



\affiliation[2]{organization={Department of Mathematics, The Ohio State University},
            addressline={231 W 18th Ave}, 
            city={Columbus},
            postcode={43210}, 
            state={OH},
            country={USA}}


\author[2]{Jiaxin Jin}


\ead{jin.1307@osu.edu}




\author[3]{Ian Stewart}


\ead{ins@maths.warwick.ac.uk}



\affiliation[3]{organization={Mathematics Institute, University of Warwick},
            addressline={Zeeman Building}, 
            city={Coventry},
            citysep={}, 
            postcode={CV4 7AL}, 
            country={UK}}



\begin{abstract}
Homeostasis is concerned with regulatory mechanisms, present in biological systems, where some specific variable is kept close to a set value as some external disturbance affects the system. 
Many biological systems, from gene networks to signaling pathways to whole tissue/organism physiology, exhibit homeostatic mechanisms. 
In all these cases there are homeostatic regions where the variable is relatively to insensitive external stimulus, flanked by regions where it is sensitive. 
Mathematically, the notion of homeostasis can be formalized in terms of an input-output function that maps the parameter representing the external disturbance to the output variable that must be kept within a fairly narrow range. 
This observation inspired the introduction of the notion of infinitesimal homeostasis, namely, the derivative of the input-output function is zero at an isolated point. 
This point of view allows for the application of methods from singularity theory to characterize infinitesimal homeostasis points (i.e. critical points of the input-output function). 
In this paper we review the infinitesimal approach to the study of homeostasis in input-output networks. 
An input-output network is a network with two distinguished nodes ‘input’ and ‘output’, and the dynamics of the network determines the corresponding input-output function of the system. 
This class of dynamical systems provides an appropriate framework to study homeostasis and several important biological systems can be formulated in this context. 
Moreover, this approach, coupled to graph-theoretic ideas from combinatorial matrix theory, provides a systematic way for classifying different types of homeostasis (homeostatic mechanisms) in input-output networks, in terms of the network topology. 
In turn, this leads to new mathematical concepts, such as, homeostasis subnetworks, homeostasis patterns, homeostasis mode interaction. 
We illustrate the usefulness of this theory with several biological examples: biochemical networks, chemical reaction networks (CRN), gene regulatory networks (GRN), Intracellular metal ion regulation and so on.
\end{abstract}




\maketitle


\section{Introduction}
\label{sec:INTRO}

The idea of homeostasis has its roots in the work of the French physiologist Claude Bernard \cite{bernard1865} who observed a kind of regulation in the `milieu int\'erieur' (internal environment) of human organs such as the liver and pancreas.
The term 'homeostasis', coined by the American physiologist Walter Cannon in 1926 \cite{cannon1929,cannon1939}, derives from the Greek language and refers to a capacity of maintaining \emph{similar stasis}.
In this context, `homeostasis' refers to certain regulatory mechanisms by which a feature is maintained at a steady state despite disturbances caused by changes in the environment.
At the scale of a whole organism, homeostasis manifests itself in many forms: body
temperature, blood sugar level, concentration of ions in body fluids with changes in the external environment. 

More recently, the study of homeostasis in biological systems through mathematical models has gained prominence.
For example, the extensive work of Nijhout, Reed, Best and collaborators~\cite{nrbu2004,best2009,nr2014,nbr2015,nbr2018} consider biochemical networks associated with metabolic signaling pathways.
Further examples include regulation of cell number and size \cite{lloyd2013}, control of sleep~\cite{wyatt1999}, and expression level regulation in housekeeping genes~\cite{antoneli2018}.  

In order to mathematically model `homeostasis' is necessary to have a framework where a precise definition can be given.
The simplest mathematical setup, very often used to model biological phenomena, is the theory of ordinary differential equations.
In this setting `homeostasis' can be interpreted in two mathematically distinct ways.  
One boils down to the existence of a `globally stable equilibrium'. 
Here, changes in the environment are considered to be perturbations of \emph{initial conditions} \cite{thomas1990}.
A stronger (and, in our view, more appropriate) usage works with a parametrized family of differential equations, with a corresponding family of stable equilibria.  
Now `homeostasis' means that some quantity defined through this equilibrium changes by a relatively small amount when the {\em parameter} varies by a much larger amount.
Homeostasis does not imply that the whole system remains invariant with change in external variables. 
In fact, changes in external variables can cause changes in certain internal variables, while other internal variables remain almost unchanged.
 
A precise definition of homeostasis can be give as follows.
Consider system of ODEs depending on an \emph{input parameter} $\II$, which varies over a range of external stimuli.
Suppose there is a family of equilibrium points $X(\II)$ and an observable $\phi$, such that the \emph{input-output function} $z(\II)=\phi(X(\II))$ is well-defined on the range of $\II$.
In this situation, we say that the system exhibits \emph{homeostasis} if, under variation of the input parameter $\II$, the input-output function $z(\II)$ remains approximately constant over the interval of external stimuli (see section \ref{SEC:STRUC}).  
Golubitsky and Stewart~\cite{gs2017} observe that homeostasis on some neighborhood of a specific value $\II_0$ follows from {\em infinitesimal homeostasis}, where $z'(\II_0) = 0$ and $'$ indicates differentiation with respect to $\II$.  This observation is essentially the well-known idea that the value of a function changes most slowly near a stationary (or critical) point.

Infinitesimal homeostasis is a \emph{sufficient} condition for homeostasis over some interval of parameters, but it is not a necessary condition.  A function can vary slowly without having a stationary point.
In applications, the quantity that experiences homeostasis, represented by the observable $\phi$ can be a function of several internal variables, such as a sum of concentrations, or the period of an oscillation (see subsection \ref{SS:PHCR}).  

There are other variations on the formulation of homeostasis.
Control-theoretic models of homeostasis often require \emph{perfect homeostasis}, also known as \emph{perfect adaptation}, in which the input-output function is exactly constant over the parameter range \cite{khammash2016,khammash2021a,khammash2021b}.
The requirement of perfect homeostasis is very strong, demanding that the model equations have a somewhat restricted functional form \cite{araujo2018}.

Another important idea of \cite{gs2017} is to take advantage of fact that infinitesimal homeostasis is suitable for analysis using methods from singularity theory (see subsection \ref{SSEC:SIOF}).
A motivational example for the use of singularity theory to analyze homeostasis is the regulation of the output `body temperature' in an opossum, when the input `environmental temperature' varies \cite{morrison1946}.
In \cite[Figs. 1--2]{gsahy2020} the graphs of body temperature against environmental temperature $\II$ are approximately linear, with nonzero slope, when $\II$ is either small or large, while in between is a broad flat region, where homeostasis occurs.
This general shape is called a `chair' by Nijhout and Reed~\cite{nr2014} (see also \cite{nbr2014, nrbu2004}) and suggest that there is a chair in the body temperature data of opossums~\cite{morrison1946} given by such a piecewise linear function \cite[Fig. 1]{nbr2014}.
Golubitsky and Stewart \cite{gs2017} take a singularity-theoretic point of view and suggest that chairs are better described locally by a homogeneous cubic function (that is, like $z(\II)\approx \II^3$) rather than by the previous
piecewise linear description.
Based on this suggestion, \cite{gs2017} goes on to show that the singularities of input-output functions are described by \emph{Elementary Catastrophe Theory} \cite{thom1969,thom1975,zeeman1977}.

The last contribution of \cite{gs2017} is that homeostasis can be thought of as a network concept.
This is motivated by the abundance of networks in biology, specially connected to ODE modeling, see e.g. \cite{nbr2014,rbgsn2017,ferrell2016}.
In fact, almost all ODE models in biology come together with a `wiring diagram', or a network, describing the interactions
among the elements in the model.
The recently published book \cite{gs2023} gives an exposition of a formal framework for studying networks of coupled ODEs, that have been developed the authors and collaborators in the past decade. 
More precisely, a network of coupled ODEs is determined by directed graph whose nodes and edges are classified into types. 
Nodes, or cells, represent the variables of a component ODE. 
Edges, or arrows, between nodes represent couplings
from the tail node to the head node. 
Nodes of the same type have the same phase spaces (up to a canonical identification); edges of the same type represent identical couplings.

In the setting of network coupled ODEs there is a preferable class of observables, namely the output of the node variables (coordinate functions).
Network systems are distinguished from large systems by the ability to keep track of the output from each node individually. 
Now homeostasis can be naturally defined as the fact that the output of a node variable (the `output node') is held approximately constant as other variables (other nodes) vary (perhaps wildly) under variation of an input parameter that affects another node (the `input node').
Placing homeostasis in the general context of network dynamics leads naturally to the methods reviewed here.

A special kind of network of coupled ODE, with two distinguished nodes (one input node and one output node)
is called an \emph{input-output network}.
It was introduced in \cite{wang2020} in their study of homeostasis on $3$-node networks.
Wang \etal~\cite{wang2021} extended the notion of \emph{input-output network} to arbitrary large networks and developed a combinatorial theory for the classification of `homeostasis types' in such networks.
A \emph{homeostasis type} is essentially a `combinatorial mechanism' that causes homeostasis in the full network and is represented by specific subnetworks, called \emph{homeostasis subnetworks}.
The homeostasis subnetworks can be subdivided into to classes called \emph{structural} and \emph{appendage}.

The motivation for the term \emph{structural homeostasis} comes from \cite{rbgsn2017}, where the authors identify the feedforward loop as one of the homeostastic motifs in $3$-node biochemical networks, that is, a generalized `feedforward loop' 
The intuition behind the term \emph{appendage homeostasis} is that homeostasis is generated by a cycle of regulatory nodes, that is, a generalized `feedback loop'.
The structural and appendage classes are abstract generalizations of the usual `feedforward' and `feedback' mechanisms. 
A striking outcome the approach of \cite{wang2021} is that they do not specify any homeostasis generating mechanisms at the outset and they find \emph{a posteriori} that there are essentially only the two types of homeostasis generating mechanisms: \emph{generalized feedback} and \emph{generalized feedforward}.

In Duncan \etal~\cite{duncan2024} the combinatorial formalism of \cite{wang2021} id further refined to allow one to determine, for each homeostasis subnetwork, which nodes are (and are not) going to be simultaneously homeostatic (besides the designated output node) when that particular homeostasis subnetwrok is the `homeostasis trigger', i.e. the one that causes homeostasis in the network.
These collections of simultaneous homeostatic nodes are called \emph{homeostasis patterns}.
The general picture that emerges from all these results is the following: the homeostasis subnetworks and the homeostasis patterns of network correspond uniquely to each other.
\pagebreak

\section{Structure}
\label{SEC:STRUC}

In this section we give the basic definition of homeostasis in a parametrized system of ordinary differential equations.
We start with a very general definition that will be specialized to more `concrete' situations throughout the paper.

\subsection{A Dynamical Formalism for Homeostasis}

Golubitsky and Stewart \cite{gs2017,gs2018} proposed a mathematical framework for the study of homeostasis based on dynamical systems theory (see \cite{gsahy2020} and the book \cite{gs2023}).

In this framework one considers a system of differential equations 
\begin{equation} \label{general_dynamics}
  \dot{X} = F(X, \II)
\end{equation}
with state variable $X=(x_1,\ldots,x_n) \in \R^n$ and input parameters $\II=(\II_1,\ldots,\II_N) \in \R^N$ representing the external input to the system.

Suppose that $(X^*,\II^*)$ is a linearly stable equilibrium of \eqref{general_dynamics}. 
By the implicit function theorem, there is a function $\tilde{X}(\II)$ defined in a neighborhood of $\II^*$ such that $\tilde{X}(\II^*) = X^*$ and $F(\tilde{X}(\II), \II) \equiv 0$. 

A smooth function $\phi:\mathbb{R}^{n}\to\mathbb{R}$ is called an \emph{observable}.
Define the \emph{input-output function} $z:\mathbb{R}^{N}\to\mathbb{R}$ associated to $\phi$ and $\tilde{X}$ as $z(\II)=\phi(\tilde{X}(\II))$.
The input-output function allows one to formulate several definitions that capture the notion of homeostasis (see~\cite{ma2009,ang2013,tang2016,gs2017,gs2018}).

\begin{definition} \label{D:inf_homeo}
Let $z(\II)$ be an input-output function.
We say that $z(\II)$ exhibits
\begin{enumerate}[(a)]
\item \emph{Perfect Homeostasis} on an open set $\Omega\subseteq\operatorname{dom}(z)$ if
\begin{equation} \label{definition_perfect_adaptation}
  \nabla z(\II)  = 0 
  \qquad\text{for all} \; \II \in\Omega
\end{equation}
That is, $z$ is constant on $\Omega$.
\item \emph{Near-perfect Homeostasis} relative to a \emph{set point} $\II_0 \in \Omega \subseteq \operatorname{dom}(z)$ if, for fixed $\delta$,
\begin{equation} \label{definition_near_perfect_adaptation}
  \| z(\II) - z(\II_0) \| \leqslant \delta
  \qquad\text{for all} \; \II\in\Omega
\end{equation}
That is, $z$ stays within the range $z(\II_0)\pm\delta$ over $\Omega$.
\item \emph{Infinitesimal Homeostasis} at the point $\II_0 \in \operatorname{dom}(z)$ if
\begin{equation} \label{definition_homeostasis}
  \nabla z(\II_0) = 0
\end{equation}
That is, $\II_0$ is a \emph{critical point} of $z$. 
\END
\end{enumerate}
\end{definition}

It is clear that perfect homeostasis implies near-perfect homeostasis, but the converse does not hold.  
Inspired by Nijhout, Reed, Best \etal~\cite{nrbu2004,best2009,nbr2014}, Golubitsky and Stewart \cite{gs2017,gs2018} introduced the notion of infinitesimal homeostasis which is intermediate between perfect and near-perfect homeostasis.
It is obvious that perfect homeostasis implies infinitesimal homeostasis. 
On the other hand, it follows from Taylor's theorem that infinitesimal homeostasis implies near-perfect homeostasis in a neighborhood of $\II_0$ (see \cite{gs2023} for details).
It is easy to see that the converse to both implications is not generally valid (see Reed \etal~\cite{rbgsn2017}).

The geometric interpretation of infinitesimal homeostasis is that, if $\II_0$ is a critical point of the input-output function, then $z(\II)$ differs from $z(\II_0)$ in a manner that depends quadratically (or to higher order) on $\|\II-\II_0\|$. This makes the graph of $z(\II)$ flatter than any growth rate with a nonzero linear term.

Definition \ref{D:inf_homeo} is very general in the sense that it allows for an arbitrary number of input parameters and and arbitrary (smooth) function of state variables as the output.
In the engineering literature, a system of inputs and outputs can be described as one of four types: SISO (single input, single output), SIMO (single input, multiple output), MISO (multiple input, single output), or MIMO (multiple input, multiple output).
By this terminology, Definition \ref{D:inf_homeo} describes MIMO systems.
However, as in bifurcation theory, the most important situation is when a single scalar input parameter is considered at a time.
Moreover, as we will define latter there is a natural class of dynamical systems, called {\em network dynamical systems} or {\em coupled cell systems}, that comes with a distinguished set of observables \cite{gs2023}. 
By adjusting the notion of homeostasis with a single input parameter to work in the class of network dynamical systems we obtain a very general description of SISO systems, called {\em input-output networks}.
Remarkably, it is in this class of systems that we can obtain the most complete theory of homeostasis (see Subsection \ref{SS:HION}).

Nevertheless, we can still say something in the most general setting.
We can use implicit differentiation and the cofactor formula for the inverse of a matrix to write a formula for 
$\nabla z$.
Let us denote the Jacobian matrix of $F$ at the equilibrium $(X(\mathcal{I}),\mathcal{I})$ (we will omit the \textasciitilde\, over $X$ from now on) by
\[
J = DF_{(X(\II),\II)}
\]
The partial derivatives of $F=(f_1,\ldots,f_n)$ with respect to $x_i$ are denoted by
$f_{j,x_i}$ and the partial derivatives with respect to $\II_M$ are denoted by $f_{j,\II_M}$. 
Let $F_{\II_m}$ denote the vector of partial derivatives of $F$ with respect to $\II_M$, that is, $F_{\II_M}=(f_{1,\II_M},\ldots,f_{n,\II_M})$.

\begin{lemma} \label{t:z_nabla}
Let $z(\mathcal{I})$ be an input-output function associated to a system of differential equations \eqref{general_dynamics}.
The gradient of $\nabla z(\II)$ with respect to the multiple input $\II$ is given by
\begin{equation} \label{e:z_prime}
\nabla z =  
\frac{-1}{\det J} 
\bigg(\langle \nabla \phi\, , \adj(J) F_{\II_M} \rangle \bigg)_{M=1}^N
\end{equation}
where $\adj(J)$ is the \emph{adjugate matrix} or \emph{classical adjoint} of $J$ (defined as the transpose of the cofactor matrix of $J$).
In particular, $\II_0$ is an infinitesimal homeostasis point of $z$ if and only if
\begin{equation} \label{e:cond_homeo}
 \langle \nabla \phi \, , \adj(J) F_{\II_M}\rangle \big|_{(X(\II_0),\II_0)} = 0 \;,\quad\forall\; M=1,\ldots,N
\end{equation}
\end{lemma}

\begin{proof}
The gradient of $z$ is
\begin{equation} \label{e:grad}
\nabla z = \bigg(\frac{\partial z}{\partial \II_1},\ldots,
\frac{\partial z}{\partial \II_N}\bigg)
\end{equation}
The partial derivative of $z(\mathcal{I})=\phi(X(\mathcal{I}))$ with respect to $\II_M$ (for each $M=1,\ldots,N$) is
\begin{equation} \label{e:chain_rule}
\frac{\partial z}{\partial \II_M} = \frac{\partial}{\partial \II_j} \phi(X(\II)) = \langle \nabla \phi\, ,  X_{\II_M} \rangle
\end{equation}
where $X_{\II_M}$ is the partial derivative of $X$ with respect to $\II_M$.
Implicit differentiation of the equation $F(X(\mathcal{I}),\mathcal{I})=0$, with respect to $\II_j$, yields the linear system
\begin{equation} \label{e:imp_diff}
J \, X_{\II_M} = - F_{\II_M}
\end{equation}
From \eqref{e:imp_diff} and the fact that $X(\II)$ is assumed to be a linearly stable equilibrium, for all $\II\in \Omega$, and hence $\det(J)\neq 0$ over $\Omega$, it follows that
\begin{equation} \label{e:adjoint}
X_{\II_M} = - \frac{1}{\det J} \adj(J) F_{\II_M}
\end{equation}
Therefore, substituting \eqref{e:adjoint} into \eqref{e:chain_rule} and then into \eqref{e:grad} we obtain \eqref{e:z_prime}.
\qed
\end{proof}

Lemma \ref{t:z_nabla} allows us to draw some general remarks about the function $\nabla z$.
The first observation is that the term $\adj(J) F_{\II_M}$ is a multivariate polynomial in the partial derivatives $f_{j,x_i}$ and $f_{j,\II_M}$, linear in $f_{j,\II_M}$.
Hence, the defining equations for homeostasis \eqref{e:cond_homeo} are polynomial functions in the partial derivatives $f_{j,x_i}$,  $f_{j,\II_M}$ and $\phi_{x_i}$, linear in both $f_{j,\II_M}$ and $\phi_{x_i}$.
This will become manifest in the several examples discussed in the paper.

The second observation is that there are two `trivial' ways in which homeostasis can occur: (i) by the existence of a critical point of $\phi$, that is, $\nabla\phi=0$ at $(X(\II_0),\II_0)$, and (ii) the simultaneous vanishing of all $F_{\II_M}$ at $(X(\II_0),\II_0)$.
The first case depends on the choice of the observable $\phi$ and it indeed occurs in some applications, for instance \cite{nrbu2004} (see Example \ref{ex:nrbu2004} below).
The second case is a highly non-generic situation which essentially says that the system is homeostatic ``at the input source''.
Henceforth, unless stated otherwise, we always assume the following {\em genericity condition}: the derivative of the vector field $F$ with respect to the input parameters generically satisfies
\begin{equation} \label{eq:genericity_condition}
F_{\II_M}\neq 0 
\; , \quad \forall\; M=1,\ldots,N
\end{equation}

The third observation is that the formula \eqref{e:z_prime} is valid as long as the equilibrium $X(\II)$ is linearly stable, since $\det(J)\neq 0$ at $(X(\II),\II)$.
Moreover, $\det(J)=0$ at some $(X(\II_0),\II_0))$ if and only if $\II_0$ is a steady-state bifurcation point.
Therefore, the input-output function is non-differentiable at the steady-state bifurcation points of $X(\II)$. 
Although it may be continuously extended beyond those points and become differentiable again.

The final observation is concerned with the fact that, very often, model equations depend on several parameters (besides the input parameters).
Because of this dependence the equilibrium family $\tilde{X}(\II)$ and input-output function may also depend on some of these parameters (generically, they depend on all of them), but they are suppressed from $z(\II)$ most the time.

More precisely, if we write the vector field in \eqref{general_dynamics} as $F(X,\II,\alpha)$, where $\alpha=(\alpha_1,\ldots,\alpha_\ell)$ is the vector of all scalar parameters that appear in the definition of $F$ then
the implicit function theorem applied to a point $(X^*,\II^*,\alpha^*)$ gives the function $\tilde{X}(\II,\alpha)$ defined in a neighborhood of $(\II^*,\alpha^*)$ such that $\tilde{X}(\II^*,\alpha^*) = X^*$ and $F(\tilde{X}(\II,\alpha),\II,\alpha) \equiv 0$.
Therefore, the input-output function $z$ can be written as $z(\II,\alpha)$ and infinitesimal homeostasis may depend on a fixed $\alpha_0$.
That is, infinitesimal homeostasis occurs at $\II_0$ only when $\alpha=\alpha_0$.

It is here that notion of infinitesimal homeostasis places the study of homeostasis in the context of singularity theory.
For instance, the dependence on the parameters $\alpha$ for the occurrence of infinitesimal homeostasis is related to higher degeneracy conditions, in addition to \eqref{definition_homeostasis}, which leads to distinct `forms' of infinitesimal homeostasis.

For example, let us consider the single parameter case $\II\in\mathbb{R}$.
We say that $z$ exhibits \emph{simple homeostasis} if
\begin{equation} \label{e:simple_homeostasis}
 z'(\II_0) = 0
 \quad\text{and}\quad 
 z''(\II_0) \neq 0
\end{equation}
In this case, $\II_0$ is called a Morse singularity (or a nondegenerate critical point) and it is persistent under any small perturbation of $\alpha$ (see subsection \ref{SSEC:SIOF}).
In fact, a Morse singularity is persistent under any small perturbation of $z$ in the space of smooth functions (with the appropriate topology) \cite{gg1973}. 
We say that $z$ exhibits \emph{chair homeostasis} if
\begin{equation} \label{e:chair_homeostasis}
 z'(\II_0) = 0
 \quad\text{and}\quad 
 z''(\II_0) = 0
 \quad\text{and}\quad 
 z'''(\II_0) \neq 0
\end{equation}
In this case, there is at least one parameter $\alpha_j$ (or a combination of several components of $\alpha$) that should be fixed together with $\II_0$ to have infinitesimal homeostasis. 
However, as we will see below, singularity theory implies that small perturbations of $z$ (that is, variation of the suppressed parameters) change the `almost flat region' only slightly.
Following Nijhout~\etal~\cite{nbr2014} we define a \emph{plateau} as a region of $\Omega$ over which $z$ is approximately constant.
In other words, the infinitesimal homeostasis points are not persistent, but the family of functions $z(\II,\alpha)$ have a plateau for all $(\II,\alpha)$ in a neighborhood of $(\II_0,\alpha_0)$ (see Definition \ref{d:universal_unfolding} in subsection \ref{SSEC:SIOF}). 

The observation that an unstable structure `determines' (unfolds) all the neighboring stable behaviors goes back to Ren\'{e} Thom's notion of a singularity as an \emph{organizing center} \cite{thom1975}.
Even though the singularity itself is unstable under perturbation, it `organizes' all its small perturbations into the universal unfolding. 
Unlike the singularity, the universal unfolding is structurally stable and thus contains all the possible persistent behaviors, as well as the number of parameters to (qualitatively) model all those behaviors.

This is important in applications because it may be very difficult to calculate exactly the values of $(\II_0,\alpha_0)$ where infinitesimal homeostasis occurs and then one must resort to numerical methods.
It is possible to compute parts of the graph of input-output functions and find the plateau using numerical methods for continuation of equilibrium points, such as \textsc{Auto} from \textsc{XPPAut} \cite{bard2002} (for example, see Figure \ref{F:KF} in subsection \ref{SS:TNION} and Figure \ref{e_coli_plot2} in subsection \ref{SS:BAC_CHEMO}).

Finally, we should mention that there is another notion of `robustness' often referred in the control theoretic literature, see  \cite{ang2013,tang2016,araujo2018,aoki2019} for example.
This notion is closely related to other concepts such as \emph{structural identifiability} \cite{meshkat2014,meshkat2015,meshkat2019,meshkat2022,meshkat2023} and \emph{dynamical compensation} \cite{karin2016,sontag2017,villaverde2017}.

\subsubsection{Examples of Homeostasis}

Here we present some examples of input-output functions associated to biological models.

\begin{example}[Folate Cycle] \label{ex:nrbu2004} \normalfont
Nijhout \etal~\cite{nrbu2004} developed a model for the folate cycle based on standard biochemical kinetics and used the model to provide new insights into several different
mechanisms of folate homeostasis.
Their model is described by 6 ODEs for the concentrations of the following substrates \cite[Eqs. 4--9]{nrbu2004}:  
(1) tetra\-hydro\-folate (THF), (2) 5-methyl\-tetra\-hydro\-folate (5mTHF), (3) di\-hydro\-folate (DHF), (4) 5,10-methyl\-ene\-tetra\-hydro\-folate (5,10-CH2-THF), (5) 5,10-methenyl\-tetra\-hydro-folate (5,10-CH=THF) and (6) 10-formyl\-tetra\-hydro\-folate (10f-THF).
The coupling structure of the model is represented by a diagram containing the 6 substrates and additional 12 enzymes \cite[Fig. 1]{nrbu2004}.
In the model, folate enters and leaves the cell as 5mTHF as indicated by the presence of the input parameter $\mathcal{I}=\mathcal{F}_{\mathrm{in}}-\mathcal{F}_{\mathrm{out}}$ in the corresponding differential equation
Here, $\mathcal{F}_{\mathrm{in}}$ and $\mathcal{F}_{\mathrm{out}}$ are the rates at which 5mTHF enters and leaves the cell, respectively.
The authors consider observables of the form $V_{\mathrm{ENZYME}}=\phi(S,F)$ representing the velocity of a reaction associated to $\text{ENZYME}$, called `fluxes' \cite[Tab. III]{nrbu2004}, in terms of substrates.
The smooth functions $\phi$ are given by certain expressions involving Michaelis-Menten functions \cite[Eqs. 1--3]{nrbu2004}.
The plots of the fluxes as functions of the input parameter 
were numerically computed and homeostasis is shown by the fact that velocities do not decline substantially until total
folate is close to $0$ \cite[Figs. 4--6]{nrbu2004}.
It can be checked that the functions $\phi$ have a critical point completely determined by the kinetic parameters entering its definition.
\END
\end{example}

\begin{example}[Kidney Flow] \label{ex:sgouralis2013} \normalfont
Sgouralis and Layton \cite{sgouralis2013} show by using a mathematical model how well the effect of the myogenic mechanism affects the afferent arteriole flow rate (output) in response to pressure variation (input).
The myogenic response in the smooth muscle enables nephrons to regulate the ﬂow in afferent arterioles for a wide range of pressure values
The solid curves in \cite[Fig. 4]{sgouralis2013} (see also \cite[Fig. 6]{nbr2014}) shows the percentage change at steady state in their model as the pressure in the afferent arteriole is varied from $50 \,mm\mathrm{Hg}$ to $230 \,mm\mathrm{Hg}$. The ﬂow rate is remarkably stable between $90 \,mm\mathrm{Hg}$ and $190 \,mm\mathrm{Hg}$. 
Below $90 \,mm\mathrm{Hg}$ the ﬂow drops quickly and above $190 \,mm\mathrm{Hg}$ the ﬂow begins to increase substantially giving the same chair-shaped curve that we have seen in the Introduction. 
The dashed curves shows what the ﬂow rate would be if the myogenic response were turned off and the walls of the arteriole responded passively.
\END
\end{example}

\begin{example}[Glycolysis Pathway] \label{ex:mulukutla14} \normalfont
Mulukutla \etal~\cite{mulukutla14} use a mathematical model based on reported mechanisms for the allosteric regulations of the enzymes in the  flux of glycolysis.
They show that glycolysis exhibits multiple steady state behavior segregating glucose metabolism into high flux and low flux states. 
Here, the input parameter is the glucose consumption and the output variable is the specific lactate production.
In \cite[Fig. 3]{mulukutla14} the authors show a phenomenon of bistability of homeostasis in cultured HeLa cells (see also \cite[Fig. 4]{duncan2019}).
The data suggests the existence of two homeostasis points on the lower branch: one which is apparent in the figure, and another which we would expect to see if it were extended further.
The similar glucose consumption rates of both types of cells in very high and very low glucose environments indicate two switches on the border of the plateaus. 
The homeostasis points and the bistable behavior are suggestive of the behavior depicted \cite[Fig. 2]{mulukutla14} obtained by simulation of the model equations.
Duncan and Golubitsky \cite[Fig. 3]{duncan2019} propose a mechanism combining homeostasis and bifurcation that qualitatively reproduces the bistability observed in \cite{mulukutla14}.
\END
\end{example}

\begin{example}[Extracellular Dopamine] \label{ex:best2009} \normalfont
Best \etal~\cite{best2009,nbr2014} propose a biological network to study homeostasis of extracellular dopamine (eDA) in response to variation in the activities of the enzyme tyrosine hydroxylase (TH) and the dopamine transporters (DAT). 
Hence, in this model we have two input parameters (eDA,TH) and one output variable DAT.
The authors derive a differential equation model for the biochemical network \cite[Fig. 1]{best2009}.
In \cite{nbr2014} the authors fix reasonable values for all parameters in the model with the exception of the concentrations of TH and DAT. 
In \cite[Fig. 8]{nbr2014} the authors show the equilibrium value of eDA as a function of TH and DAT in their model. 
The white dots indicate the predicted eDA values for the observationally determined values of TH and DAT in the wild-type genotype (large white disk) and the polymorphisms observed in human populations (small white disks).
These polymorphisms raise or lower the activity of TH and lower the activity of the DAT as indicated in \cite[Tab. 1]{nbr2014}.
This result is scientifically important because almost all of these genetic polymorphisms (white disks) lie on the plateau and that indicates homeostasis of eDA. 
Presumably, this is because polymorphisms that would result in large changes in eDA are deleterious and would, therefore, be removed from the population. Thus, although the polymorphisms have large effects on the activities of TH and
DAT, they have only a very small effect on the phenotypic variable eDA, and thus can be considered to be cryptic genetic variation.
Note that the plateau contains a line from left to right at about eDA $= 0.9$. 
In this respect, the surface graph in \cite[Fig. 8]{nbr2014} appears to resemble that of a non-singular perturbed hyperbolic umbilic (see Table \ref{T:classification}) and Figure \ref{F:plateau}. 
See also the level contours of the hyperbolic umbilic in Figure \ref{F:plateau}. 
This figure shows that the hyperbolic umbilic is the only low-codimension singularity that contains a single line in its zero set, see \cite{gs2018} for more details.
\END
\end{example}

\subsection{Singularity Theory of Input-Output Functions: Elementary Catastrophe Theory}
\label{SSEC:SIOF}

As discussed in the Introduction, Nijhout \etal~\cite{nbr2014} observe that homeostasis appears in many applications through the notion of a {\em chair}.  
Golubitsky and Stewart~\cite{gs2018}
observed that a chair can be thought of as a singularity of a scalar input-output function, one where $z(\II)$ `looks like' a homogeneous cubic, e.g. $z(\II) \approx \II^3$.  
More precisely, the mathematics of singularity theory~\cite{poston-stewart1978,golubitsky1978} replaces `looks like' by `up to a change of coordinates.'

\begin{definition}
\label{D:singularity}
A smooth function $z:\R^N\to\R$ is singular at a point $\II_0\in\R^N$ if 
\begin{equation} \label{e:singularity}
\nabla z(\II_0) = 0
\end{equation}
A point $\II_0$ satisfying \eqref{e:singularity} is called a {\em singularity}, or {\em critical point}, of $z$.
A point $\II_0$ is a {\em nondegenerate singularity} of $z$ if, in addition to \eqref{e:singularity}, its Hessian matrix is invertible.
Otherwise, $\II_0$ is a called a {\em degenerate singularity}.
\END
\end{definition}

\begin{definition} \label{D:right_equivalence}
Two smooth functions $z,w:\R^N\to\R$ are {\em right equivalent} on a neighborhood of $\II_0\in\R^N$ if 
\[
w(\II) = z(\Lambda(\II)) + K
\]
where $\Lambda:\R^N\to\R^N$ is an \emph{invertible change of coordinates}, or simply a \emph{diffeomorphism}, on a neighborhood of $\II_0$ and $K\in\R^N$ is a constant.
\END
\end{definition}

The simplest classification theorem of singularity theory states that $z:\R^N\to\R$ is right equivalent to $w(\II_1,\ldots,\II_N) = \sum_{i=1}^N\pm\II_i^2$ on a neighborhood of a singularity if and only if the singularity is nondegenerate.
This is the content of the classical Morse Lemma.
Geometrically, this says that the graph of $z$ `looks like' a quadratic function near a nondegenerate singularity.
Now, the the main goal of singularity theory is the study of degenerate singularities.

The transformations of the input-output map $z(\II)$ given in Definition~\ref{D:right_equivalence} are just the standard change of coordinates in elementary catastrophe theory~\cite{golubitsky1978, poston-stewart1978, zeeman1977}.
In \cite{gs2017,gs2018} it is shown that a special class of transformations on the vector field induces the class of right equivalences on the input-output functions.
Although their proof is formulated for a special class of observables (the projections onto the coordinate variables) the use of right equivalence is completely justified.
The proof that the action of the diffeomorphism $\Lambda$ on the vector field induces the correct action on the input-output function works in the general setting and the addition of the constant $K$ is justified because we are looking for plateaus on which $z$ is approximately constant but not specifically what that constant is.

Now we can therefore use standard results from elementary catastrophe theory to find normal forms and universal unfoldings of $z$, as we now explain.
Informally, the {\em codimension} of a singularity is the number of conditions on derivatives that determine it. 
This is also the minimum number of extra variables required to specify all small perturbations of the singularity, up to changes of coordinates. 
These perturbations can be organized into a family of maps called the {\em universal unfolding}, which has that number of extra variables.

\begin{definition} \label{d:universal_unfolding}
A smooth function $W(\II, a)$, $a\in\R^\ell$, is an {\em unfolding} of $w(\II)$ if $W(\II,0)=w(\II)$.
Here, $a\in\R^\ell$ is called {\em unfolding parameter}.
An unfolding $W$ of $w$ is called an {\em universal unfolding} if every unfolding $V(\II,b)$, $b\in\R^m$, of $w$ {\em factors through} $W$.  
That is, 
\begin{equation}  \label{e:factors_through}
V(\II,b) = W(\Lambda(\II,b), A(b)) + K(b)
\end{equation}
where $A,K:\R^m\to\R^\ell$ and $\Lambda:\R^N\times\R^m\to\R^N$.
\END
\end{definition}

It follows that {\em every} small perturbation $V(\cdot, b)$ is equivalent to a perturbation $W(\cdot,A(b))$ of $w$ in the $W$ family.  

\begin{example}[Single Input Parameter] \normalfont
Let $z:\R\to\R$ be singular at the origin.
Because $z(\II)$ is $1$-dimensional, we consider singularity types near the origin of a $1$-variable function $w(\II)$.  
Such singularities are determined by the first non-vanishing $\II$-derivative $w^{(k)}(0)$ (unless all derivatives vanish, which is an `infinite codimension' phenomenon that we do not discuss further).
If such $k$ exists, the normal form is $\pm\II^k$. 
When $k \geq 3$ the universal unfolding for catastrophe theory equivalence is
\[
\pm \II^k + a_{k-2}\II^{k-2} + a_{k-3}\II^{k-3} + \cdots + a_1\II
\]
for parameters $a=(a_j)$ and when $k = 2$ (a nondegenerate singularity) the universal unfolding is $\pm\II^2$.
The codimension $\ell$ in this setting is therefore $\ell=k-2$.
To summarize: the normal form of the input-output function near a nondegenerate singularity (codimension $0$) is
\begin{equation} \label{e:simple_homeo}
z(\II) = \pm\II^2
\end{equation}
and no unfolding parameter is required.
Here, $\II_0=0$ is a simple homeostasis point.
Similarly,
\begin{equation} \label{e:chair_homeo}
z(\II) = \pm\II^3
\end{equation}
is the normal form of the input-output function near a least degenerate singular point (codimension $1$), and 
\begin{equation} \label{e:chair_unfolding}
z(\II;a) = \pm\II^3 + a \II
\end{equation}
is a universal unfolding. 
The geometry of this universal unfolding is shown in Figure \ref{F:unfolding_chair_sing}.
Here, the singularity occurs when $(\II_0,a_0)=(0,0)$ and it is a chair homeostasis point, Figure \ref{F:unfolding_chair_sing}(b).
When $a\neq 0$ there is no singularity at $\II_0=0$ but the plateau persists all $a$ near $a_0=0$.
The difference is that when $a<0$ the plateau has two nondegenerate singular points, Figure \ref{F:unfolding_chair_sing}(a) and when $a>0$ the plateau has no singular points, Figure \ref{F:unfolding_chair_sing}(c).
See \cite{castrigiano2004} for details.
\END
\end{example}

\begin{figure*}[!bt]
\begin{subfigure}[b]{0.3\textwidth}
\centering
\includegraphics[width=\textwidth, trim=2mm 2mm 0.0mm 0.0mm, clip=true]{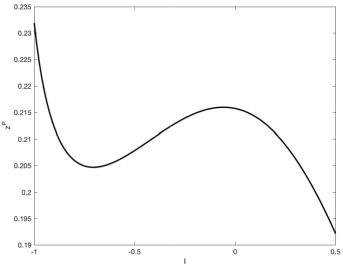}
\caption{$a<0$}
\label{F:a_samller_0}
\end{subfigure} 
\centering
\begin{subfigure}[b]{0.3\textwidth}
\centering
\includegraphics[width=\textwidth, trim=2mm 2mm 0.0mm 0.0mm, clip=true]{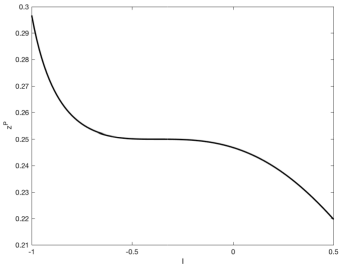}
\caption{$a=0$}
\label{F:a_equals_0}
\end{subfigure}
\centering
\begin{subfigure}[b]{0.3\textwidth}
\centering
\includegraphics[width=\textwidth, trim=2mm 2mm 0.0mm 0.0mm, clip=true]{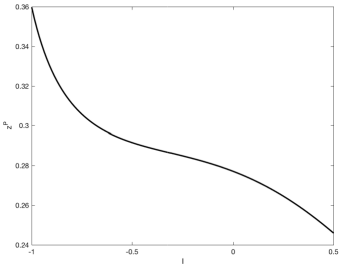}
\caption{$a>0$}
\label{F:a_bigger_0}
\end{subfigure}
\caption{\label{F:unfolding_chair_sing}
Input-output function exhibiting chair homeostasis at some point $\II_0$. 
The corresponding to the universal unfolding is $z(\II;a) = -\II^3 + a \II$, up to right equivalence.
(a) the unfolding parameter $a<0$ and the degenerate singularity unfolds into two nondegenerate singularities; (b) the unfolding parameter $a=0$ and the singularity is degenerate; (c) the unfolding parameter $a>0$ and the degenerate singularity unfolds into no singular points.}
\end{figure*}

\begin{example}[Multiple Input Parameters] \normalfont
We will consider the $N = 2$ case. 
Table \ref{T:classification} summarizes the classification when $N = 2$, so $\II= (\II_1, \II_2) \in \R^2$. 
Here the list is restricted to codimension $\leq 3$.
The associated geometry, especially for universal unfoldings, is described in~\cite{brocker1975, gibson1979, poston-stewart1978} up to codimension 4.
Singularities of much higher codimension have also been classified, but the complexities increase considerably. 
For example, Arnold~\cite{arnold1976} provides
an extensive classification up to codimension 10 (for the complex analog).
Note that in the $N = 2$ case the normal forms for $N = 1$ appear again, but now there is an extra quadratic term $\pm \II_2^2$. 
This term is a consequence of the {\emph splitting lemma} in singularity theory, arising here when the Hessian of $z$ has rank 1 rather than rank 0 (corank 1 rather than corank \cite{brocker1975, poston-stewart1978,zeeman1977}.
The presence of the $\pm \II_2^2$ term affects the range over which $z(\II)$ changes when $\II_2$ varies, but not when $\II_1$ varies.
\END
\end{example}

The standard geometric features considered in catastrophe theory focus on the gradient of the function $z(\II)$ in normal form.
In contrast, what matters here is the function itself. 
Specifically, we are interested in the region in the $\II$-space where the function $z$ is approximately constant, i.e. the plateau of $z$.

More specifically, for each normal form $z(\II)$ we choose a small $\delta>0$ 
and form the set
\begin{equation} \label{e:plateau}
P_\delta = \{\II\in\Omega: \|z(\II)\| \leq \delta\}.
\end{equation}
This is the $\delta$-\emph{plateau} region on which $z(\II)$ is approximately constant, where $\delta$ specifies how good the approximation is.  
Even though the definition of infinitesimal homeostasis is qualitative it is possible to extract quantitative bounds on the size of the plateau using Taylor's theorem with remainder.
See \cite[Chap. 5]{gs2023} for details.

As we mentioned before, if $z$ is perturbed slightly (by variation of the suppressed parameters), the plateau $P_\delta$ varies continuously.
Therefore, one can compute the approximate plateau by focusing on the singularity, rather than on its universal unfolding. 
In other words, for sufficiently small perturbations plateaus of singularities depend mainly on the singularity itself and not on its universal unfolding.

This observation is important because the universal unfolding has many zeros of the gradient of $z(\II)$, hence `homeostasis points' near which the 
value of $z(\II)$ varies more slowly than linear.  
However, this structure seems less important when considering the relationship of infinitesimal 
homeostasis with homeostasis.  
See the discussion of the unfolding of the chair summarized in \cite[Figure~3]{gs2015}.

The `qualitative' geometry of the plateau -- that is, its differential topology and associated invariants -- is characteristic of the singularity. 
This offers one way to infer the probable type of singularity from numerical data; it also provides information about the region in which
the system concerned is behaving homeostatically.  
We do not develop a formal list of invariants here, but we indicate a few possibilities.

The main features of the plateaus associated with the six normal forms of Table \ref{T:classification}
are illustrated in Figure \ref{F:plateau}. 
Figure plots, for each normal form, a sequence of contours from $-\delta$ to $\delta$; the union  is a picture of the plateaus.
By unfolding theory, these features are preserved by small perturbations of the model, and by the choice of $\delta$ in~\eqref{e:plateau} provided it is sufficiently small.
Graphical plots of such perturbations (not shown) confirm this assertion.

\begin{table*}[!hbt]
\caption{Classification of singularities of input-output maps $\R^2 \rightarrow \R$ of codimension $\leq 3$.}
\label{T:classification}
\begin{tabular}{llcl}
\toprule
name & normal form & codim & universal unfolding \\
\midrule
Morse (simple${}^\pm$) & $\pm \II_1^2 \pm \II_2^2$ & 0 & $\pm \II_1^2 \pm \II_2^2$ \\
fold (chair) & $\II_1^3 \pm \II_2^2$ & 1 & $\II_1^3 + a\II_1 \pm \II_2^2$\\
cusp${}^\pm$ & $\pm \II_1^4 \pm \II_2^2$ & 2 & $\pm \II_1^4 + a\II_1^2 + b\II_1 \pm \II_2^2$ \\
swallowtail  & $\II_1^5\pm \II_2^2$ & 3 & $\II_1^5 + a\II_1^3 + b\II_1^2 +c\II_1  \pm \II_2^2$ \\
hyperbolic umbilic & $\II_1^3+\II_2^3$ & 3 & $\II_1^3+\II_2^3 +a\II_1\II_2 + b\II_1 +c\II_2$\\
elliptic umbilic & $\II_1^3-3\II_1\II_2^2$ & 3 & $\II_1^3-3\II_1\II_2^2 + a(\II_1^2+\II_2^2) + b\II_1 +c\II_2$\\
\bottomrule
\end{tabular}
\end{table*}


\begin{figure*}[h!]  
\centerline{%
\includegraphics[width=1.75in]{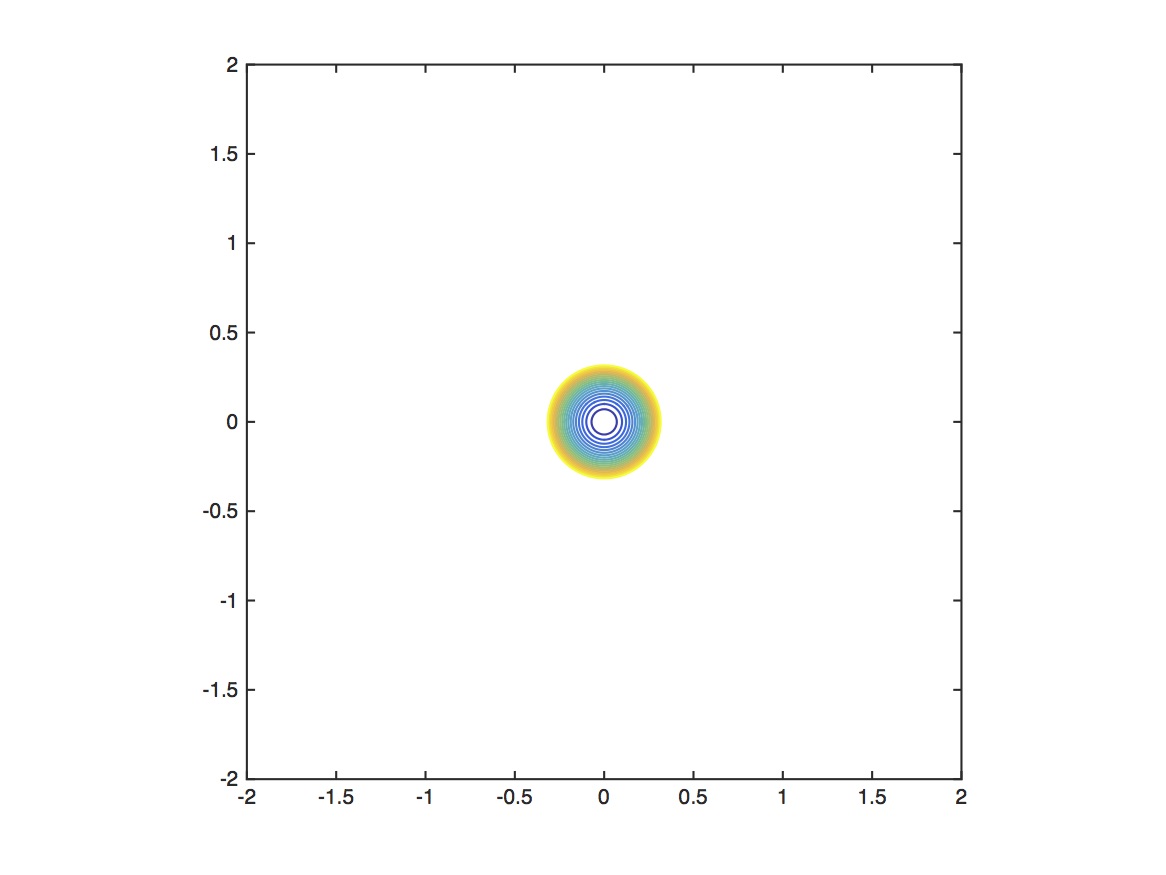} 
\includegraphics[width=1.75in]{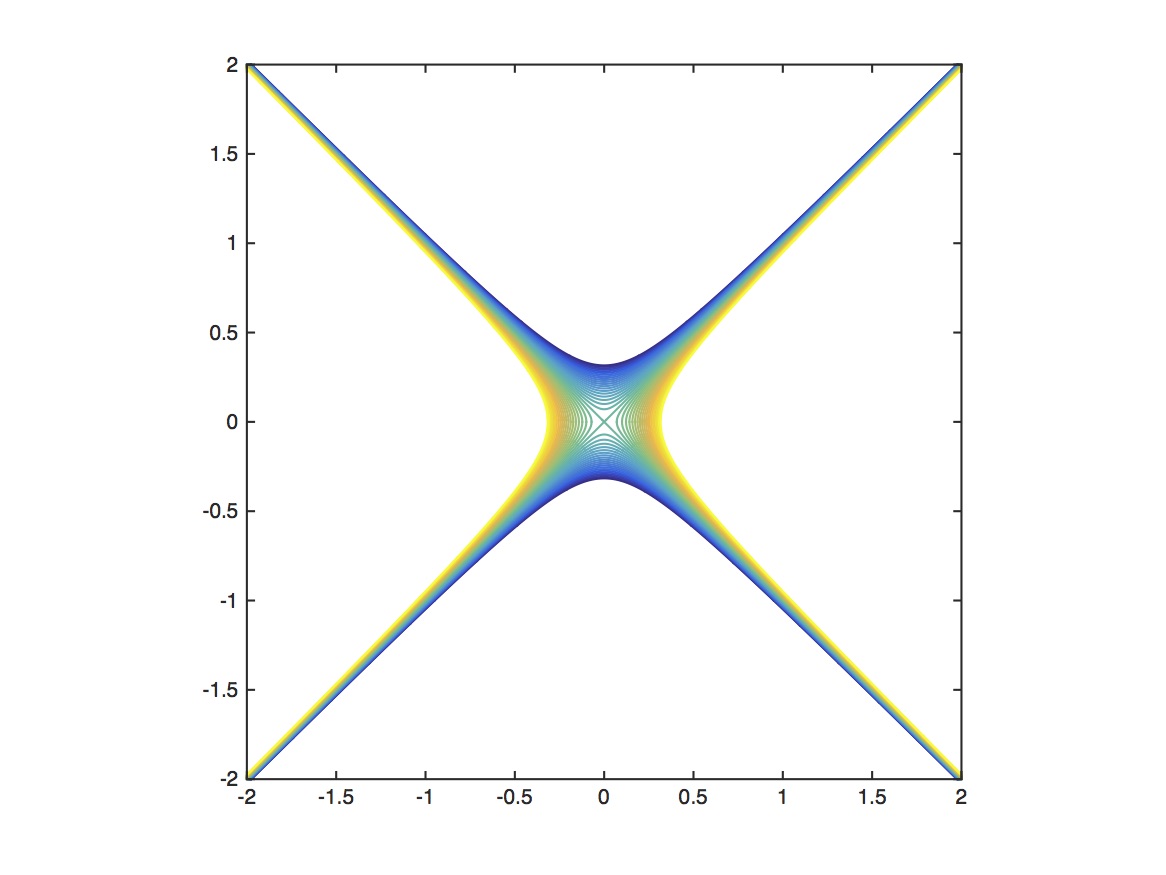} 
\includegraphics[width=1.75in]{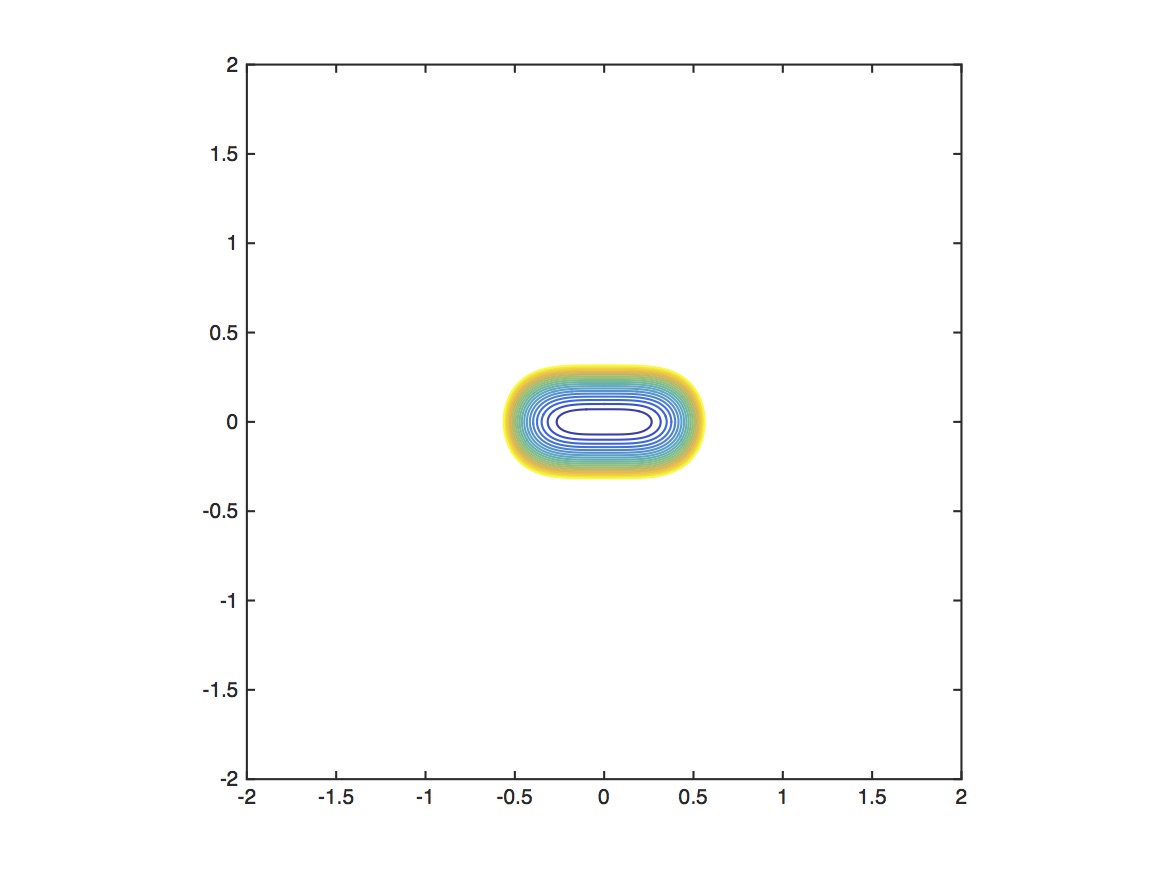} 
\includegraphics[width=1.75in]{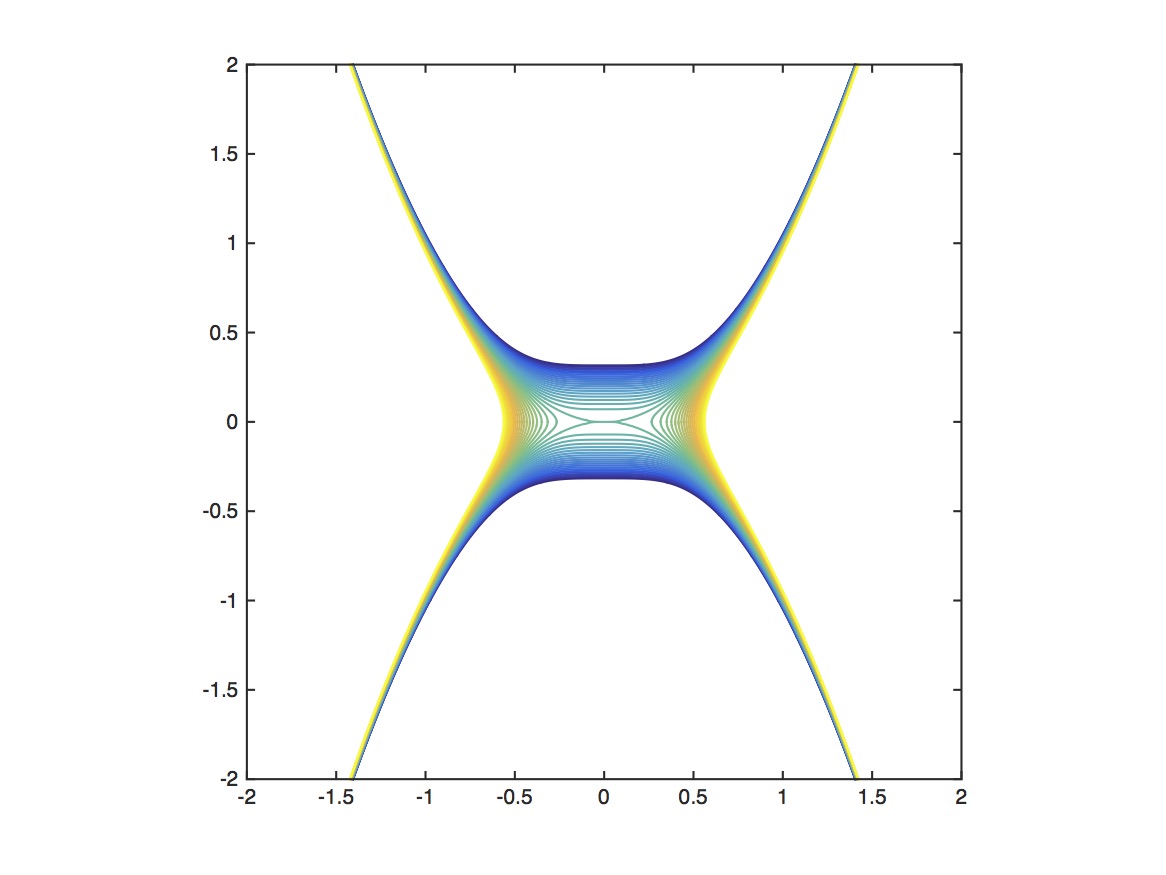}}
\centerline{simple$^+$ \hspace{1.3in} simple$^-$ \hspace{1.3in} cusp$^+$ \hspace{1.5in} cusp$^-$}
\vspace{0.1in}
\centerline{%
\includegraphics[width=1.75in]{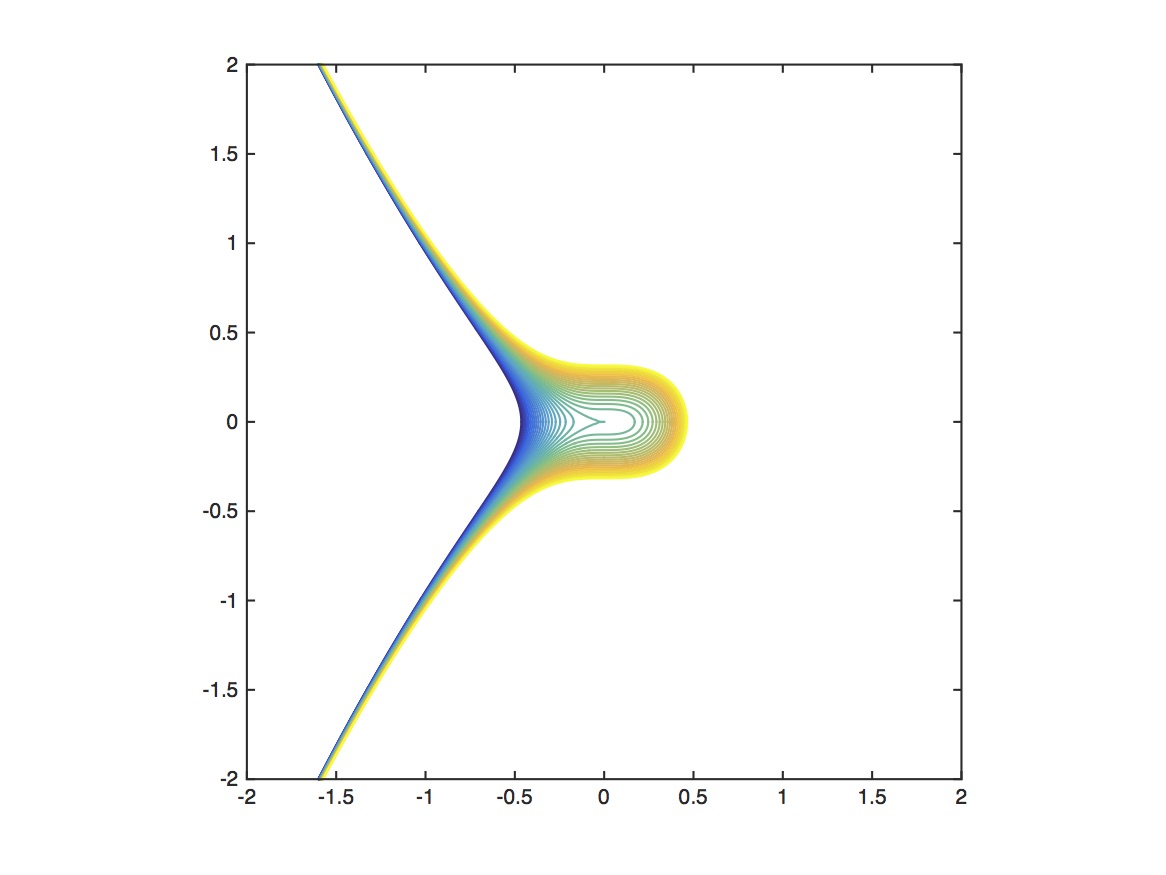} 
\includegraphics[width=1.75in]{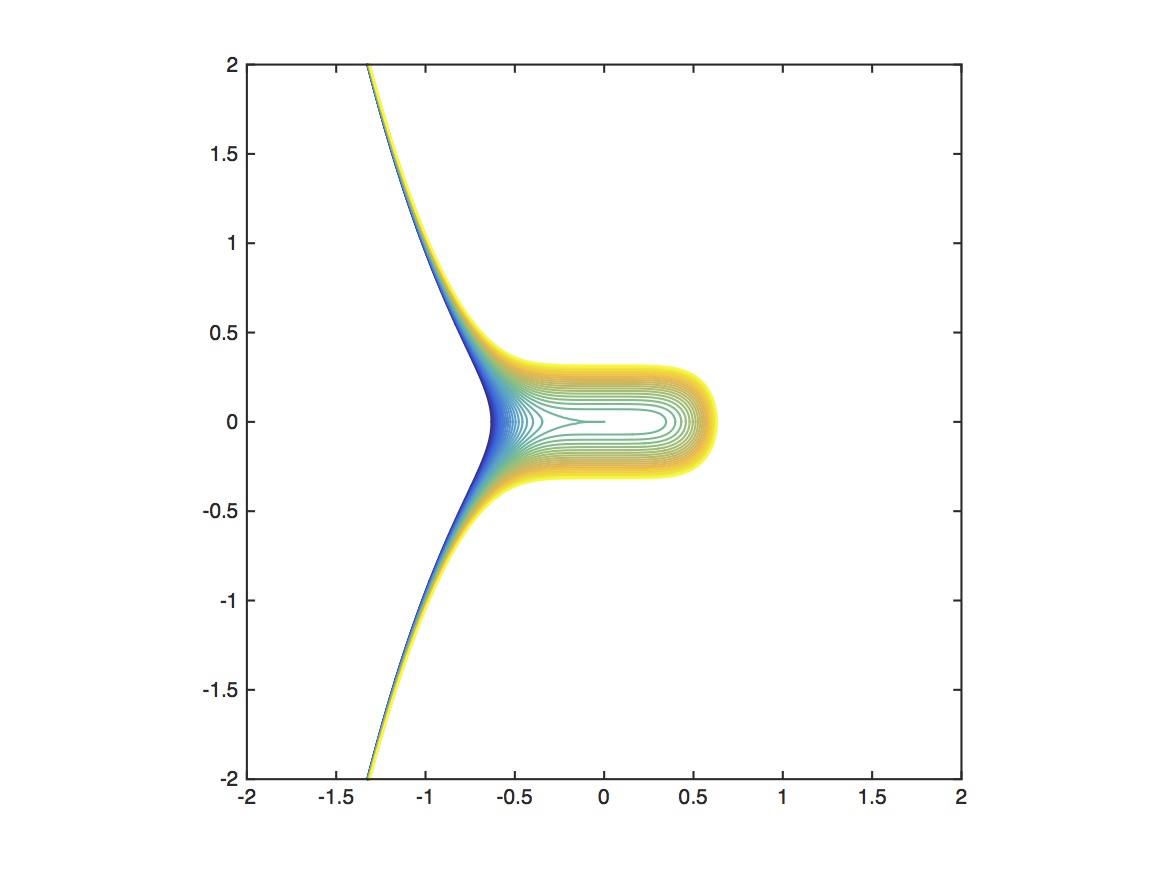} 
\includegraphics[width=1.75in]{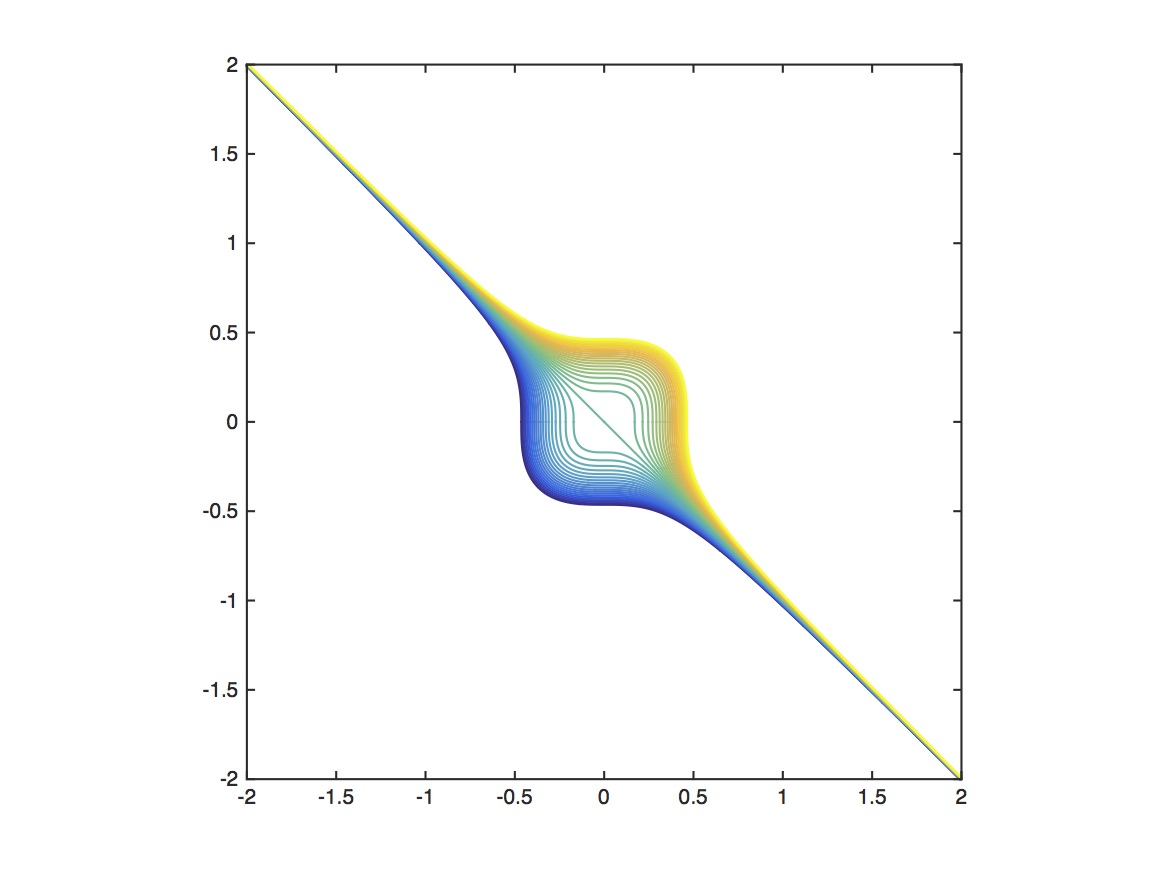} 
\includegraphics[width=1.75in]{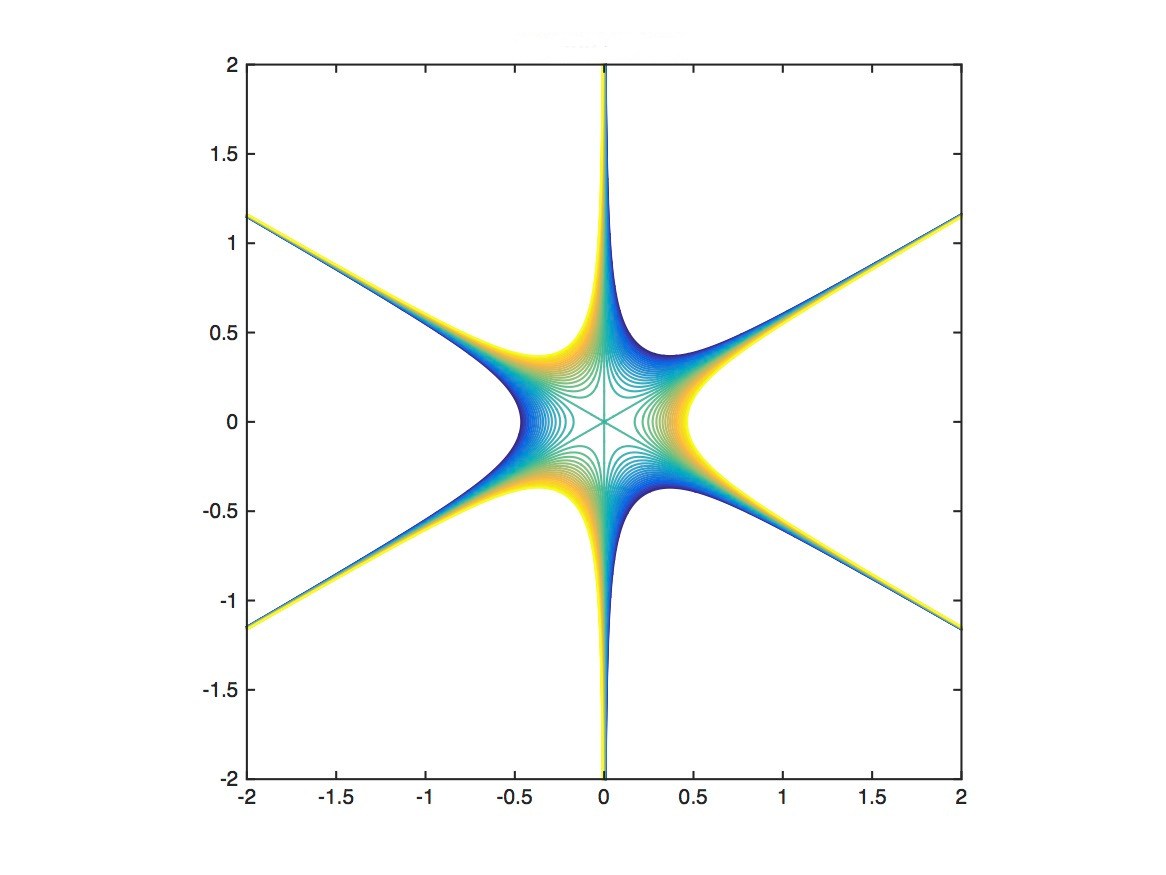}}
\centerline{chair \hspace{1.3in} swallowtail \hspace{1.2in} hyperbolic \hspace{1.2in} elliptic}
\caption{Plateaus shown by contour plots for each singularity in Table~\ref{T:classification}. Reproduced from \cite[Fig. 4]{gs2018}.  
Each plot has 200 equally spaced contour levels for $\delta$ from $-0.2$ to $0.2$.}  
\label{F:plateau}
\end{figure*}

The notion of infinitesimal homeostasis point can be generalized in several directions.
For instance, Andrade \etal~\cite{andrade2022} consider the situation where the infinitesimal homeostasis point $\II_0$ is at the boundary of the domain of the input-output function, which may be finite or infinity.
They introduce the notion of \emph{asymptotic infinitesimal homeostasis} for input-output functions $z$ defined on a domain of the form $D=[\II^{-},+\infty)$, for example.
In this case we say that $z$ exhibits \emph{asymptotic infinitesimal homeostasis} at $+\infty$ if
\[
 \lim_{\mathcal{I} \to \infty} z'(\mathcal{I}) = 0.
\]
In \cite{andrade2022} the authors use some simple criteria for the existence of asymptotic infinitesimal homeostasis to study a model for copper homeostasis (see subsection \ref{SS:METAL}).

\begin{theorem}
Let $z:D \to \mathbb{R}$, with $D=(\mathcal{I}^{-}, +\infty)$ be a smooth function.
\begin{enumerate}[(1)]
\item If $z$ exhibits near-perfect homeostasis on $D$ (i.e., it is upper and lower bounded on $D$) and is monotonic then $z$ exhibits asymptotic infinitesimal homeostasis.
\item {\rm (Hadamard's Lemma~\cite{hadamard1914})} If the first $n$ derivatives $z',\ldots,z^{(n)}$ of $z$ are bounded, then $\displaystyle\lim_{\II\to+\infty}z(\II)=c$ with $c\in \mathbb{R}$ implies that
\[
 \lim_{\II\to+\infty} z^{(j)}(\II)=0 
\]
for all $1\leq j \leq n$.
\end{enumerate}
\end{theorem}

It is interesting to mention that asymptotic infinitesimal homeostasis is related to asymptotic uniformity \cite{gabriel2023} and singularities at infinity \cite{tibuar1999}.

\subsection{Homeostasis in Input-Output Networks}
\label{SS:HION}

In this subsection we introduce input-output networks.
This concept is crucial for the development of a rich combinatorial theory.

\begin{definition} \label{defn:io_network}
An \emph{input-output network} is a directed graph $\mathcal{G}$ with a distinguished input node $\iota$ and a distinguished output node $o$. 
The network $\mathcal{G}$ is a \emph{core network} if every node in $\mathcal{G}$ is downstream from the input node $\iota$ and upstream from the output $o$. \END
\end{definition}

An {\em admissible system} of differential equations associated with $\mathcal{G}$ has the form
\begin{align}\label{eq:system}
\begin{split}
    \dot{x}_\iota &= f_\iota(x_\iota,x_\rho,x_o,\II) \\
    \dot{x}_\rho &= f_\rho(x_\iota,x_\rho,x_o) \\
    \dot{x}_o &= f_o(x_\iota,x_\rho,x_o)
\end{split}
\end{align}
where $\II\in\R$ is the input parameter, $X = (x_\iota,x_\rho,x_o) \in \R\times \R^n \times \R$ is the vector of state variables associated to nodes in $\mathcal{G}$, and $F(X,\II) = (f_\iota(X,\II),f_\rho(X),f_o(X))$ is a smooth family of mappings on the state space $\R\times \R^n \times \R$. 
Note that $\II$ appears only in the equation of system \eqref{eq:system} corresponding to the input node.

We can write \eqref{eq:system} as 
\[
 \dot{X} = F(X,\II).
\]
We denote the partial derivative of the function associated to node $j$ with respect to the state variable associated to node $k$, also called \emph{linearized couplings}, by 
\[
f_{j,x_k} = \frac{\partial}{\partial x_k} f_j
\]
We assume $f_{j,k} \equiv 0$ precisely when no arrow connects node $k$ to node $j$. 
That is, $f_j$ is independent of $x_k$ when there is no arrow $k\to j$.  
This is a modeling assumption made in $\mathcal{G}$.
Finally, we assume the \emph{genericity condition} \eqref{eq:genericity_condition}, which in this case reads
\[
 f_{\iota,\II} \neq 0
\]
generically. 

Suppose $\dot{X} = F(X,\II_0)$ has a hyperbolic equilibrium at $X^*$ at $\II^*$.
As before, we can define an input-output function $\II \mapsto x_o(\II)$ defined on a neighborhood of $\II^*$. 

\begin{remark}[SISO Systems] \label{rmk:SISO}
Each node in an input-output network $\mathcal{G}$ corresponds to a one-dimensional state variable of the admissible system.
In particular, the output node corresponds to a scalar quantity and the input parameter is a scalar quantity. 
The class of systems described by input-output networks can be seen as the class of {\em single input, single output} (SISO) systems.
As we will see later, it is possible to extend this definition to include input-output networks with multiple input nodes, but a single input parameter and single output, and multiple inputs, single output (MISO). \END
\end{remark}

\begin{remark}[Input Node $\neq$ Output Node] \label{rmk:input=output}
In Definition \ref{defn:io_network} we explicitly exclude the possibility that the output node is one of the input nodes.
For now, this assumption is taken into account purely for the sake of convenience. 
In fact, all the results should be valid when the input $=$ output, but then all the theorems and proofs should be properly adapted to take this particular case into account. 
This possibility will be considered in greater detail in Subsection \ref{SS:METAL}
\END
\end{remark}

\begin{remark}[Multiple Outputs]
\label{rmk:multiple_outputs}
It is important to clarify why the input-output function $z$ should be always real-valued ($1$-dimensional).
This seems to be the correct way to formulate the notion of simultaneously tracking `multiple outputs'.
Even in the control theoretic context of single input, multiple output (SIMO) and multiple input, multiple output (MIMO) the observable $\phi$ is real-valued (typically it is a linear functional).
Actually, there is a singularity theoretic reason for this requirement, as well.

Let us try to define a two-output input-output function by considering a network with one input node $x_\iota$ and two output nodes $x_{o_1}$ and $x_{o_2}$.
Then the input-output function would be
a mapping $\mathbb{R}\to\mathbb{R}\times\mathbb{R}$ given by $\II\to\big(x_{o_1}(\II),x_{o_2}(\II)\big)$.
This, in turn, would force us to consider more general coordinate changes on the target space $\mathbb{R}\times\mathbb{R}$ (recall that in the real-valued situation the action of the right equivalence on the target is by a translation by a constant).
Even in this simple two-output case one is led to consider general \emph{network preserving coordinate changes} \cite{gs2017b,antoneli2022}.
Typically, the most natural coordinate change on $\mathbb{R}\times\mathbb{R}$ that preserves the network structure is by a diagonal diffeomorphism $\varphi:\mathbb{R}\times\mathbb{R}\to\mathbb{R}\times\mathbb{R}$, given by $\varphi(x_1,x_2)=(\varphi_1(x_1),\varphi_2(x_2))$, where $\varphi_j:\mathbb{R}\to\mathbb{R}$ are diffeomorphisms.
The set of smooth mappings of the form $\mathbb{R}\to\mathbb{R}\times\mathbb{R}$, with the equivalence relation generated by a pair of diffeomorphisms $(\psi,\varphi)$, where $\psi$ acts on the source $\mathbb{R}$ and $\varphi$
is a diagonal diffeomorphism acting on the target 
$\mathbb{R}\times\mathbb{R}$ is called the set of \emph{divergent diagrams} \cite{mancini2002}.
The classification of singularities of divergent diagrams have been obtained and the conclusion is that all `degenerate' singularities have \emph{infinite codimension} 
\cite{dufour1977,dufour1979b,dufour1979a,mancini2002}.
Therefore, we would not be able to define an analogue of chair homeostasis (nor any other higher order singularity) in this setting.
\END
\end{remark}

Wang~\etal~\cite{wang2021} obtain a formula for the derivative of the input output function $x_o(\II)$, with respect to $\II$, that is a particular case of formula \eqref{e:z_prime} for input-output networks.
Here, the observable is the projection onto the output coordinate.
However, in this case the formula for 
$x_o'(\II)$ is much simpler than the general case and provides a straightforward criterion for the occurrence of homeostasis in an input-output network.

The \emph{homeostasis matrix} of \eqref{eq:system} is the $(n+1)\times (n+1)$ matrix $H$ is obtained from the $(n+2)\times(n+2)$ Jacobian matrix $J$ of \eqref{eq:system} by deleting its first row and last column. 
Indeed
\begin{align*}
J = \begin{pmatrix}
f_{\iota,\iota} & f_{\iota,\rho} & f_{\iota, o} \\
f_{\rho,\iota} & f_{\rho,\rho} & f_{\rho,o} \\
f_{o,\iota} & f_{o,\rho} & f_{o,o}
\end{pmatrix}
\Longrightarrow
H = \begin{pmatrix}
f_{\rho,\iota} & f_{\rho,\rho} \\
f_{o,\iota} & f_{o,\rho}
\end{pmatrix}
\end{align*}
where $J$ and $H$ are both functions of $(X(\II),\II)$ as in \eqref{eq:system}.  More precisely: 

\begin{theorem}[{\cite[lemma 1.5]{wang2021}}] \label{lem:irreducible}
The derivative of the input-output function $x_o(\II)$, with respect to $\II$ is given by
\begin{equation}
 x_o'(\II) = \pm f_{\iota,\II} \frac{\det(H)}{\det(J)} 
\end{equation}
Therefore, \eqref{eq:system} undergoes infinitesimal homeostasis at $\II_0$ if and only if $\det(H) = 0$, evaluated at $(X_0,\II_0)$. 
\end{theorem}

Homeostasis in a given network $\mathcal{G}$ can be determined by analyzing a simpler network that is obtained by eliminating certain nodes and arrows from $\mathcal{G}$.  We call the network formed by the remaining nodes and arrows the {\em core subnetwork}.

\begin{definition} \label{D:updown} \normalfont
A node $\tau$ in a network $\mathcal{G}$ is \emph{downstream} from a node $\rho$ in $\mathcal{G}$ if there exists a path in $\mathcal{G}$ from $\rho$ to $\tau$.  Node $\rho$ is \emph{upstream} from node $\tau$ if $\tau$ is downstream from $\rho$. 
\end{definition}

These relationships are important when trying to classify infinitesimal homeostasis.  For example, if the output node $o$ is not downstream from the input node $\iota$, then the input-output function $x_o(\II)$ is identically constant in $\II$.  Although technically this is a form of infinitesimal homeostasis, it is an uninteresting form.

\begin{definition} \label{D:core} \normalfont
Let $\mathcal{G}$ be an input-output network.
\begin{enumerate}[(a)]
\item The input-output network is a {\em core network} if every node is both upstream from the output node and downstream from the input node.  
\item Every input-output network $\mathcal{G}$ has a core subnetwork $\mathcal{G}_c$ whose nodes are the nodes in $\mathcal{G}$ that are both upstream from the output node and downstream from the input node and whose arrows are the arrows in $\mathcal{G}$ whose head and tail nodes are both nodes in $\mathcal{G}_c$. 
\end{enumerate}
\end{definition}

\begin{theorem}[{\cite[Thm. 2.4]{wang2021}}] \label{T:coreA}
Let $\mathcal{G}$ be an input-output network and let $\mathcal{G}_c$ be the associated core subnetwork. The input-output function associated with $\mathcal{G}_c$ has a point of infinitesimal homeostasis at $\II_0$ if and only if the input-output function associated with $\mathcal{G}$ has a point of infinitesimal homeostasis at $\II_0$.
\end{theorem}

It follows from Corollary~\ref{T:coreA} that classifying infinitesimal homeostasis for networks $\mathcal{G}$ is equivalent to classifying infinitesimal homeostasis for the core subnetwork $\mathcal{G}_c$.

\begin{definition} \label{D:backward} \normalfont
Let $\mathcal{G}_1$ and $\mathcal{G}_2$ be two input-output networks with the same number of nodes.
\begin{enumerate}[(a)]
\item We say that $\mathcal{G}_1$ and $\mathcal{G}_2$ are \emph{core equivalent} if the determinants of their homeostasis matrices are identical as polynomials in the variables $f_{j,x_k}$.
\item A {\em backward arrow} is an arrow whose head is the input node $\iota$ or whose tail is the output node $o$.
\end{enumerate}
\end{definition}

\begin{corollary} \label{P:core_equivalent}
If two core networks differ from each other by the presence or absence of backward arrows, then the core networks are core equivalent.
\end{corollary} 
\begin{proof}
The linearized couplings associated to backward arrows are of the form $f_{\iota,x_k}$ and $f_{k,x_o}$, which do not appear in the homeostasis matrix $H$. \qed
\end{proof}

Therefore, backward arrows can be ignored when computing infinitesimal homeostasis with the homeostasis matrix $H$.  However, backward arrows cannot be totally ignored, since they are involved in the determination of both the equilibria of \eqref{eq:system} and their stability.

Corollary~\ref{P:core_equivalent} can be generalized to a theorem giving necessary and sufficient graph theoretic conditions for core equivalence \cite[Theorem 3.3]{wang2021}.

Before proceeding to the classification results, we provide context for our results by looking at some of the biochemical models discussed by Reed in \cite{rbgsn2017}.  
In doing so we show that input-output networks form a natural category in which homeostasis may be explored. 

\subsection{$3$-node Input-Output Networks: Biochemical Networks}
\label{SS:TNION}

There are many examples of biochemical networks in the literature.
In particular 
examples, modelers decide which substrates are important and how the 
various substrates interact.  
The network resulting 
from the detailed modeling of the production of extracellular dopamine (eDA) by 
Best \etal~\cite{best2009} and Nijhout \etal~\cite{nbr2014}.   These authors derive a 
differential equation model for this biochemical network and use the results to 
study homeostasis of eDA with respect to variation of the enzyme tyrosine 
hydroxylase (TH) and the dopamine transporters (DAT).

In another direction, relatively small biochemical network models are often derived to help analyze a particular biochemical phenomenon. We present four examples; three are discussed in 
Reed \etal~\cite{rbgsn2017} and one in Ma \etal~\cite{ma2009}. 
These examples belong to a class that we call biochemical {\em input-output networks} and will help  to interpret the mathematical results.

\begin{example}[Feedforward Excitation] \normalfont
The input-output network corresponding to feedforward excitation is in Figure \ref{fig:feedforward}.   
This motif occurs in a biochemical network when a substrate activates the enzyme that removes 
a product. The standard biochemical network diagram for this process is
shown in Figure~\ref{fig:feedforward}(a). 
Here $\mathbf{X}$, $\mathbf{Y}$, $\mathbf{Z}$ are the names of chemical substrates  and their 
concentrations are denoted by lower case $x$, $y$, 
$z$. Each straight arrow represents a flux coming 
into or going away from a substrate. The differential equations for each substrate simply 
state that the rate of change of the concentration is the sum of the arrows going towards 
the substrate minus the arrows going away (conservation of mass). The curved line indicates 
that substrate is activating an enzyme.  

\begin{figure}[!htp] 
\centering
\begin{subfigure}{0.48\textwidth}
\centering
\includegraphics[width=2.2in]{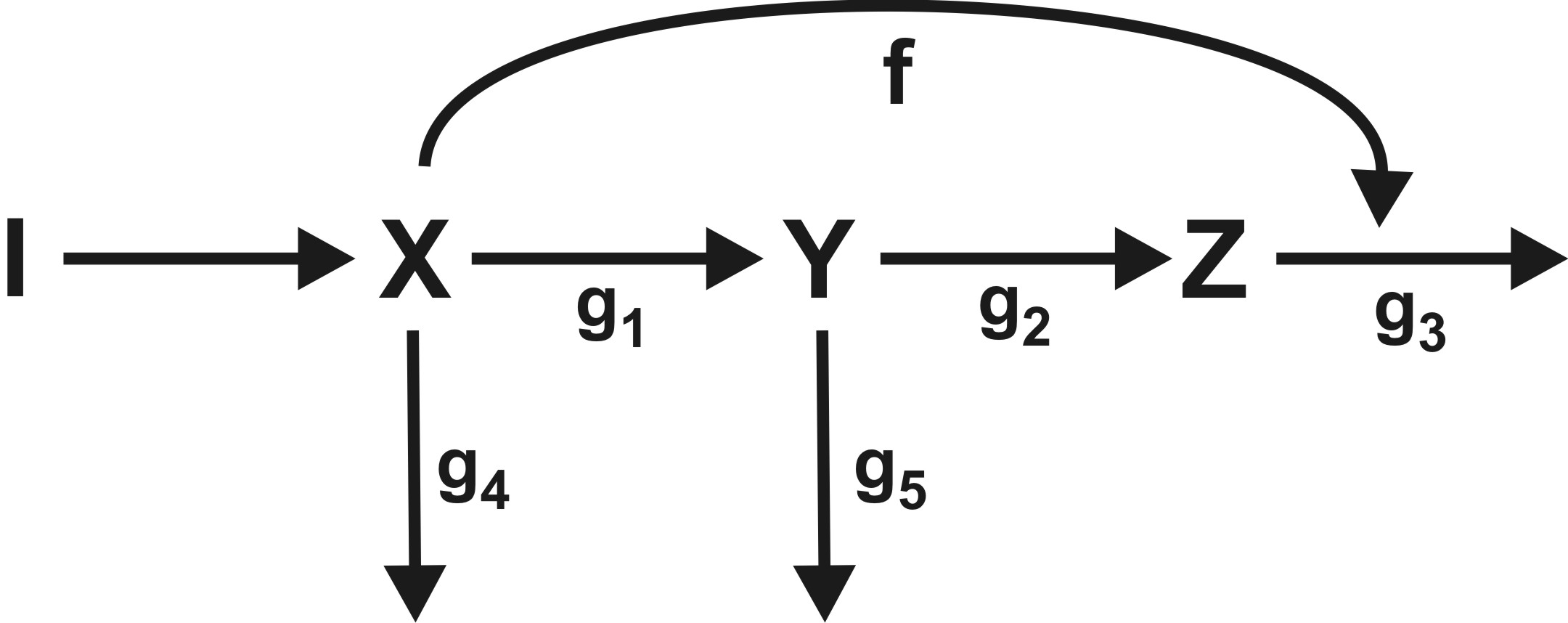}
\caption{FE motif from Reed \etal~\cite{rbgsn2017}}
\label{fig:feedforward_motif}
\vspace{2mm}
\end{subfigure}
\begin{subfigure}{0.48\textwidth}
\centering
\includegraphics[width=2.2in]{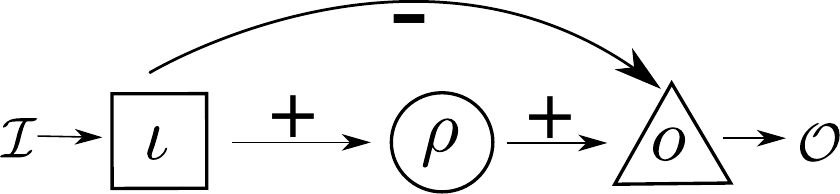}
\caption{Input-output network derived from (a)}  
\label{fig:feedforward_IO}
\end{subfigure}
\caption{Feedforward excitation (FE).}
\label{fig:feedforward}
\end{figure}

Both diagrams in Figure~\ref{fig:feedforward} represent the same information, but in different ways. 
Figure~\ref{fig:feedforward}(b) uses nodes to represent variables, and arrows to represent couplings. 
In other areas, conventions can differ, so it is necessary to translate between the two representations. 
The simplest method is to write down the model ODEs.

The equations corresponding to the biochemical network in Figure \ref{fig:feedforward} (a) are as follows
\begin{equation} \label{e:excitation}
\begin{split}
\dot{x} & = \II -g_1(x) - g_4(x) \\
\dot{y} & = g_1(x) - g_2(y) - g_5(y)\\
\dot{z} & = g_2(y) - f(x)g_3(z)
\end{split}
\end{equation}

In this motif, one path consists of two excitatory couplings: $g_1(x) > 0$ from $\bf X$ to $\bf Y$ and $g_2(y) >0$ from $\bf Y$ to $\bf Z$. 
The other path is an excitatory coupling $f(x) > 0$ from $\bf X$ to the synthesis or degradation $g_3(z)$ of $\bf Z$ and hence is an inhibitory path from $\bf X$ to $\bf Z$ having a negative sign.

It is shown in \cite{rbgsn2017} {(and reproduced using the abstract theory in \cite{wang2020})} that the model equations \eqref{e:excitation} for feedforward excitation leads to infinitesimal homeostasis at $X_0$ if 
\begin{equation} \label{e:cond_homeostase1}
f_x(x_0) = \frac{g_1'(x_0) g_2'(y_0)}{g_3(z_0)(g_2'(y_0) + g_5'(y_0))}
\end{equation}
where $X_0 = (x_0,y_0,z_0)$ is a stable equilibrium.  

Figure~\ref{fig:feedforward}(b) redraws the diagram in Figure~\ref{fig:feedforward}(a) using the abstract definition of input-output network, together with some extra features (the signs along the arrows) that are crucial to this particular application.
We consider $x$ to be a distinguished {\em input variable}, with $z$ as a distinguished {\em output variable}, while $y$ is an intermediate {\em regulatory variable}. 
Accordingly we change notation and write $x_\iota = x$, $x_\rho = y$ and $x_o = z$.

The general form of the equations associated to the diagram of Figure~\ref{fig:feedforward}(b) in the state variables $x_\iota, x_\rho, x_o$ is
\begin{equation} \label{e:feedforward}
\begin{split}
\dot{x}_\iota & = f_\iota(x_\iota,\II) \\
\dot{x}_\rho & = f_\rho(x_\iota,x_\rho) \\
\dot{x}_o & = f_o(x_\iota, x_\rho, x_o)
\end{split}
\end{equation}
In Figure~\ref{fig:feedforward}(b)
these variables are associated with three nodes $\iota, \rho, o$. 
Each node has its own symbol: a square for $\iota$, circle for $\rho$, and triangle for $o$.
Here these symbols are convenient ways to show which type of variable (input, regulatory, output) the node corresponds to. 
Arrows indicate that the variables corresponding to the tail node occur in the component of the ODE corresponding to the head node. 
For example, the component for $\dot{x}_o$ is a function of $x_\iota$, $x_\rho$, and $x_o$. 
We therefore draw an arrow from $\iota$ to $o$ and an arrow from $\rho$ to $o$. We do not draw an arrow from $o$ to itself, however: by convention, every node variable can appear in the component for that node. In a sense, the node symbol (circle) represents this `internal' arrow.

The Jacobian of \eqref{e:feedforward} is 
\[
J = \left( \begin{array}{cc;{2pt/2pt}c}
f_{\iota, x_\iota} & 0 & 0 \\ \hdashline[2pt/2pt]
f_{\rho, x_\iota} & f_{\rho,x_\rho} & 0 \\
f_{o, x_\iota} & f_{o, x_\rho} & f_{o, x_o}
\end{array} \right)
\]
The bottom-left block of $J$ is the homeostasis matrix $H$ of of \eqref{e:feedforward}.
Since $J$ is lower triangular, it follows that linear stability occurs when the linearized self-couplings are all negative.

The mathematics described here shows that infinitesimal homeostasis occurs in the system in the second column of \eqref{e:excitation} if and only if
\begin{equation} \label{e:cond_homeostase2}
\det(H) = f_{\rho, x_\iota} f_{o,x_\rho} - f_{\rho,x_\rho} f_{o,x_\iota} = 0
\end{equation}
at the stable equilibrium $X_0$.
It is easy to see that \eqref{e:cond_homeostase1}
is a particular case of \eqref{e:cond_homeostase2}.

Figure~\ref{fig:feedforward}(b) incorporates some additional information.
The arrow from $\II$ to node $\iota$ indicates that $\II$ occurs in the equation for $\dot{x}_\iota$ as the input parameter.
Similarly the arrow from node $o$ to the symbol $\mathcal{O}$ indicates that node $o$ is the output node.
Finally, the $\pm$ signs indicate which arrows are excitatory or inhibitory. 
This extra information is special to biochemical networks and does not appear as such in the general theory. 

In Figure \ref{F:KF}(a) we show the graph of the input-output function $x_o$, as a function of the input parameter $\II\in[60,140]$, for the model equations \eqref{e:excitation}.
Here, $g_1(x)=g_2(x)=g_3(x)=g_4(x)=x$, $g_5(x)=4x$ and 
\[
 f(x) = 1 + \frac{1}{1+\exp\left(\frac{50-x}{a}\right)}
\]
where $a=8.33$ and the infinitesimal homeostasis point occurs at $\II_0=100$.
In this case, it can be shown that the infinitesimal homeostasis occurring at $\II_0=100$  and $a\approx 8.33$ is in fact a chair singularity and $\sigma$ is the unfolding parameter, see \cite{rbgsn2017} for details.
\END
\end{example}

\begin{example}[Product inhibition] \normalfont
Also called \emph{feedback inhibition}, this is probably one of the simplest and best known homeostatic mechanisms in biochemistry. 
In its simplest form, feedback inhibition means that the product of a biochemical chain inhibits one or more of the enzymes involved in its
own synthesis. 
Thus if the concentration of the end product goes up, synthesis is slowed, and if the concentration goes down, the inhibition is partially withdrawn and the synthesis goes faster.
More specifically, suppose that substrate $\bf X$ influences $\bf Y$, which influences $\bf Z$,
and $\bf Z$ inhibits the flux $g_1$ from $\bf X$ to $\bf Y$.
The biochemical network for this process is shown in Figure~\ref{fig:product}(a). 

\begin{figure}[!htp] 
\centering
\begin{subfigure}{0.48\textwidth}
\centering
\includegraphics[width=2.2in]{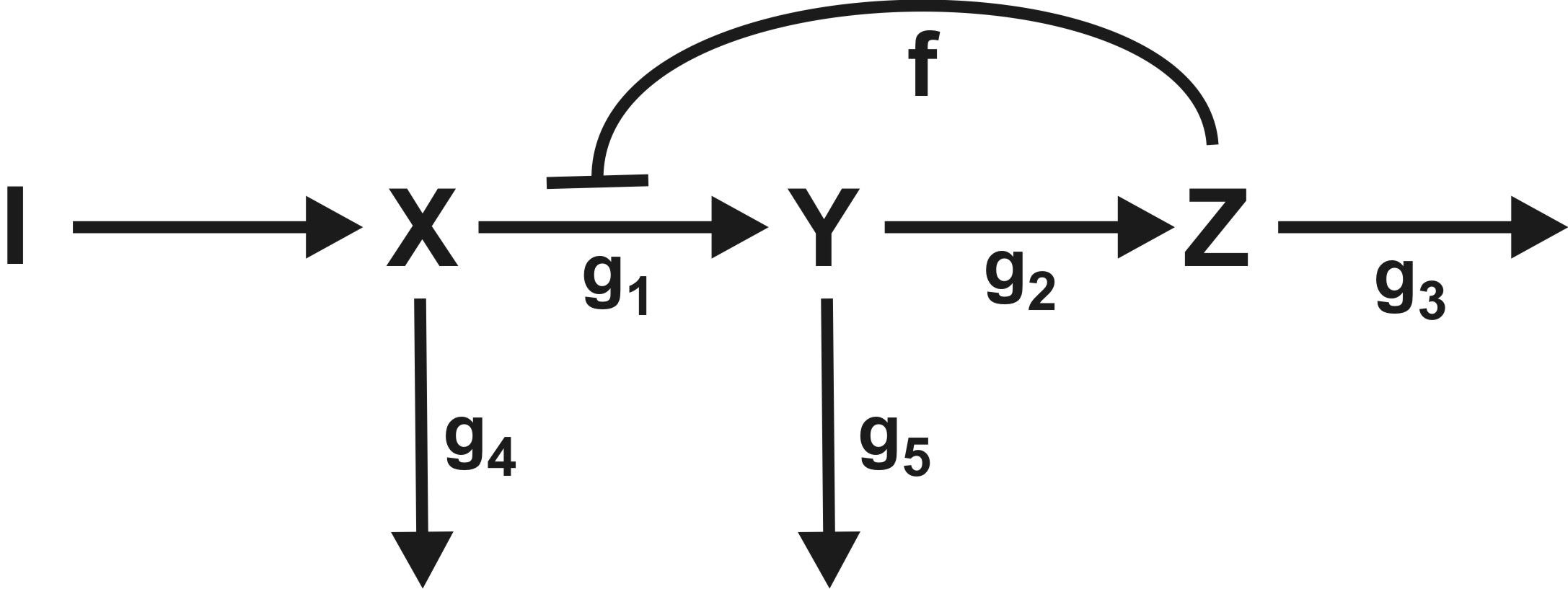}
\caption{PI motif from Reed \etal~\cite{rbgsn2017}}
\label{fig:product_motif}
\vspace{2mm}
\end{subfigure}
\begin{subfigure}{0.48\textwidth}
\centering
\includegraphics[width=2.2in]{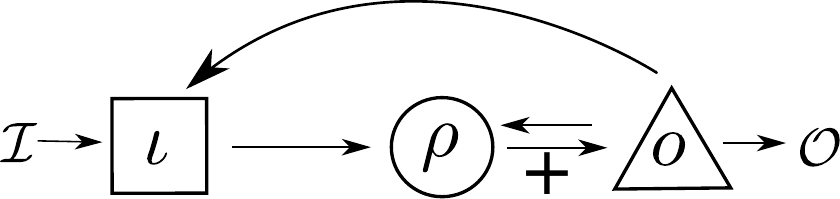}
\caption{Input-output network derived from (a)}  
\label{fig:product_IO}
\end{subfigure}
\caption{Product inhibition (PI).}
\label{fig:product}
\end{figure}

This time the model equations for Figure~\ref{fig:product}(a) are
\begin{equation} \label{e:product}
\begin{split}
\dot{x} & = \II - g_4(x) - f(z) g_1(x) \\
\dot{y} & = f(z)g_1(x) - g_2(y) - g_5(y)\\
\dot{z} & = g_2(y) - g_3(z)
\end{split}
\end{equation}
and the input-output equations associated to \eqref{e:product} can be read directly from Figure~\ref{fig:product}(b)
\begin{equation} \label{e:prod_inhibition}
\begin{split}
\dot{x}_\iota & = f_\iota(x_\iota,x_o, \II) \\
\dot{x}_\rho & = f_\rho(x_\iota,x_\rho,x_o) \\
\dot{x}_o & = f_o(x_\rho, x_o)
\end{split}
\end{equation}

Reed \etal~\cite{rbgsn2017} explain that, due to biochemical reasons, the model equation \ref{e:product} for product inhibition must satisfy some additional constraints 
\begin{equation} \label{e:prod_inhib_nd}
f > 0 \quad g_1' < 0 \quad g_2' <  0
\end{equation}

The Jacobian of \eqref{e:prod_inhibition} is 
\[
J = \left( \begin{array}{cc;{2pt/2pt}c}
f_{\iota, x_\iota} & 0 & f_{\iota, x_o} \\ \hdashline[2pt/2pt]
f_{\rho, x_\iota} & f_{\rho,x_\rho} & f_{\rho, x_o} \\
0 & f_{o, x_\rho} & f_{o, x_o}
\end{array} \right)
\]
The bottom-left block of $J$ is the homeostasis matrix $H$ of of \eqref{e:prod_inhibition}.

Our general mathematical results show that the model system \eqref{e:product} exhibits infinitesimal homeostasis at a stable equilibrium $X_0$ if and only if 
\begin{equation} \label{e:det_H_pi}
 \det(H) = f_{\rho,x_\iota} f_{o,x_\rho} = 0
\end{equation}
That is, either 
\begin{equation} \label{e:homeo_prod_inhib}
f_{\rho,x_\iota} = f(z_0)g_1'(x_0) = 0 
\quad\text{or}\quad 
f_{o,x_\rho} = g_2'(y_0) = 0
\end{equation}
It follows from \eqref{e:prod_inhib_nd} and \eqref{e:homeo_prod_inhib} that the model equation
\eqref{e:product} cannot exhibit infinitesimal homeostasis.
Nevertheless, Reed et al.~\cite{rbgsn2017} show that this biochemical network equations do exhibit {\em near-prefect homeostasis}; that is, the output $z$ is {\em almost} constant for a broad range of input values $\II$.
In the general admissible equations \ref{e:prod_inhibition} infinitesimal homeostasis 
can occur generically.
However, due to the special form of the model equation \eqref{e:product} and the additional constraints \eqref{e:prod_inhib_nd}, infinitesimal homeostasis is forced to occur at the `boundary of the universal unfolding', where $g_1'=0$ or $g_2'=0$.

In Figure \ref{F:KF}(b) we show the graph of the input-output function $x_o$, as a function of the input parameter $\II\in[0,200]$, for the model equations \eqref{e:product}.
Here, $g_1(x)=g_2(x)=g_3(x)=g_4(x)=g_5(x)=x$ and 
\[
 f(x) = 1 + \frac{b}{1+\exp\left(\frac{x-50}{a}\right)}
\]
where $a=0.5$ and $b=200$.
As explained above infinitesimal homeostasis cannot occur in this model equation.
The almost flat region occurs from $\II=50$ to $\II=110$, see \cite{rbgsn2017} for details.
\END
\end{example}

\begin{example}[Substrate Inhibition] \normalfont
The biochemical network model for substrate inhibition is given in Figure~\ref{fig:substrate}(a), and the associated model system is given in the first column of \eqref{e:substrate_inhibition}. This biochemical network and the model system are discussed in Reed \etal~\cite{rbgsn2017}.

\begin{figure}[!htp] 
\centering
\begin{subfigure}{0.48\textwidth}
\centering
\includegraphics[width=2.2in]{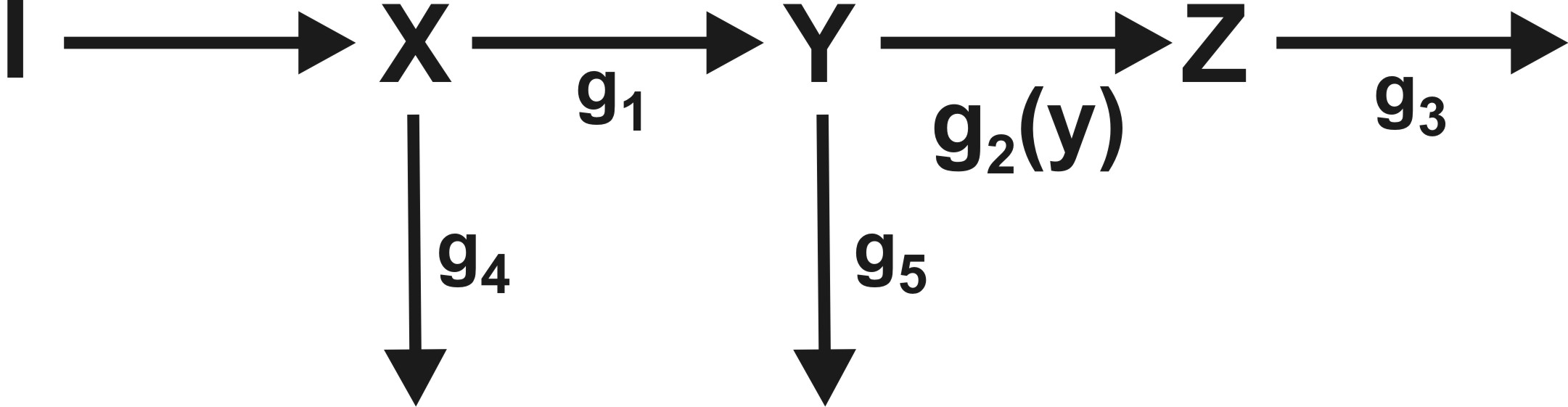}
\caption{SI motif from \cite{rbgsn2017}}
\label{fig:substrate_motif}
\vspace{3mm}
\end{subfigure}
\begin{subfigure}{0.48\textwidth}
\centering
\includegraphics[width=2.2in]{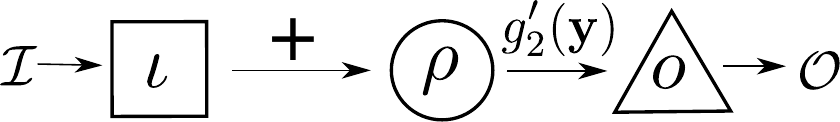}
\caption{Input-output network derived from (a)}  
\label{fig:substrate_IO}
\end{subfigure}
\caption{Substrate inhibition (SI).}
\label{fig:substrate}
\end{figure}

The equations associated to the diagram for substrate inhibition from Figure \ref{fig:substrate}(a) are
\begin{equation} \label{e:substrate_inhibition}
\begin{split}
\dot{x} & = \II - g_1(x) - g_4(x) \\
\dot{y} & = g_1(x) - g_5(y) - g_2(y)\\
\dot{z} & = g_2(y) - g_3(z) 
\end{split}
\end{equation}
Reed \etal~\cite{rbgsn2017} provides a biochemical justification for taking $g_1' > 0$ for $x > 0$, whereas the coupling (or kinetics term) $g_2'$ can change sign.

That model system of ODEs can be easily translated to the input-output system
\begin{equation} \label{e:subs_inhibition}
\begin{split}
\dot{x}_\iota & = f_\iota(x_\iota,\II) \\
\dot{x}_\rho & = f_\rho(x_\iota,x_\rho) \\
\dot{x}_o & = f_o(x_\rho, x_o)
\end{split}
\end{equation}

The Jacobian of \eqref{e:subs_inhibition} is 
\[
J = \left( \begin{array}{cc;{2pt/2pt}c}
f_{\iota, x_\iota} & 0 & 0 \\ \hdashline[2pt/2pt]
f_{\rho, x_\iota} & f_{\rho,x_\rho} & 0 \\
0 & f_{o, x_\rho} & f_{o, x_o}
\end{array} \right)
\]
The bottom-left block of $J$ is the homeostasis matrix $H$ of of \eqref{e:subs_inhibition}.
Since $J$ is lower triangular, it follows that linear stability occurs when the linearized self-couplings are all negative.

The mathematics described here shows that the condition for infinitesimal homeostasis is
\begin{equation} \label{e:det_H_si}
 \det(H) = f_{\rho,x_\iota} f_{o,x_\rho} = 0
\end{equation}
Note that this is the same condition as in the case of product inhibition \eqref{e:det_H_pi}.
Therefore, the network for product inhibition and the network for substrate inhibition are \emph{core equivalent}.
In fact, the network for substrate inhibition form Figure \ref{fig:substrate} can be obtained from the network from for product inhibition form Figure \ref{fig:product}  by removing two backward arrows: $o \to \iota$ and $o \to \rho$.
These two backward arrows correspond to two off-diagonal entries in the last column of the Jacobian matrix that are absent in the substrate inhibition.
As explained before, they do not affect the homeostasis matrix but can influence the linear stability of the system.
In fact, as noted above, the absence of these two arrows in the substrate inhibition network renders the Jacobian Matrix lower triangular and so the linear stability depends only on the linearized self-couplings of the three nodes.

On the assumption that $g_1'= f_{\rho,x_\iota} >0$, infinitesimal homeostasis is possible only if the coupling $\rho\to o$ is neutral, that is, if $f_{o,x_\rho} = g_2' = 0$ at the equilibrium point. 
This conclusion agrees with the observation in \cite{rbgsn2017} that $\bf Z$ can exhibit infinitesimal homeostasis in the substrate inhibition motif if the infinitesimal homeostasis is built into the kinetics for the function $g_2$ which controls the coupling between $\bf Y$ and $\bf Z$.  

Reed \etal~\cite{rbgsn2017} note that neutral coupling can arise from substrate inhibition of enzymes, enzymes that are inhibited by their own substrates.  
See the discussion in \cite{rln2010}.  
This inhibition leads to reaction velocity curves that rise to a maximum (the coupling is excitatory) and then descend (the coupling is inhibitory) as the substrate concentration increases.  Infinitesimal homeostasis with neutral couplings arising from substrate inhibition often has important biological functions and has been estimated to occur in about $20\%$ of enzymes \cite{rln2010}.
In the 1930s, Haldane \cite{haldane1930} introduced the concept of substrate inhibition in which the
substrate of the reaction itself inhibits the enzyme that catalyzes the reaction.
Golubitsky and Wang \cite{wang2020} have called homeostasis similar to the one occurring in substrate inhibition \emph{Haldane homeostasis}, since it arise from neutral coupling, that is, the linearized coupling between to nodes changes sign as the input parameter is varied.

In Figure \ref{F:KF}(c) we show the graph of the input-output function $x_o$, as a function of the input parameter $\II\in[0,350]$, for the model equations \eqref{e:substrate_inhibition}.
Here, $g_1(x)=g_2(x)=g_3(x)=g_4(x)=g_5(x)=x$ and 
\[
 f(x) = 1 + \frac{1}{1+\exp\left(\frac{50-x}{a}\right)}
\]
where $a=8$.
Here, unlike in the case of product inhibition, infinitesimal homeostasis can occur by a Haldane type mechanism.
After an initial growth from $1$ to $2$ over the $\II$-range $[0,100]$ the curve becomes flat and the system exhibits perfect homeostasis for $\II > 100$.
\END
\end{example}

\begin{example}[Negative Feedback Loop] \label{ex:NF_loop} \normalfont
Here, each enzyme $\mathbf{X}$, $\mathbf{Y}$, $\mathbf{Z}$ in the feedback loop motif, see Figure \ref{fig:degradation-homeostasis}(a) can have active and inactive forms. 
In the kinetic equations \eqref{e:NFL} the coupling from $\bf X$ to $\bf Z$ is non-neutral according to Ma \etal~\cite{ma2009}. 
Hence, in this model only null-degradation homeostasis is possible.
In the kinetic equations in the model the $\dot{y}$
equation does not depend on $y$ and homeostasis can only be perfect homeostasis. 
However, this model is a simplification
based on saturation in $y$ \cite{ma2009}. 
In the original system $\dot{y}$ does depend on $y$
and we expect standard null-degradation homeostasis to be possible in that system.
The corresponding input-output network to the negative feedback loop motif in Figure \ref{fig:degradation-homeostasis}(a) is shown in Figure \ref{fig:degradation-homeostasis}(b).

\begin{figure}[!htp] 
\centering
\begin{subfigure}{0.35\textwidth}
\centering
\includegraphics[width=1.3in]{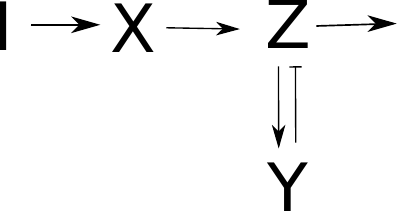}
\caption{NF motif adapted from Ma \etal~\cite{ma2009}}
\label{fig:degradation-homeostasis-motif}
\vspace{3mm}
\end{subfigure}
\begin{subfigure}{0.35\textwidth}
\centering
\includegraphics[width=1.8in]{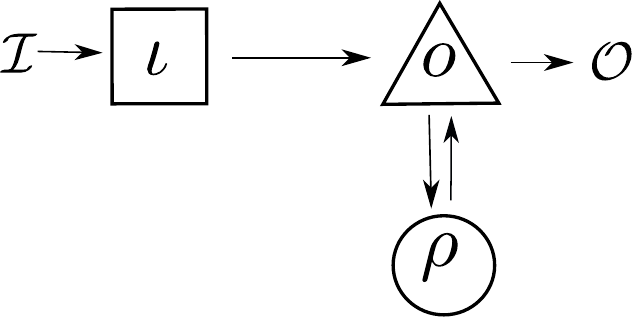}
\caption{Input-output network derived from (a)} 
\label{fig:degradation-homeostasis_IO}
\end{subfigure}
\caption{Negative feedback loop (NF)  Unlike the arrows in Figures \ref{fig:feedforward}, \ref{fig:substrate} and \ref{fig:product} that represent mass transfer between substrates, in this negative feedback loop motif positive or negative arrows between enzymes indicate the activation or inactivation of an enzyme by a different enzyme.}
\label{fig:degradation-homeostasis}
\end{figure}

The equations associated to the negative feedback loop diagram of Figure \ref{fig:degradation-homeostasis}(a) are
\begin{equation} \label{e:NFL}
\begin{split}
\dot{x} & = \II k_{\II x} \frac{1-x}{(1-x) + K_{\II x}} - F_xk'_{F_x}\frac{x}{x + K'_{F_x}}\\
\dot{y} & = z k_{zy} - F_yk'_{F_y}\\
\dot{z} & = x k_{xz}\frac{1-z}{(1-z)+K_{xz}} - yk'_{yz}\frac{z}{z+K'_{yz}}
\end{split}
\end{equation}
where $k_{\II x}$, $K_{\II x}$, $F_x$, $k'_{F_x}$, $K'_{F_x}$, $k_{zy}$, $F_y$, $k'_{F_y}$, $k_{xz}$, $K_{xz}$, $k'_{yz}$, $K'_{yz}$ are 12 constants.
That model system of ODEs is easily translated to the input-output system
\begin{equation} \label{e:NFL_IO}
\begin{split}
\dot{x}_\iota & = f_\iota(x_\iota, \II) \\
\dot{x}_\rho & = f_\rho( x_\rho, x_o)  \\
\dot{x}_o & = f_o(x_\iota,x_\rho, x_o) 
\end{split}
\end{equation}

The Jacobian of \eqref{e:NFL_IO} is 
\begin{equation} \label{e:J_NFL_IO}
J = \left( \begin{array}{cc;{2pt/2pt}c}
f_{\iota, x_\iota} & 0 & 0 \\ \hdashline[2pt/2pt]
0 & f_{\rho,x_\rho} & f_{\rho, x_o} \\
f_{o, x_\iota} & f_{o, x_\rho} & f_{o, x_o}
\end{array} \right)
\end{equation}
The bottom-left block of $J$ is the homeostasis matrix $H$ of of \eqref{e:NFL_IO}.

The mathematics described here shows that the model system \eqref{e:NFL_IO} exhibits infinitesimal homeostasis if and only 
\begin{equation}
\det(H) = f_{o,x_\iota} f_{\rho,x_\rho} = 0
\end{equation}
That is, either 
\begin{equation} \label{e:homeo_null_deg}
f_{\rho,x_\iota} = 0 
\quad\text{or}\quad 
f_{\rho,x_\rho} = 0
\end{equation}
The first case is the Haldane homeostasis type discussed before.
The second case is called \emph{null-degradation}, since it arises when the degradation constant (i.e., the linearized self-coupling) of the
regulatory node changes sign when the input parameter is varied.

Stability of the equilibrium in this motif implies negative feedback between $\rho$ and $o$.  
From the Jacobian matrix \ref{e:J_NFL_IO} of \eqref{e:NFL_IO} it follows that, at null-degradation homeostasis ($f_{\rho, x_\rho} = 0$), linear stability implies that 
\begin{equation}\label{e:i-o-stablecond-2}
f_{\iota,x_\iota} < 0 \;,
\qquad  f_{o,x_o} < 0 \;, 
\qquad  f_{\rho,x_o} f_{o, x_\rho} < 0
\end{equation}
Conditions \eqref{e:i-o-stablecond-2} imply that both the input node and the output node need to degrade and the couplings $\rho\to o$ and $o\to \rho$ must have opposite signs. 
This observation agrees with \cite{ma2009} that homeostasis is possible in the network motif Figure \ref{fig:degradation-homeostasis}(a) if there is a negative loop between $\bf Y$ and $\bf Z$ and when the linearized internal dynamics of $\bf Y$ is zero. 
Therefore, the negative feedback is `forced' by the condition for null-degradation and stability.

In Ferrell \cite{ferrell2016} the author
reviews some motifs that are capable of displaying perfect or near-perfect homeostasis and presents a simplification of biochemical example of Ma \etal~\cite{ma2009}, that exhibits null-degradation homeostasis.
In Figure \ref{F:KF}(d) we show the graph of the input-output function $x_o$, as a function of the input parameter $\II\in[0,100]$, for the model equations of \cite[Fig 2]{ferrell2016}.
After an almost instantaneous jump from $0$ to $1$ the curve becomes flat and the system exhibits, in fact perfect homeostasis for $\II > 2$.
\END
\end{example}

\begin{figure*}[!htb]
\begin{subfigure}[b]{0.45\textwidth}
\centering
\includegraphics[width=\textwidth,trim=0.25cm 1cm 1.5cm 2.5cm, clip=true]{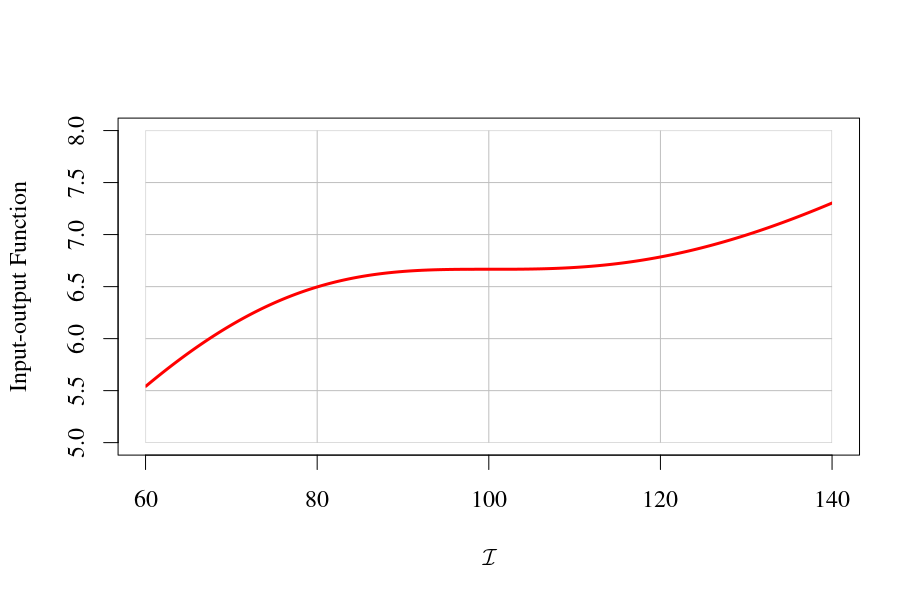}
\caption{FE motif from Reed \etal~\cite{rbgsn2017}}
\label{F:FE}
\end{subfigure} 
\centering
\begin{subfigure}[b]{0.45\textwidth}
\centering
\includegraphics[width=\textwidth,trim=0.25cm 1cm 1.5cm 2.5cm, clip=true]{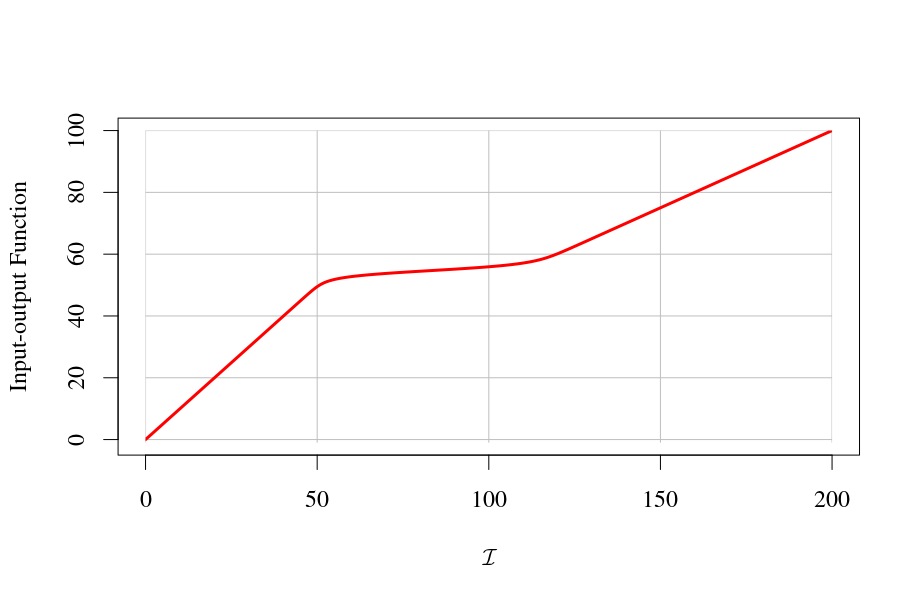}
\caption{PI motif from Reed \etal~\cite{rbgsn2017}}
\label{F:PI}
\end{subfigure}
\centering
\begin{subfigure}[b]{0.45\textwidth}
\centering
\includegraphics[width=\textwidth,trim=0.25cm 1cm 1.5cm 2.5cm, clip=true]{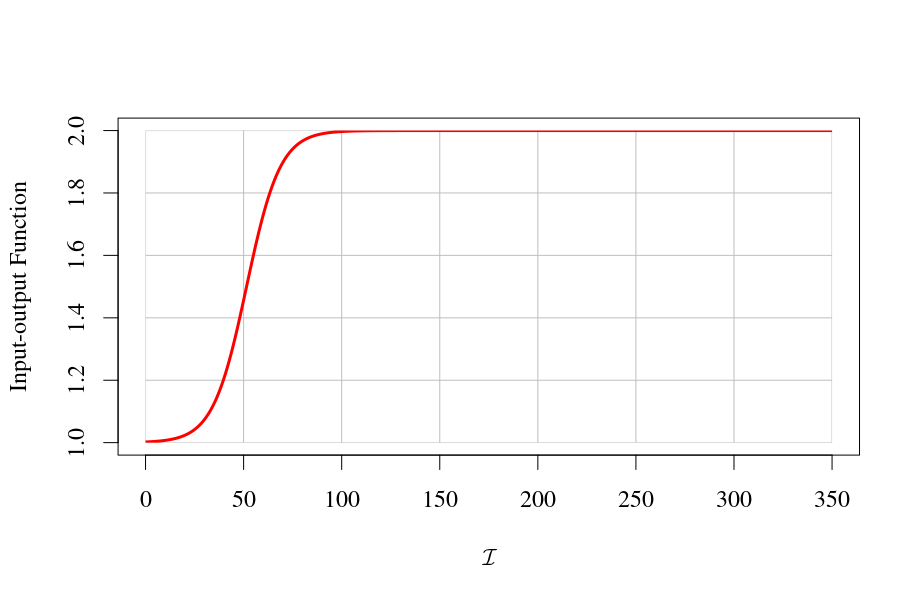}
\caption{SI motif from Reed \etal~\cite{rbgsn2017}}
\label{F:SI}
\end{subfigure}
\centering
\begin{subfigure}[b]{0.45\textwidth}
\centering
\includegraphics[width=\textwidth,trim=0.25cm 1cm 1.5cm 2.5cm, clip=true]{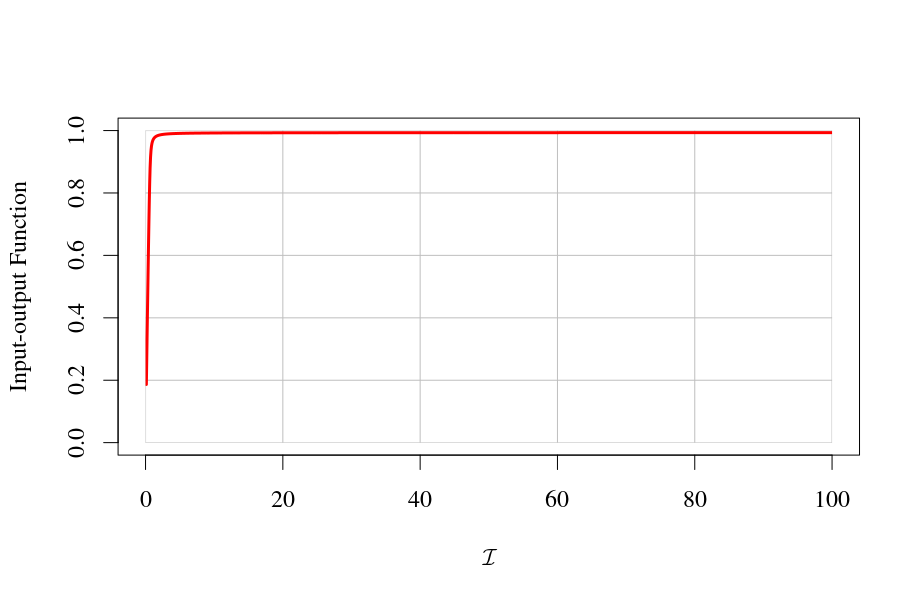}
\caption{NF motif from Ferrell~\cite{ferrell2016}}
\label{F:NF}
\end{subfigure}
\caption{\label{F:KF}
Input-output functions of the four $3$-node biochemical networks.
Input-output functions were computed by numerical continuation of an equilibrium using \textsc{Auto} from \textsc{XPPAut} \cite{bard2002}.}
\end{figure*}

As we have seen in these four examples of $3$-node biochemical networks there are plenty of possibilities of homeostatic behaviors already in very small networks.
Golubitsky and Wang \cite{wang2020} started a systematic investigation of the $3$-node input-output networks in an attempt to organize the understanding of the biochemical examples.
Their success was the motivation for Wang \etal~\cite{wang2021} to undertake the general case of an $N+2$-node input-output network and eventually led to a general theory.
In hindsight, their theorem for classification of the $3$-node input-output networks can be simply stated as

\begin{theorem}[{\cite{wang2020,wang2021}}]
\label{thm:class_3_node_core_i_o}
There are $3$ core equivalence clas\-ses of $3$-node input-output core networks with input node $\neq$ output node.
Representatives are given by: (i) the Feedforward Excitation motif, (ii) the Substrate Inhibition motif and (iii) the Negative Feedback motif.
Any other $3$-node input-output network with two distinguished nodes can be obtained from these $3$ representatives by adding or removing backward arrows.
\end{theorem}

As we shall see later, if we allow input node $=$ output node then there are two more core equivalence classes, thus giving a complete classification of all $3$-node input-output networks (see subsection \ref{SS:METAL}).

As one goes to four-node input-output networks the situation is much more complicated.
Huang and Golubitsky \cite{huang2022} employs the theory developed in Wang \etal~\cite{wang2021} to show that there are $20$ core equivalence classes of four-node input-output core networks with input node $\neq$ output node.

\subsection{Homeostasis Patterns in Input-Output Networks}

Now we consider the notion of a homeostasis pattern on a given input-output network $\mathcal{G}$.  
A {\em homeostasis pattern} is the set of nodes $j$ in $\mathcal{G}$ (including the output node $o$) such that the node coordinate $x_j$, as a function of $\II$, satisfies $x'_j(\II_0) = 0$.  
In other words, a homeostasis pattern is a set of nodes $S$ of $\mathcal{G}$, that includes the output node $o$, and all nodes in $S$ are simultaneously (infinitesimally) homeostatic at a given parameter value $\II_0$.  
 
However, before going into the details of this theory, we use such calculations to give some indication of why the input-output network in Figure \ref{F:example8} has exactly the $4$ homeostasis patterns exhibited in Figure \ref{F:admissible}. 

Consider for example the $6$ node input-output network $\mathcal{G}$ shown in Figure \ref{F:example8}.

\begin{figure}[!htb]
\centering
\includegraphics[width = .3\textwidth]{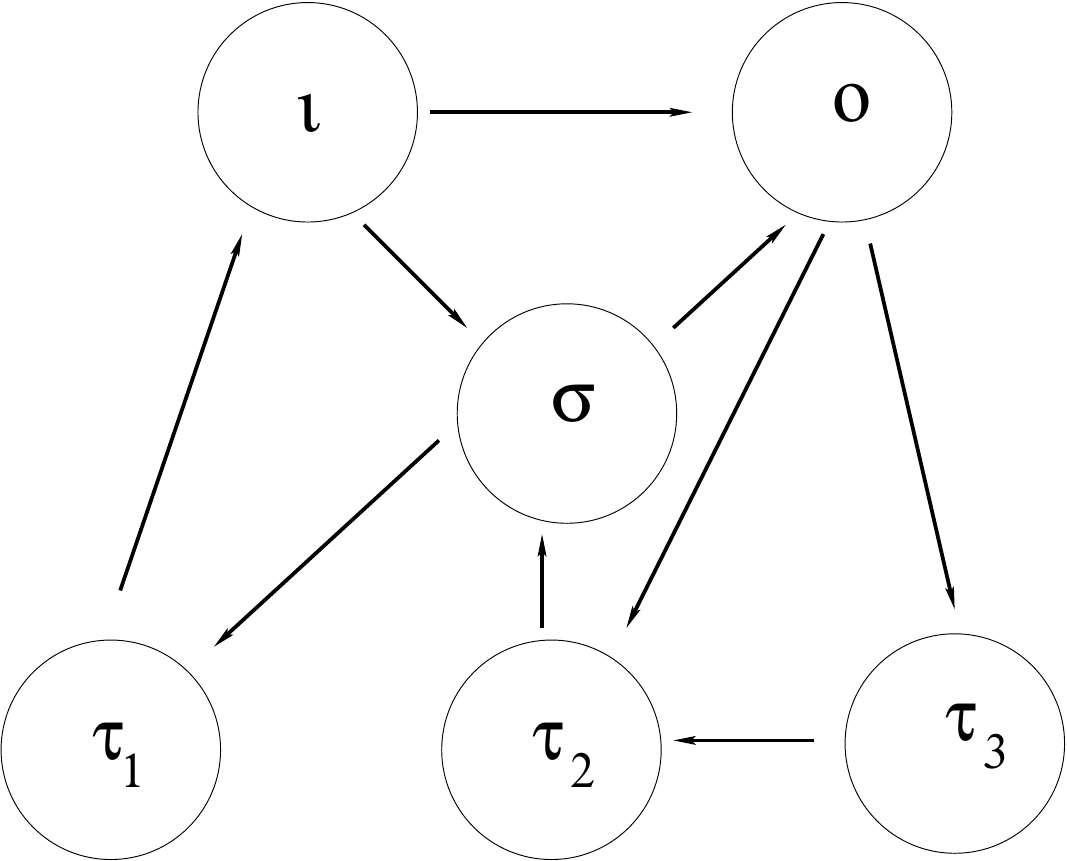}
\caption{A $6$-node input-output network. \label{F:example8}}
\end{figure}

Although there are exactly $31$ subsets of nodes of $\mathcal{G}$ including the output node $o$, only $4$ subsets define homeostasis patterns: $\{o\}$, $\{o, \tau_3\}$, $\{o, \tau_2, \tau_3\}$ and
$\{o, \tau_2, \tau_3, \sigma, \iota\}$.  These homeostasis patterns can be graphically represented by coloring the nodes of $\mathcal{G}$ that are homeostatic (see Figure \ref{F:admissible}).

\begin{figure*}[!htb]
\begin{center}
\begin{subfigure}[c]{0.2\textwidth}
\centering
\includegraphics[width=0.85\textwidth]{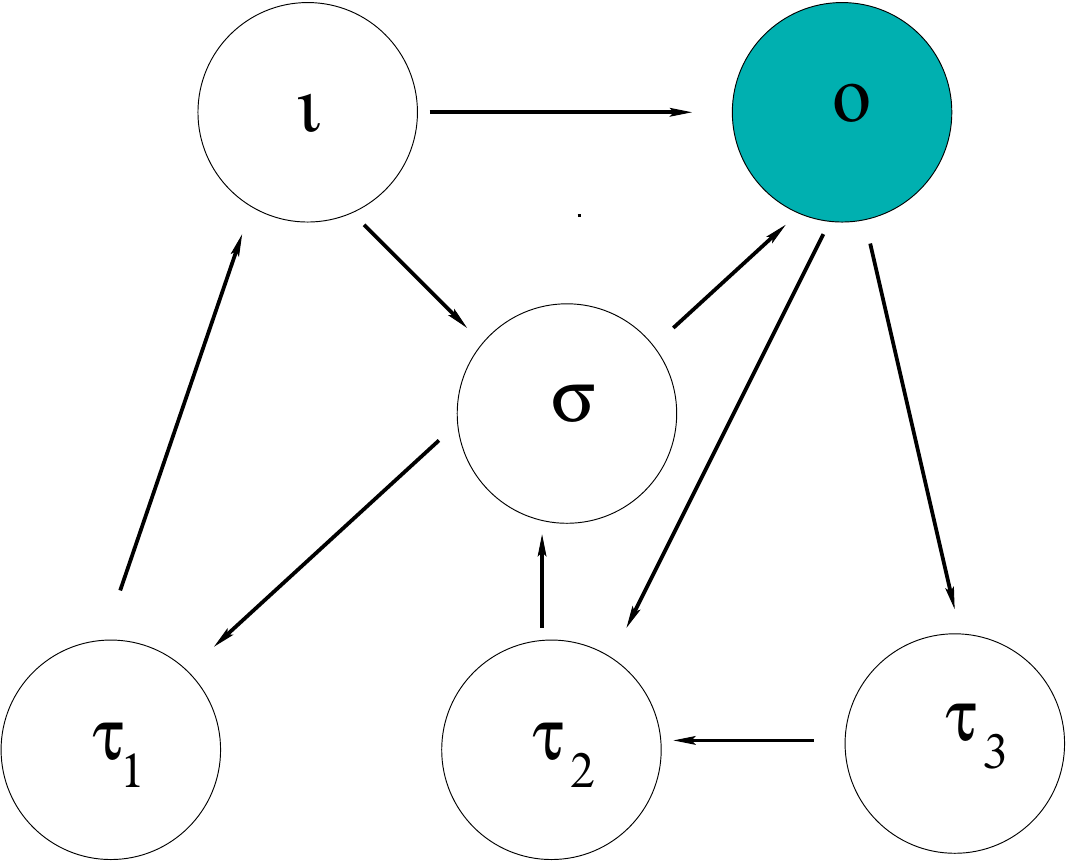}
\caption{$f_{\tau_3,\tau_3}$}
\end{subfigure}
\begin{subfigure}[c]{0.2\textwidth}
\centering
\includegraphics[width=0.85\textwidth]{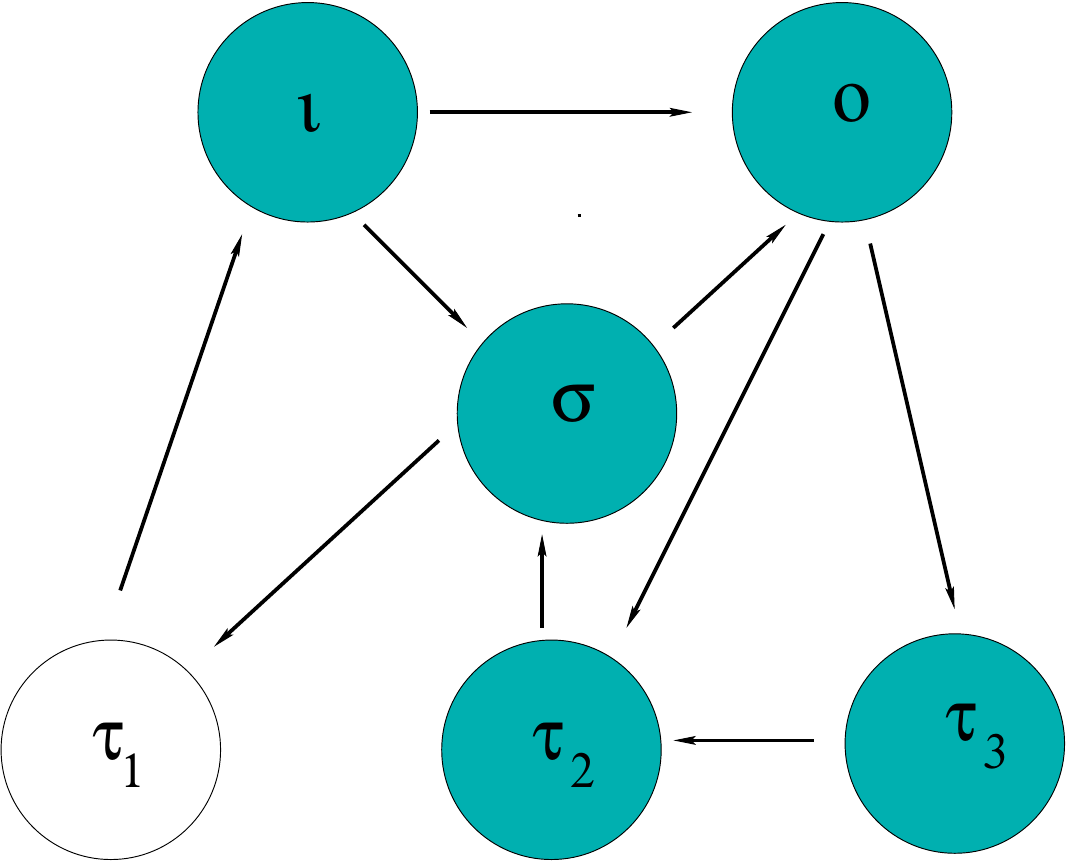}
\caption{$f_{\tau_1,\tau_1}$}
\end{subfigure}
\begin{subfigure}[c]{0.2\textwidth}
\centering
\includegraphics[width=0.85\textwidth]{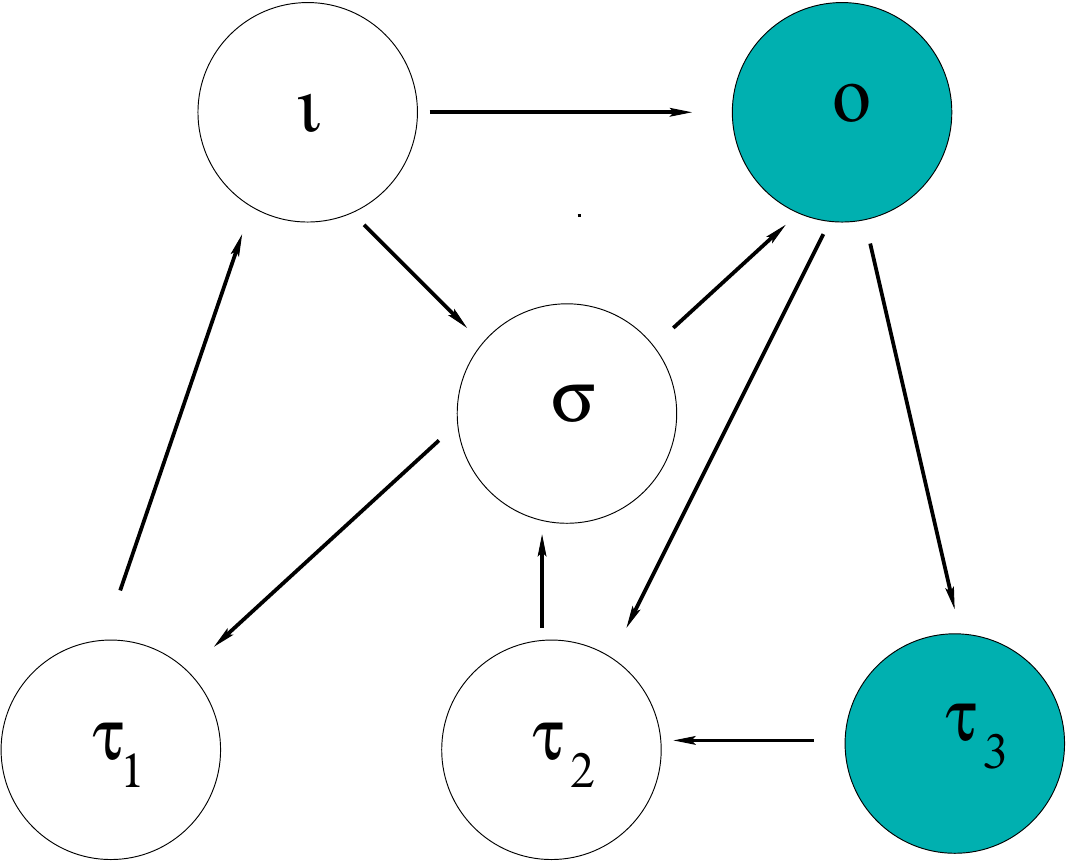}
\caption{$f_{\tau_2,\tau_2}$}
\end{subfigure}
\begin{subfigure}[c]{0.25\textwidth}
\centering
\includegraphics[width=0.67\textwidth]{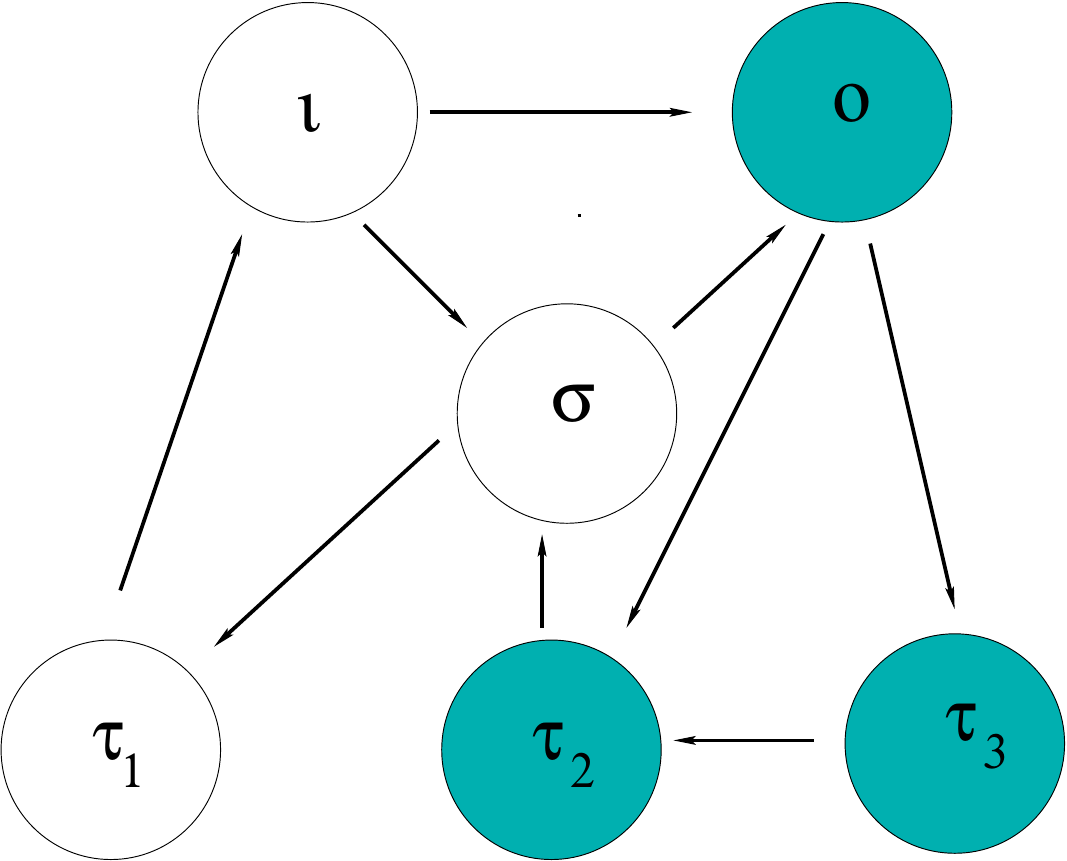}
\caption{$(f_{\sigma,\iota} f_{o,\sigma} - f_{\sigma,\sigma} f_{o,\iota})$}
\end{subfigure}
\caption{The four infinitesimal homeostasis patterns of \eqref{e:admissible}.  Cyan nodes are homeostatic. \label{F:admissible}}
\end{center}
\end{figure*}

The admissible system of parametrized equations for the network in Figure~\ref{F:example8}, in coordinates  $X = (\iota, \sigma, \tau_1,\tau_2, \tau_3, o)$, is:
\begin{equation} \label{e:admissible}
\begin{array}{ccl}
\dot{\iota} & = & f_\iota(\iota, \tau_1,\II) \\
\dot{\sigma} & = & f_\sigma(\iota, \sigma, \tau_2) \\
\dot{\tau}_1 & = & f_{\tau_1}(\sigma,\tau_1) \\
\dot{\tau}_2 & = & f_{\tau_2}(\tau_2,\tau_3,o) \\
\dot{\tau}_3 & = & f_{\tau_3}(\tau_3,o) \\
\dot{o} & = & f_o(\iota, \sigma,o)
\end{array}
\end{equation}

The homeostasis matrix $H$ is obtained from the Jacobian matrix $J$ of \eqref{e:admissible} and is given by
\[
H = \begin{pmatrix}
f_{\sigma,\iota} & f_{\sigma,\sigma} & 0 & f_{\sigma,\tau_2} & 0\\  
 0 & f_{\tau_1, \sigma} &  f_{\tau_1, \tau_1} & 0 & 0 & \\ 
 0 & 0 & 0 & f_{\tau_2,\tau_2} & f_{\tau_2,\tau_3} \\  
0 & 0 & 0 & 0& f_{\tau_3,\tau_3} \\
f_{o,\iota} & f_{o,\sigma} & 0 & 0 & 0 
\end{pmatrix}
\]
Using row and column expansion it is straightforward to calculate
\begin{equation}  \label{H:example8}
\det(H) = f_{\tau_3,\tau_3} \, f_{\tau_1,\tau_1} \, f_{\tau_2,\tau_2} \,
(f_{\sigma,\iota} f_{o,\sigma} - f_{\sigma,\sigma} f_{o,\iota})  
\end{equation}
Theorem \ref{lem:irreducible}says that the input-output function $x_o(\II)$ undergoes infinitesimal homeostasis at $\II_0$ if and only if $\det(H) = 0$, evaluated at $(X(\II_0),\II_0)$.  
Here $X(\II_0)$ is the equilibrium used to construct the input-output function.  
The expression $\det(H)$ is a multivariate polynomial in the partial derivatives $f_{j,\ell}$ of the components of the admissible vector field.  
As a polynomial, $\det(H)$ is reducible with $4$ irreducible factors, so $\det(H)=0$ if and only if one of its irreducible factors vanishes.

According to~\cite{wang2021} these irreducible factors determine the homeostasis types (see section  \ref{SEC:CLASS}).
Hence, \eqref{e:admissible} has $4$ homeostasis types.  One of the main results of this paper says that each homeostasis type determines a unique homeostasis pattern.
Moreover, our theory gives a purely combinatorial procedure to find the set of nodes that belong to each homeostasis pattern.  
When applied to Figure \ref{F:example8} it yields the four patterns in Figure \ref{F:admissible}.

In a simple example such as Figure \ref{F:example8} it is also possible to find the sets of nodes in each homeostasis pattern by a `bare hands' calculation based on the admissible ODEs.
Here this calculation serves as a check on the results.  Direct calculations with the ODEs can, of course, be used instead of the combinatorial approach in sufficiently simple cases.

Finally, we determine the four homeostasis patterns by direct calculation.  We do this by  assuming that there is a one-parameter 
family of stable equilibria $X(\II)$ where $X(\II_0) = X_0$  using implicit differentiation with respect to $\II$
(indicated by $'$), and expanding \eqref{e:admissible} to first order at  $\II_0$.  The linearized system of equations is
\begin{equation}  \label{e:1st_order}
\begin{array}{rcl}
0 & = & f_{\iota,\iota}\iota' + f_{\iota,\tau_1}\tau_1' + f_{\iota,\II} \\
0 & = & f_{\sigma,\iota}\iota' + f_{\sigma,\sigma}\sigma' + f_{\sigma,\tau_2}\tau_2' \\
0 & = & f_{\tau_1,\sigma}\sigma' + f_{\tau_1,\tau_1}\tau_1' \\
0 & = & f_{\tau_2,\tau_2}\tau_2' + f_{\tau_2,\tau_3}\tau_3' \\
0 & = & f_{\tau_3,\tau_3}\tau_3'\\
0 & = & f_{o,\iota}\iota' + f_{o,\sigma}\sigma'
\end{array}
\end{equation}
Next we compute the homeostasis patterns corresponding to the $4$ homeostasis types of \eqref{e:admissible}.

\medskip

\noindent
\textbf{(a) Homesotasis Type: $f_{\tau_3,\tau_3}$; 
Homeostatic nodes: $\{o\}$}  \\
Equation~\eqref{e:1st_order} becomes
\begin{equation}  \label{e:1st_orderA}
\begin{array}{rcl}
0 & = & f_{\iota,\iota}\iota' + f_{\iota,\tau_1}\tau_1' + f_{\iota,\II} \\
0 & = & f_{\sigma,\iota}\iota' + f_{\sigma,\sigma}\sigma' + f_{\sigma,\tau_2}\tau_2' \\
0 & = & f_{\tau_1,\sigma}\sigma' + f_{\tau_1,\tau_1}\tau_1' \\
0 & = & f_{\tau_2,\tau_2}\tau_2' + f_{\tau_2,\tau_3}\tau_3' \\
0 & = & 0 \\
0 & = & f_{o,\iota}\iota' + f_{o,\sigma}\sigma'
\end{array}
\end{equation}
Since $f_{\tau_2,\tau_2}$ and $f_{\tau_2,\tau_3}$ are generically nonzero at homeostasis, the 
fourth equation implies that generically $\tau_2'$ and $\tau_3'$ are nonzero.
The second and sixth equations can be rewritten as
\[
\begin{pmatrix}
f_{\sigma,\iota} & f_{\sigma,\sigma} \\  f_{o,\iota} & f_{o,\sigma}
\end{pmatrix}
\begin{pmatrix}
\iota'  \\ \sigma'
\end{pmatrix}
= -\begin{pmatrix}
f_{\sigma,\tau_2}\tau_2' \\ 0
\end{pmatrix}
\]
Generically the right hand side of this matrix equation at $\II_0$ is nonzero; hence generically $\iota'$ and $\sigma'$ are also nonzero.  The third equation implies that generically $\tau_1'$ is nonzero.  Therefore, in this case, the only homeostatic node is $o$.

\medskip

\noindent
\textbf{(b) Homeostasis Type: $f_{\tau_1,\tau_1}$; 
Homeostatic nodes: \\ $\{ \iota, \tau_2, \tau_3, \sigma, o\}$} \\
In this case \eqref{e:1st_order} becomes
\begin{equation}  \label{e:1st_orderB}
\begin{array}{rcl}
0 & = & f_{\iota,\iota}\iota' + f_{\iota,\tau_1}\tau_1' + f_{\iota,\II} \\
0 & = & f_{\sigma,\iota}\iota' + f_{\sigma,\sigma}\sigma' + f_{\sigma,\tau_2}\tau_2' \\
0 & = & f_{\tau_1,\sigma}\sigma'  \\
0 & = & f_{\tau_2,\tau_2}\tau_2' + f_{\tau_2,\tau_3}\tau_3' \\
0 & = & f_{\tau_3,\tau_3}\tau_3'\\
0 & = & f_{o,\iota}\iota' + f_{o,\sigma}\sigma'
\end{array}
\end{equation}
The fifth equation implies that generically $\tau_3' = 0$.  The fourth equation implies that generically $\tau_2'=0$. The third equation implies that generically $\sigma' = 0$ and the sixth equation implies that generically $\iota'$ is zero.
 It follows that the infinitesimal homeostasis pattern is  $\iota' =\tau_2' = \tau_3' = \sigma'  = o' = 0$.

\medskip

\noindent
\textbf{(c) Homeostasis Type: $f_{\tau_2,\tau_2}$; 
Homeostatic nodes: $\{\tau_3,o\}$}  \\
Equation~\eqref{e:1st_order} becomes
\begin{equation}  \label{e:1st_orderC}
\begin{array}{rcl}
0 & = & f_{\iota,\iota}\iota' + f_{\iota,\tau_1}\tau_1' + f_{\iota,\II} \\
0 & = & f_{\sigma,\iota}\iota' + f_{\sigma,\sigma}\sigma' + f_{\sigma,\tau_2}\tau_2' \\
0 & = & f_{\tau_1,\sigma}\sigma' + f_{\tau_1,\tau_1}\tau_1' \\
0 & = &  f_{\tau_2,\tau_3}\tau_3' \\
0 & = & f_{\tau_3,\tau_3}\tau_3'\\
0 & = & f_{o,\iota}\iota' + f_{o,\sigma}\sigma'
\end{array}
\end{equation}
The fourth or fifth equation implies that $\tau_3' = 0$. The first and sixth equations imply that $\iota'$, $\sigma'$, and $\tau_2'$  are nonzero. The third equation implies that generically $\tau_1'$ is nonzero.  Hence the 
infinitesimal homeostasis pattern is $\{\tau_3,o\}$.

\medskip

\noindent
\textbf{(d) Homeostasis Type: $f_{\sigma,\iota} f_{o,\sigma} - f_{\sigma,\sigma} f_{o,\iota} $; Homeostatic nodes $\{\tau_2, \tau_3, o\}$} \\
To repeat, Equation \eqref{e:1st_order} is
\begin{equation}  \label{e:1st_orderD}
\begin{array}{rcl}
0 & = & f_{\iota,\iota}\iota' + f_{\iota,\tau_1}\tau_1' + f_{\iota,\II} \\
0 & = & f_{\sigma,\iota}\iota' + f_{\sigma,\sigma}\sigma' + f_{\sigma,\tau_2}\tau_2' \\
0 & = & f_{\tau_1,\sigma}\sigma' + f_{\tau_1,\tau_1}\tau_1' \\
0 & = & f_{\tau_2,\tau_2}\tau_2' + f_{\tau_2,\tau_3}\tau_3' \\
0 & = & f_{\tau_3,\tau_3}\tau_3'\\
0 & = & f_{o,\iota}\iota' + f_{o,\sigma}\sigma'
\end{array}
\end{equation}
The fifth equation implies generically that $\tau_3' = 0$ and the fourth equation implies generically that $\tau_2' = 0$. Again, the second and sixth equations can be rewritten in matrix form as
\[
\begin{pmatrix}
f_{\sigma,\iota} &  f_{\sigma,\sigma} \\  f_{o,\iota} & f_{o,\sigma}
\end{pmatrix}
\begin{pmatrix}
\iota'  \\ \sigma'
\end{pmatrix}
= -\begin{pmatrix}
f_{\sigma,\tau_2}\tau_2' \\ 0
\end{pmatrix} = 0
\]
Hence generically $\iota'$ and $\sigma'$ are nonzero. The third equation implies that $\tau_1$ is nonzero.  Hence the homeostatic nodes are $\tau_2,\tau_3,o$.

The homeostasis types of \eqref{e:admissible} with the corresponding homeostasis patterns are summarized in Table~\ref{T:example8}.

\begin{table}[htb]
\caption{Infinitesimal homeostasis patterns for admissible systems in \eqref{e:admissible}.}
\label{T:example8} 
\begin{tabular}{clc}
\toprule
Homeostasis Type & Homeostasis Pattern &  Figure~\ref{F:admissible} \\
\midrule
$f_{\tau_3,\tau_3}$ & $\{o\}$ & (a)  \\
$f_{\tau_1,\tau_1}$ & $\{\iota,\tau_2,\tau_3,\sigma,o\}$ & (b) \\
$f_{\tau_2,\tau_2}$ & $\{\tau_3,o\}$ & (c) \\
$f_{\sigma,\iota} f_{o,\sigma} - f_{\sigma,\sigma} f_{o,\iota}$ & $\{\tau_2,\tau_3,o\}$ & (d) \\
\bottomrule
\end{tabular}
\end{table}

In principle the homeostasis patterns can be computed in the manner shown above; but in practice this becomes very complicated.

\section{Classification}
\label{SEC:CLASS}

In this section we describe the theory of Wang \etal~\cite{wang2021} and Duncan \etal~\cite{duncan2024} about the classification of homeostasis types and homeostasis patterns in an input-output network, respectively. 

\subsection{Homeostasis Subnetworks}

In \cite{wang2021} the authors show that the determination of infinitesimal homeostasis in an input-output networks reduces to the study of core networks. 
We assume throughout that the input-output networks are core networks.  

The first main result of \cite{wang2021} is the observation that one can apply the Frobenius-K\"onig \cite{schneider1977,brualdi1991} theorem to the homeostasis matrix $H$ and obtain a decomposition of $\det(H)$ into its irreducible factors.

\begin{theorem}[{\cite[Theorem 1.11]{wang2021}}] \label{thm:irreducible}
Assume \eqref{eq:system} has a hyperbolic equilibrium at $(X^*,\II^*)$. 
Then there are permutation matrices $P$ and $Q$ such that $PHQ$ is block upper triangular with square diagonal blocks $B_1,\ldots,B_m$. 
The blocks $B_j$ are irreducible in the sense that each $B_j$ cannot be further block triangularized. It follows that 
\begin{align}\label{eq:irreducible}
\det(H) = \det(B_1)\cdots\det(B_m)
\end{align}
is an irreducible factorization of $\det(H)$. 
\end{theorem}

\begin{definition} \rm
Let $\mathcal{G}$ be an input-output network and $H$ its homeostasis matrix. Each irreducible square block $B_{\eta}$ in \eqref{eq:irreducible} is called a {\em homeostasis block}.  
Further, we say that infinitesimal homeostasis in $\mathcal{G}$ is of homeostasis type $B_{\eta}$ if for all $\xi \neq \eta$, $\det(B_{\eta}) = 0$ and $\det(B_{\xi}) \neq 0$. \END
\end{definition}

\begin{remark} \rm 
Let $B_\eta$ be a homeostasis type and let
\[
h_\eta(\II) \equiv \det B_\eta(X(\II),\II)
\]
A chair point of type 
$\eta$ occurs at $\II_0$ if $h_\eta(\II_0) = h_\eta '(\II_0) = 0$ and $h_\eta '' (\II_0) \neq 0$. \END
\end{remark}

In principle every homeostasis type can lead to infinitesimal homeostasis, that is, $h_\eta(\II_0) = 0$ for some input 
value $\II_0$. 
For simplicity, we say that node $x_o$ is \emph{homeostatic at $\II_0$}.

Let $B_\eta$ be a $k\times k$ diagonal block of $H$, hence $\det(B_\eta)$ has degree $k$.  
Since the entries of $B_\eta$ are entries of $H$, these entries have the form $f_{\rho,x_\tau}$; that is, the entries are either $0$ (if $\tau\to\rho$ is not an arrow in $\mathcal{G}$), {\em self-coupling} (if $\tau = \rho$), or {\em coupling} (if $\tau\to\rho$ is an arrow in $\mathcal{G}$).  

Since $P$ and $Q$ in Theorem \ref{thm:irreducible} are constant permutation matrices, all entries in each row (resp.~column) of $B_\eta$ must lie in a single row (resp.~column) of $H$.  
Hence, $B_\eta$ has the form
\begin{equation} \label{eq:K}
B_\eta = \begin{pmatrix}
f_{\rho_1,x_{\tau_1}} & \cdots & f_{\rho_1,x_{\tau_k}} \\ 
	\vdots & \ddots & \vdots \\
	f_{\rho_k,x_{\tau_1}} & \cdots & f_{\rho_k,x_{\tau_k}}
\end{pmatrix}
\end{equation}
It follows that the number of self-coupling entries of $B_\eta$ are the same no matter which permutation matrices $P$ and $Q$ are used to bring $H$ to the block-diagonal form.

\begin{theorem}[{\cite[Theorem 4.4]{wang2021}}]
Each $k\times k$ submatrix $B_\eta$ has either $k$ or $k-1$ self-coupling entries.
\end{theorem}

\begin{definition} \label{D:classes}
The homeostasis class of an irreducible component $B_\eta$ of degree $k$ is \emph{appendage} if $B_\eta$ has $k$ self-couplings and \emph{structural} if $B_\eta$ has $k-1$ self-couplings.
\END
\end{definition}

Now we can associate a \emph{homeostasis subnetwork} $\mathcal{K}_\eta\subset\mathcal{G}$ with each homeostasis block $B_\eta$ and give a graph-theoretic description of each $\mathcal{K}_\eta$.

\begin{definition} \rm \label{D:K_eta}
The \emph{homeostasis subnetwork} $\mathcal{K}_\eta$ of $\mathcal{G}$ associated with the homeostasis block $B_\eta$ is defined as follows.  
The \emph{nodes} in $\mathcal{K}_\eta$ are the union of nodes $p$ and $q$ where $f_{p,x_q}$ is a nonzero entry in $B_\eta$ and the \emph{arrows} of $\mathcal{K}_\eta$ are the union of arrows $q\to p$ where $p\neq q$. 
\END
\end{definition}

In order to give a combinatorial characterization of the structural and appendage homeostasis subnetworks we need some simple concepts from graph theory.

\begin{definition} \label{D:simple_node}
Let $\mathcal{G}$ be a core input-output network. 
\begin{enumerate}[(a)]
\item A \emph{simple path} from node $\kappa_1$ to node $\kappa_2$ in $\mathcal{G}$ is a directed path that starts at $\kappa_1$, ends at $\kappa_2$, and visits each node on the path exactly once. 
We denote the existence of a simple path from $\kappa_1$ to $\kappa_2$ by $\kappa_1 \pathto \kappa_2$. 
A \emph{simple cycle} is a simple path whose first and last nodes are identical. \item An \emph{$\iota o$-simple path} is a simple path from the input node $\iota$ to the output node $o$.
\item A node $\sigma$ is \emph{simple} if it lies on an $\iota o$-simple path. 
A node $\tau$ is \emph{appendage} if it is not simple.
\item A simple node $\rho$ is \emph{super-simple} if it lies on every $\iota o$-simple path. \END
\end{enumerate}
\end{definition}   

We typically use $\sigma$ to denote a simple node, $\rho$ to denote a super-simple node, and $\tau$ to denote an appendage node when the type of the node is assumed \textit{a priori}. Otherwise, we use $\kappa$ to denote an arbitrary node. Note that $\iota$ and $o$ are super-simple nodes.  

Let $\rho_0, \rho_1, \ldots, \rho_q, \rho_{q+1}$ be the super-simple nodes, where $\rho_0 = \iota$ and $\rho_{q+1} = o$.
The super-simple nodes are totally ordered by the order of their appearance on any $\iota o$-simple path, and this ordering is independent of the $\iota o$-simple path.  We denote
an $\iota o$-simple path by 
\[
\iota \pathto \rho_1 \pathto \cdots \pathto \rho_q \pathto o  
\]
where $\rho_j\pathto\rho_{j+1}$ indicates a simple path from $\rho_j$ to $\rho_{j+1}$. 
The ordering of the super-simple nodes is denoted by  
\[
\rho_0 \prec \rho_1 \prec \cdots \prec \rho_q \prec \rho_{q+1},
\]
and $\prec$ is a total ordering. 
The ordering $\prec$ extends to a partial ordering of simple nodes, as follows. 
If there exists a super-simple node $\rho$ and an $\iota o$-simple path such that 
\[
\iota \pathto \sigma_1 \pathto \rho \pathto \sigma_2 \pathto o
\]
then the partial orderings 
\[
\sigma_1 \prec \rho \qquad \rho \prec \sigma_2 \qquad \sigma_1 \prec \sigma_2
\]
are valid. In this partial ordering every simple node is comparable to every super-simple node 
but two simple nodes that lie between the same adjacent super-simple nodes need not be 
comparable.

We recall the definition of transitive (or strong) components of a network. Two nodes are {\em equivalent} if there is a path from one to the other and back. A {\em transitive component} is an equivalence class for this equivalence relation.

\begin{definition}\label{defn:super_appendage}
\begin{enumerate}[(a)]
\item Let $S$ be an $\iota o$-simple path.  The {\em complementary subnetwork} of $S$ is the network $C_S$ whose nodes are nodes that are not in $S$ and whose arrows are those that connect nodes in $C_S$.

\item An appendage node $\tau$ is {\em super-appendage} if for each $C_S$  containing $\tau$, the transitive component of $\tau$ in $C_S$ consists only of appendage nodes. \END
\end{enumerate}
\end{definition} 

Note that this definition of super-appendage leads to a slightly different, but equivalent, definition of homeostasis subnetwork to the one given in \cite{wang2021}. 
However, this change enables us to define pattern networks in a more straightforward way (see Remark \ref{R:appendage}).

Now we start with the definition of structural subnetworks. 

\begin{definition} \label{D:structural_subnet}
Let $1\leq  j \leq q+1$ and $\rho_{j-1}\prec\rho_j$ be two consecutive super-simple nodes.
Then the $j^{th}$ \emph{simple subnetwork} $\mathcal{L}_j''$, 
the $j^{th}$ {\em augmented simple subnetwork} $\mathcal{L}_j'$, and
the $j^{th}$ \emph{structural subnetwork} $\mathcal{L}_j$ 
are defined in four steps as follows.
\begin{enumerate}[(a)]
\item The $j^{th}$ \emph{simple subnetwork} $\mathcal{L}_j''$ consists of simple nodes $\sigma$ where 
\[
\rho_{j-1}\prec \sigma \prec \rho_j
\] 
and all arrows connecting these nodes. Note that $\mathcal{L}_j''$ does not contain the super-simple nodes $\rho_{j-1}$ and $\rho_j$, and  $\mathcal{L}_j''$ can be the empty set.

\item An appendage but not super-appendage node $\tau$ is {\em linked} to $\mathcal{L}_j''$ if for some complementary subnetwork $C_S$ the transitive component of $\tau$ in $C_S$ is the union of $\tau$, nodes in $\mathcal{L}_j''$, and non-super-appendage nodes. The set of $j^{th}$-{\em linked  appendage nodes} $T_j$ is the set of non-super-appendage nodes that are linked to $\mathcal{L}_j''$.

\item The $j^{th}$ {\em augmented simple subnetwork} $\mathcal{L}_j'$ is    
\[
\mathcal{L}_j' =   \mathcal{L}_j'' \cup T_j
\]
and all arrows connecting these nodes. 
\item The $j^{th}$ \emph{structural subnetwork} $\mathcal{L}_j$ consists of the augmented simple subnetwork $\tilde{\mathcal{L}}_j$ and adjacent super-simple nodes. That is 
\[
\mathcal{L}_j = \{\rho_{j-1}\} \cup \mathcal{L}_j' \cup \{\rho_j\}
\]
and all arrows connecting these nodes. \END
\end{enumerate}
\end{definition}

Note that an structural subnetwork $\mathcal{L}_j$ is itself an input-output network, with input node $\rho_{j-1}$ and output node $\rho_{j}$.

\ignore{
\begin{definition}
Define $\sigma_1 \preceq \sigma_2$ if either $\sigma_1 \prec \sigma_2$, $\sigma_1 = \sigma_2$  is a super-simple node, or $\sigma_1$ and $\sigma_2$ are in the same simple subnetwork. \END
\end{definition}
}

Next we define the appendage subnetworks, which were defined in Section 1.7.2 of \cite{wang2021} as any transitive component of the subnetwork consisting only of appendage nodes and the arrows between them.  

\begin{definition} \label{D:appendage_subnet}
An \emph{appendage subnetwork} $\mathcal{A}_k$, $1 \leq k \leq r$ is a transitive component of the subnetwork of super-appendage nodes. \END
\end{definition}

\begin{remark} \rm \label{R:appendage}
In \cite{wang2021} the authors define an appendage subnetwork as a transitive component of appendage nodes $\mathcal{A}$ that satisfy the {\em no cycle condition}.
The no-cycle condition is formulated in terms of the non-existence of a cycle between appendage nodes in $\mathcal{A}$ and the simple nodes in $C_S$ for all simple $\iota o$-simple paths $S$.
Here, we define an appendage subnetwork as a transitive component of super-appendage nodes, which are defined in terms of transitive components with respect to $C_S$ for all simple $\iota o$-simple paths $S$. These two definitions are equivalent because two nodes belong to the same transitive component if and only if both nodes lie on a (simple) cycle.
\END
\end{remark}

\begin{theorem}[{\cite{wang2021}}] \label{thm:homeo_sub_net}
Let $\mathcal{K_\eta}$ be a homeostasis subnetwork of $\mathcal{G}$ with corresponding diagonal block $B_\eta$.
\begin{enumerate}[(1)]
\item If $B_\eta$ is structural then $\mathcal{K_\eta}=\mathcal{L}_j$ for some $1 \leq k \leq q$. 
Moreover, $\mathcal{L}_j$ is an input-output network, hence
$B_\eta$ is an homeostasis matrix and
\[
 \det(B_\eta) = \det(H(\mathcal{L}_j))
\]
where $H(\mathcal{L}_j)$ is the homeostasis matrix of
$\mathcal{L}_j$.
\item If $B_\eta$ is appendage then $\mathcal{K_\eta}=\mathcal{A}_k$ for some $1 \leq k \leq r$. 
Moreover, $\mathcal{A}_k$ is a network without any distinguished node, hence $B_\eta$ is a jacobian matrix and
\[
 \det(B_\eta) = \det(J(\mathcal{A}_k))
\]
where $J(\mathcal{A}_j)$ is the jacobian matrix of
$\mathcal{A}_k$.
\end{enumerate}
\end{theorem}

Theorem \ref{thm:homeo_sub_net} says that the network $\mathcal{G}$ decomposes into subnetworks $\mathcal{K}_\eta$, each of which is of two types: structural or appendage.
Two of these subnetworks may have common nodes only when they are consecutive structural subnetworks.
Furthermore, the irreducible factors of $\det(H)$ can be obtained directly from the set of homeostasis subnetworks.

\subsection{Homeostasis Patterns}

\begin{definition}
A {\em homeostasis pattern} corresponding to the homeostasis block $B_\eta$ at $\II_0$ is the collection of 
all nodes, including the output node $o$, that are simultaneously forced to be homeostatic at $\mathcal{I}_0$. \END
\end{definition}

\begin{definition} \label{D:homeo_inducing}
Assume that the output node $o$ is homeostatic at $(X_0,\II_0)$, that is, $x_o'(\II_0) = 0$ for the input value $\II_0$.
\begin{enumerate}[(a)]
\item We call the homeostasis subnetwork 
$\mathcal{K}_\eta$  \emph{homeostasis inducing} if $h_{\mathcal{K}_\eta} \equiv \det(B_\eta) = 0$ at $(X_0,\II_0)$. 

\item Homeostasis of a node $\kappa \in \mathcal{G}$ is \emph{induced} by a homeostasis subnetwork $\mathcal{K}$, denoted $\mathcal{K} \Rightarrow \kappa$, if $\kappa$ is generically homeostatic whenever $\mathcal{K}$ is homeostasis inducing.

\item A homeostasis subnetwork $\mathcal{K}$ \emph{induces} a subset of nodes $\mathcal{N}$ ($\mathcal{K} \Rightarrow \mathcal{N}$), if $\mathcal{K} \Rightarrow \kappa$ for each node $\kappa \in \mathcal{N} \subset \mathcal{G}$. \END
\end{enumerate}
\end{definition}

By definition, every homeostasis subnetwork $\mathcal{K}$ induces homeostasis in the output node $o$, that is, $\mathcal{K} \Rightarrow o$.

In this subsection we construct the \emph{homeostasis pattern network} $\mathcal{P}$ associated with $\mathcal{G}$ (see
Definition~\ref{D:HPN}), which serves to organize the homeostasis subnetworks and to clarify how each homeostasis subnetwork connects to the others.

The homeostasis pattern network is defined in the following steps.
First, we define the structural pattern network $\mathcal{P}_\mathcal{S}$ in terms of the structural subnetworks of $\mathcal{G}$.
Second, we define the appendage pattern network $\mathcal{P}_\mathcal{A}$ in terms of the appendage subnetworks of $\mathcal{G}$.
Finally, we define how the nodes in $\mathcal{P}_\mathcal{S}$ connect to nodes in $\mathcal{P}_\mathcal{A}$ and conversely.

\begin{definition}
The {\em structural pattern network} $\mathcal{P}_\mathcal{S}$ is the feedforward network whose nodes are the super-simple nodes $\rho_j$ and the \emph{backbone} nodes $\widetilde{\mathcal{L}}_j$, where $\widetilde{\mathcal{L}}_j$ is  
the augmented structural subnetwork $\mathcal{L}_j'$ treated as a  single node.  
The nodes and arrows of 
$\mathcal{P}_\mathcal{S}$ are given as follows.
\begin{equation} \label{backbone_e}
\iota = \rho_0 \to \widetilde{\mathcal{L}}_1 \to \rho_1 \to \cdots  \to \widetilde{\mathcal{L}}_{q+1} \to  \rho_{q+1} = o 
\end{equation}
If a structural subnetwork ${\mathcal{L}}$ consists of an arrow between two adjacent super-simple nodes (Haldane type) then the corresponding augmented structural subnetwork ${\mathcal{L}'}$ is the empty network; nevertheless the corresponding backbone node $\widetilde{\mathcal{L}}$ must be included in the structural pattern network $\mathcal{P}_\mathcal{S}$.
\END
\end{definition}

\begin{definition}
The {\em appendage pattern network} $\mathcal{P}_\mathcal{A}$ is the network whose nodes are the components $\widetilde{\mathcal{A}}$ in the condensation of the subnetwork of super-appendage nodes. 
Such a node $\widetilde{\mathcal{A}}$ is called an {\em appendage component}.
An arrow connects nodes $\widetilde{\mathcal{A}}_1$ and $\widetilde{\mathcal{A}}_2$ if and only if there are super-appendage nodes $\tau_1 \in \widetilde{\mathcal{A}}_1$ and $\tau_2 \in \widetilde{\mathcal{A}}_2$ such that $\tau_1 \to \tau_2$ in $\mathcal{G}$. 
\END
\end{definition}

\begin{definition} \normalfont
Let $\mathcal{G}$ be a directed graph.
Consider the equivalence relation $\sim$ given by path equivalence on the set of nodes of $\mathcal{G}$.
The equivalence classes are the path components of $\mathcal{G}$.
Let $\mathcal{G}^c=\mathcal{G}/\!\!\sim$ be the quotient graph called the \emph{condensation of $\mathcal{G}$}. 
That is, the nodes of $\bar{\mathcal{G}}$ are the path components of $\mathcal{G}$ and the arrows of $\bar{\mathcal{G}}$ are defined as follows. There is a directed arrow form path component $\mathcal{A}_1$ to path component $\mathcal{A}_2$ if there are nodes $\tau_1\in\mathcal{A}_1$, $\tau_2\in\mathcal{A}_2$
and a direct arrow form $\tau_1$ to $\tau_2$.
\END
\end{definition}

To complete the homeostasis pattern network, we describe how the nodes in $\mathcal{P}_\mathcal{A}$  and the nodes in 
$\mathcal{P}_\mathcal{S}$ are 
connected. To do so, we take advantage of the feedforward ordering of the nodes in $\mathcal{P}_\mathcal{S}$ and the feedback ordering of the nodes in $\mathcal{P}_\mathcal{A}$. 

\begin{definition} 
A simple path from $\kappa_1$ to $\kappa_2$ is an \emph{appendage path} if some node on this path is an appendage node and every node on this path, except perhaps for $\kappa_1$ and $\kappa_2$, is an appendage node. \END
\end{definition}

\begin{definition}[{How $\mathcal{P}_\mathcal{A}$ connects to $\mathcal{P}_\mathcal{S}$}]
\label{D:V_max}
Given a node $\widetilde{\mathcal{A}} \in \mathcal{P}_\mathcal{A}$, we construct a unique arrow from $\widetilde{\mathcal{A}}$ to the structural pattern 
network $\mathcal{P}_\mathcal{S}$ in two steps: 
\begin{enumerate}[(a)]
\item Consider the collection of nodes $\mathcal{V}$ in $\mathcal{P}_\mathcal{S}$ for 
which there exists a simple node $\sigma \in \mathcal{V}$ and appendage node $\tau \in \widetilde{\mathcal{A}}$, such that there is an  appendage path from $\tau$ to $\sigma$.

\item Let $\mathcal{V}_{max}(\widetilde\mathcal{A})$ be a maximal node in this collection, that is, the most downstream in $\mathcal{P}_\mathcal{S}$. 
It follows from \eqref{backbone_e} that $\mathcal{V}_{max}$ is either a super-simple node $\rho_j$ or a backbone node $\widetilde{\mathcal{L}}_j$. 
Maximality implies that $\mathcal{V}_{max}$ is uniquely defined. We then say that there is an arrow from $\widetilde{\mathcal{A}}$ to $\mathcal{V}_{max} \in \mathcal{P}_\mathcal{S}$. \END
\end{enumerate}
\end{definition}
 
\begin{definition}[{How $\mathcal{P}_\mathcal{A}$ is connected from $\mathcal{P}_\mathcal{S}$}] \label{D:V_min}
Given a node $\widetilde{\mathcal{A}} \in \mathcal{P}_\mathcal{A}$ we choose uniquely an arrow from the structural pattern 
network $\mathcal{P}_\mathcal{S}$ to $\widetilde{\mathcal{A}}$ in two steps: 

\begin{enumerate}[(a)]
\item Consider the collection of nodes $\mathcal{V}$ in $\mathcal{P}_\mathcal{S}$ 
for which there exists a simple node $\sigma \in \mathcal{V}$ and appendage node $\tau \in \widetilde{\mathcal{A}}$, such that there is an  appendage path from $\sigma$ to $\tau$.

\item Let $\mathcal{V}_{min}(\widetilde{\mathcal{A}})$ be a minimal node in this collection, that is, the most upstream node in $\mathcal{P}_\mathcal{S}$. Then $\mathcal{V}_{min}$ is either a super-simple node $\rho_j$ or a backbone node $\widetilde{\mathcal{L}}_j$, and the minimality implies  uniqueness of $\mathcal{V}_{min}$. We then say that there is an arrow from $\mathcal{V}_{min} \in \mathcal{P}_\mathcal{S}$ to $\widetilde{\mathcal{A}}$. \END
\end{enumerate}
\end{definition}

Since we consider only core input-output networks, all appendage nodes are downstream from $\iota$ and upstream from $o$.
Hence, for any node $\widetilde{\mathcal{A}}\in\mathcal{P}_\mathcal{A}$, there always exist nodes $\mathcal{V}_{min}, \mathcal{V}_{max} \in \mathcal{P}_\mathcal{S}$ as mentioned above.

\begin{definition} \label{D:HPN}
The \emph{homeostasis pattern network} $\mathcal{P}$ is the network whose nodes are the union of the nodes of the structural pattern network $\mathcal{P}_\mathcal{S}$ and the appendage pattern network $\mathcal{P}_\mathcal{A}$. The arrows of $\mathcal{P}$ are the arrows of $\mathcal{P}_\mathcal{S}$, the arrows of $\mathcal{P}_\mathcal{A}$, and the arrows between $\mathcal{P}_\mathcal{S}$ and $\mathcal{P}_\mathcal{A}$ as described above. \END
\end{definition}

\begin{remark} \rm \label{R:ss_nocorrespond}
The super-simple nodes in $\mathcal{P}$  correspond to the super-simple nodes of $\mathcal{G}$. 
Each super-simple node $\rho_j \in \mathcal{G}$ (for $1\le j\le q$) belongs to exactly two structural subnetworks $\mathcal{L}_{j-1}$ and $\mathcal{L}_{j}$.  Thus they are not associated to a single homeostasis subnetwork of $\mathcal{G}$.  \END
\end{remark}

It follows from Remark~\ref{R:ss_nocorrespond} that there is a correspondence between the homeostasis subnetworks of $\mathcal{G}$ and the non-super-simple nodes of $\mathcal{P}$.

\begin{remark} \label{R:correspond} \rm $ $
\begin{enumerate}[(a)]
\item Each structural subnetwork $\mathcal{L} \subseteq \mathcal{G}$  corresponds to the backbone node $\widetilde{\mathcal{L}} \in \mathcal{P}_\mathcal{S}$. Note that the augmented structural subnetworks $\mathcal{L}' \subsetneq \mathcal{L}$ are not homeostasis subnetworks.

\item Each appendage subnetworks $\mathcal{A} \subset \mathcal{G}$ corresponds to a appendage component $\widetilde{\mathcal{A}} \in \mathcal{P}_\mathcal{A}$.

\item For simplicity in notation we let $\mathcal{V}_\mathcal{S}$ denote a node in $\mathcal{P}_\mathcal{S}$.
Further we let $\widetilde{\mathcal{V}}$ denote a non-super-simple node of $\mathcal{P}$ and $\mathcal{V}$ denote its corresponding homeostasis subnetwork.
\END
\end{enumerate}
\end{remark}

The main point of introducing the homeostasis pattern network $\mathcal{P}$ is to
relate homeostatic induction between the set of homeostasis subnetworks of $\mathcal{G}$ to induction between nodes in $\mathcal{P}$.
In Definition \ref{D:homeo_inducing_pattern} bellow  we formalize this notion.
Hence, every node in a homeostasis pattern (which can be backbone or appendage) is induced by either a backbone node or an appendage node in the homeostasis pattern network $\mathcal{P}$.

\begin{definition} \label{D:homeo_inducing_pattern}
Let $\widetilde{\mathcal{V}}_1, \widetilde{\mathcal{V}}_2 \in \mathcal{P}$ be non-super-simple nodes and $\rho\in\mathcal{P}$ be a super-simple node.
Let $\mathcal{V}_1, \mathcal{V}_2 \subset \mathcal{G}$ be the corresponding homeostasis subnetworks to $\widetilde{\mathcal{V}}_1, \widetilde{\mathcal{V}}_2 \in \mathcal{P}$.
We say that $\widetilde{\mathcal{V}}_1$ {\em induces} $\widetilde{\mathcal{V}}_2$, denoted by $\widetilde{\mathcal{V}}_1 \Rightarrow \widetilde{\mathcal{V}}_2$, if and only if $\mathcal{V}_1 \Rightarrow \mathcal{V}_2$. 
We say that $\widetilde{\mathcal{V}}_1$ {\em induces} $\rho$, denoted by $\widetilde{\mathcal{V}}_1 \Rightarrow \rho$, if and only if $\mathcal{V}_1 \Rightarrow \rho$. \END
\end{definition}

We exclude super-simple nodes of $\mathcal{P}$ from being `homeostasis inducing' because they are not associated to a homeostasis subnetwork of $\mathcal{G}$ (see Remark \ref{R:ss_nocorrespond}).
However, when a backbone node $\widetilde{\mathcal{L}}_j\in \mathcal{P}$ induces homeostasis on other nodes of $\mathcal{P}$, it is the corresponding structural subnetwork $\mathcal{L}_j$, with its two super-simple nodes $\rho_{j-1}, \rho_j$ that induce homeostasis.

As explained before, the homeostasis pattern network $\mathcal{P}$ allows us to characterize homeostasis patterns by reducing to four possibilities that are covered by Theorems \ref{thm:struct_to_struct} - \ref{thm:app_to_app}.

Structural homeostasis patterns are given by the following two theorems.

\begin{theorem}[{Structural $\Rightarrow$ Structural}] \label{thm:struct_to_struct}
A backbone node $\widetilde{\mathcal{L}}_j \in \mathcal{P}_\mathcal{S}$ induces
every node of the structural pattern network $\mathcal{P}_S$ strictly downstream from $\widetilde{\mathcal{L}}_j$, but no other nodes of $\mathcal{P}_S$.
\end{theorem}

See Figure \ref{fig:struct_to_struct} for an application of Theorem~\ref{thm:struct_to_struct}. 

\begin{figure}[!h]
\centering
\includegraphics[width = .45\textwidth]{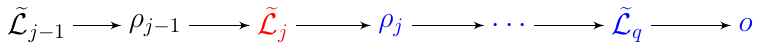}
\caption{An example of structural network induced by structural homeostasis. 
Suppose the backbone node $\widetilde{\mathcal{L}}_j$ in red is homeostasis inducing.  Then Theorem \ref{thm:struct_to_struct} 
implies that the blue nodes in the structural pattern network are all homeostatic.
}
\label{fig:struct_to_struct}
\end{figure}

\begin{theorem}[Structural $\Rightarrow$ Appendage] \label{thm:struct_to_app}
A backbone node $\widetilde{\mathcal{L}}_j \in \mathcal{P}_\mathcal{S}$ induces every
appendage component of $\mathcal{P}_\mathcal{A}$ whose $\mathcal{V}_{min}$ (see Definition~\ref{D:V_min}) is strictly downstream, but no other nodes of $\mathcal{P}_\mathcal{A}$.
\end{theorem}

See Figure  \ref{fig:struct_to_app} for an application of Theorem~\ref{thm:struct_to_app}.  

\begin{figure}[!h]
\centering
\includegraphics[width = .2\textwidth]{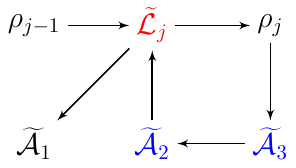}
\caption{An example of appendage subnetworks induced by structural homeostasis.  
Suppose the backbone node $\widetilde{\mathcal{L}}_j$ in red is homeostasis inducing. Then by Theorem \ref{thm:struct_to_app} and the fact that the super-simple node $\rho_j$ is strictly downstream from $\widetilde{\mathcal{L}}_j$ the blue appendage components downstream from $\rho_j$ are homeostatic.}
\label{fig:struct_to_app}
\end{figure}

The appendage homeostasis patterns are characterized by the following two theorems.

\begin{theorem}[Appendage $\Rightarrow$ Structural] \label{thm:app_to_struct}
An appendage component $\widetilde{\mathcal{A}} \in \mathcal{P}_\mathcal{A}$ induces every super-simple node of $\mathcal{P}_S$ downstream 
from $\mathcal{V}_{max}(\widetilde{\mathcal{A}})$ (see Definition~\ref{D:V_max}), but no other super-simple nodes. 
Further, an appendage component $\widetilde{\mathcal{A}} \in \mathcal{P}_\mathcal{A}$ induces a backbone node $\widetilde{\mathcal{L}}_j$ if and only if 
 $\widetilde{\mathcal{L}}_j$ is strictly downstream from $\mathcal{V}_{max}(\widetilde{\mathcal{A}})$.
\end{theorem}

See Figure \ref{fig:app_to_struct} for an application of Theorem~\ref{thm:struct_to_struct}. 

\begin{figure}[!h]
\centering
\includegraphics[width = .35\textwidth]{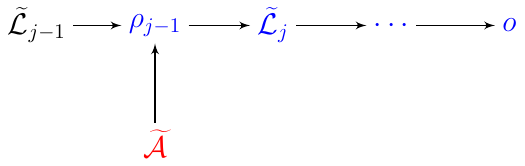}
\caption{An example of structural subnetworks induced by appendage homeostasis.
Suppose the appendage component $\widetilde{\mathcal{A}}$ in red is homeostasis inducing. 
Since $\widetilde{\mathcal{A}}$ connects to the super-simple node $\rho_{j-1}$, then by Theorem \ref{thm:app_to_struct} the blue nodes in the structural pattern network are homeostatic.}
\label{fig:app_to_struct}
\end{figure}

\begin{theorem}[Appendage $\Rightarrow$ Appendage] \label{thm:app_to_app}
An appendage component $\widetilde{\mathcal{A}}_i \in \mathcal{P}_\mathcal{A}$ induces an appendage component $\widetilde{\mathcal{A}}_j \in \mathcal{P}_\mathcal{A}$ if and only if 
$\widetilde{\mathcal{A}}_i$ is strictly upstream from $\widetilde{\mathcal{A}}_j$ and every path from $\widetilde{\mathcal{A}}_i$ to $\widetilde{\mathcal{A}}_j$ in $\mathcal{P}$ 
 contains a super-simple node $\rho$ satisfying $\mathcal{V}_{max} (\widetilde{\mathcal{A}}_i) \preceq \rho \preceq \mathcal{V}_{min} (\widetilde{\mathcal{A}}_j)$.   
\end{theorem}

See Figure \ref{fig:app_to_app} for an application of Theorem~\ref{thm:app_to_app}. 

\begin{figure}[!h] 
\centering
\includegraphics[width=.35\textwidth]{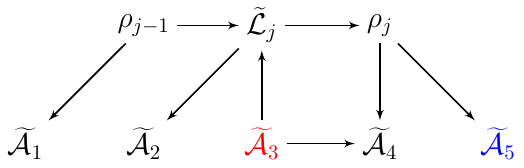}
\caption{
An example of appendage subnetworks induced by appendage homeostasis. 
Suppose the appendage component $\widetilde{\mathcal{A}}_3$ in red is homeostasis inducing. 
Since $\widetilde{\mathcal{A}}_3$ has only one path to the blue appendage component $\widetilde{\mathcal{A}}_5$ containing the super-simple node $\rho_j$, then by Theorem \ref{thm:app_to_app} $\widetilde{\mathcal{A}}_5$ is homeostatic,  but no other appendage subnetwork is homeostatic.}
\label{fig:app_to_app}
\end{figure}

\begin{example} \label{E:example8_pattern} \normalfont
We consider homeostasis patterns for the admissible systems in \eqref{e:admissible} obtained from Figure~\ref{F:example8}.  Specifically, we show how the theorems in this section lead to the determination of the homeostasis patterns that were derived by direct calculation from the equations \eqref{e:1st_order}. 
The corresponding homeostasis pattern network $\mathcal{P}$ is shown in Figure \ref{F:example8b}.
The answer is listed in Table~\ref{T:example8}.

\medskip

\noindent
\textbf{Case (a)}  $f_{\tau_3, \tau_3}$:
homeostasis is induced by the node $\widetilde{\mathcal{A}}_3$ of $\mathcal{P}_{\mathcal{A}}$.
Theorem \ref{thm:app_to_struct} shows $\widetilde{\mathcal{A}}_3$ induces $\{ o \}$, which is the only super-simple node of $\mathcal{P}_\mathcal{S}$ downstream from $\mathcal{V}_{max}(\widetilde{\mathcal{A}}_3)$.
And $\widetilde{\mathcal{A}}_3$ induces no backbone node.
Theorem \ref{thm:app_to_app} shows  $\widetilde{\mathcal{A}}_3$ induces no appendage component.
Therefore, in this case the homeostasis pattern is $\{o\}$.

\medskip

\noindent
\textbf{Case (b)}  $f_{\tau_1, \tau_1}$:
homeostasis is induced by the node $\widetilde{\mathcal{A}}_1$ of $\mathcal{P}_{\mathcal{A}}$.
Theorem \ref{thm:app_to_struct} shows $\widetilde{\mathcal{A}}_1$ induces $\{ \iota, o \}$, which are the super-simple nodes of $\mathcal{P}_\mathcal{S}$ downstream from $\mathcal{V}_{max}(\widetilde{\mathcal{A}}_1)$.
Also $\widetilde{\mathcal{A}}_1$ induces $\{ \widetilde{\mathcal{L}}_1 \}$, which is the backbone node of $\mathcal{P}_\mathcal{S}$ downstream from $\iota$.
Theorem \ref{thm:app_to_app} shows $\widetilde{\mathcal{A}}_1$ induces $\{ \widetilde{\mathcal{A}}_2, \widetilde{\mathcal{A}}_3 \}$, which are the nodes of $\mathcal{P}_\mathcal{A}$ downstream from $\widetilde{\mathcal{A}}_1$ and each path contains the super-simple node $o$ with $\mathcal{V}_{max}(\widetilde{\mathcal{A}}_1) \preceq o \preceq \mathcal{V}_{min}(\widetilde{\mathcal{A}}_2)$ or $\mathcal{V}_{min}(\widetilde{\mathcal{A}}_3)$. Therefore, in this case the homeostasis pattern is $\{ \iota, \ \sigma, \ o, \ \tau_2, \ \tau_3 \}$.

\medskip

\noindent
\textbf{Case (c)}  $f_{\tau_2, \tau_2}$:
homeostasis is induced by the node $\widetilde{\mathcal{A}}_2$ of $\mathcal{P}_{\mathcal{A}}$.
Theorem \ref{thm:app_to_struct} shows $\widetilde{\mathcal{A}}_2$ induces $\{ o \}$, which is the super-simple nodes of $\mathcal{P}_\mathcal{S}$ downstream from $\mathcal{V}_{max}(\widetilde{\mathcal{A}}_2)$.
And $\widetilde{\mathcal{A}}_2$ induces no backbone node.
Theorem \ref{thm:app_to_app} shows $\widetilde{\mathcal{A}}_2$ induces $\{ \widetilde{\mathcal{A}}_3 \}$, which is the node of $\mathcal{P}_\mathcal{A}$ downstream from $\widetilde{\mathcal{A}}_2$ and each path contains the super-simple node $o$ with $\mathcal{V}_{max}(\widetilde{\mathcal{A}}_2) \preceq o \preceq \mathcal{V}_{min}(\widetilde{\mathcal{A}}_3)$.
Therefore, in this case the homeostasis pattern is $\{ o, \ \tau_3 \}$.

\medskip

\noindent
\textbf{Case (d):}  
$\det \left(\begin{smallmatrix}
f_{\sigma, \iota} & f_{\sigma, \sigma} \\
f_{o, \iota} & f_{o, \sigma}
\end{smallmatrix}\right)$:
homeostasis is induced by node $\widetilde{\mathcal{L}}_1$ of $\mathcal{P}_{\mathcal{S}}$.
Theorem \ref{thm:struct_to_struct} shows $\widetilde{\mathcal{L}}_1$ induces $\{ o \}$, which is the node of $\mathcal{P}_\mathcal{S}$ downstream from $\widetilde{\mathcal{L}}_1$.
Theorem \ref{thm:struct_to_app} shows $\widetilde{\mathcal{L}}_1$ induces $\{ \widetilde{\mathcal{A}}_2, \widetilde{\mathcal{A}}_3 \}$, whose $\mathcal{V}_{min}$ are strictly downstream from $\widetilde{\mathcal{L}}_1$.
Therefore, in this case the homeostasis pattern is $\{ o, \ \tau_2, \ \tau_3 \}$.
\END
\end{example}

\begin{figure}[!htb]
\centering
\includegraphics[width =.35\textwidth]{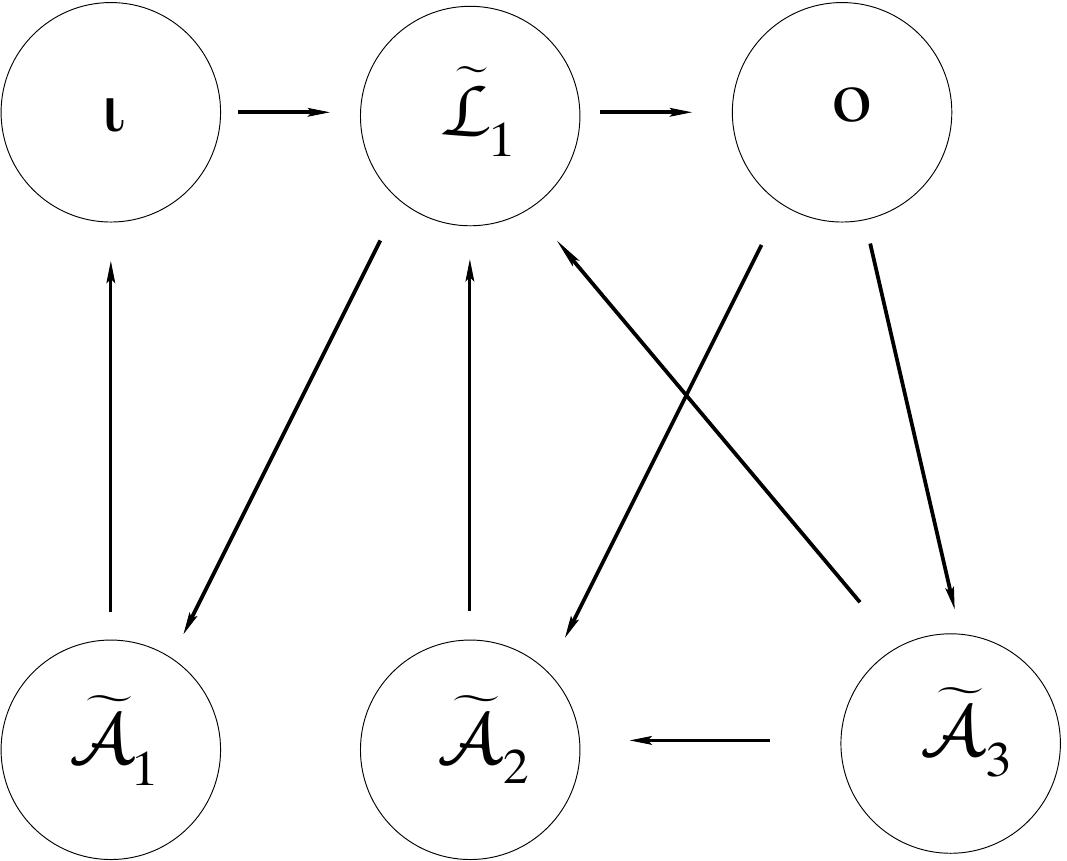}
\caption{Pattern network $\mathcal{P}$ obtained from input-output network $\mathcal{G}$ shown in Figure~\ref{F:example8}. \label{F:example8b}}
\end{figure}

Finally, we give three general results about the induction relation.
First, in Theorem~\ref{thm:pattern_from_subnetworks_intro} we state that the induction relation is characterized by its behavior on homeostasis subnetworks.
Next, in Theorems~\ref{thm:always_related_intro} we state that induction applies in at least one direction for distinct homeostasis subnetworks and that no subnetwork induces itself.
Finally, in  Theorem~\ref{thm:bijective_corresp_intro} we state that distinct subnetworks have distinct homeostasis patterns. 
That is, the set of nodes induced by a homeostasis subnetwork is unique among all homeostasis subnetworks.
For proofs see \cite{duncan2024}.

\begin{theorem} \label{thm:pattern_from_subnetworks_intro}
Suppose $\mathcal{K}_1$ and $\mathcal{K}_2$ are distinct homeostasis subnetworks of $\mathcal{G}$.  Let $\kappa$ be a node of 
$\mathcal{K}_2$ where $\kappa\neq o$.  If $\mathcal{K}_1\Rightarrow \kappa$, then $\mathcal{K}_1 \Rightarrow \mathcal{K}_2$.
\end{theorem}

\begin{theorem} \label{thm:always_related_intro}
Let $\mathcal{K}_1$ be a homeostasis subnetwork of $\mathcal{G}$. Then generically $\mathcal{K}_1 \not\Rightarrow \mathcal{K}_1$.
Moreover, let $\mathcal{K}_2$ be some other homeostasis subnetwork of $\mathcal{G}$.
Then one of the following relations holds:
\begin{enumerate}[(a)]
\item $\mathcal{K}_1 \Rightarrow \mathcal{K}_2$ and $\mathcal{K}_2 \not\Rightarrow \mathcal{K}_1$,

\item $\mathcal{K}_2 \Rightarrow \mathcal{K}_1$ and $\mathcal{K}_1 \not\Rightarrow \mathcal{K}_2$,

\item $\mathcal{K}_1 \Rightarrow \mathcal{K}_2$ and $\mathcal{K}_2 \Rightarrow \mathcal{K}_1$.
\end{enumerate}
\end{theorem}

Theorems~\ref{thm:pattern_from_subnetworks_intro} and \ref{thm:always_related_intro} imply
that each homeostasis subnetwork has a unique homeostasis pattern associated to it.  Specifically:

\begin{theorem} \label{thm:bijective_corresp_intro}
Let $\mathcal{K}_1$ and $\mathcal{K}_2$ be two distinct homeostasis subnetworks of $\mathcal{G}$. 
Then the set of subnetworks induced by $\mathcal{K}_1$ and the set of subnetworks induced by $\mathcal{K}_2$ are distinct.
\end{theorem}

Theorem \ref{thm:bijective_corresp_intro} implies that there is a bijective correspondence between  homeostasis subnetworks of $\mathcal{G}$ and homeostasis patterns supported by $\mathcal{G}$.
Recall that theorem \ref{thm:homeo_sub_net} implies that there is a bijective correspondence between the homeostasis types of $G$, that is, the
irreducible factors of $\det(H)$ (or the diagonal blocks of $H$) and homeostasis subnetworks of $\mathcal{G}$.
Therefore we can summarize the main result of our theory as bijective correspondence between three sets defined for any input-output network $\mathcal{G}$:
\[
\begin{split}
\{\text{Homeostasis types}\} 
& \simeq \{\text{Homeostasis subnetworks}\} \\
& \simeq \{\text{Homeostasis patterns}\}
\end{split}
\]

\section{Applications}
\label{SEC:APPL}

In this section we consider four examples of applications of the theory of \cite{wang2021,duncan2024} presented in the previous section.
The first application is to gene regulatory networks, following the approach in \cite{antoneli2018, antoneli2023}.
The second application is to intracellular metal ion regulation, following the approach of \cite{antoneli2024}.
The third application is to bacterial chemotaxis. 
Here, the theory of \cite{wang2021} needs to be adapted to the case of multiple input nodes (with one input parameter) and this was accomplished in \cite{madeira2022}.
In the last application we consider a generalization of the notion of homeostasis to the periodic case and apply this to the modeling of period homeostasis in the circadian rhythm.

\subsection{Genetic Regulatory Networks}
\label{SS:GEGRN}

{\em Gene expression} is the process by which the information encoded in a gene is turned into a biological function, which ultimately manifests itself as a phenotype effect.
This is accomplished by a complex series of enzymatic chemical reactions within the cell leading to the synthesis of specific macro-molecules called the {\em gene product}.
The process of gene expression is used by all known life -- eukaryotes (including multicellular organisms), prokaryotes (bacteria and archaea), and even viruses -- to generate the molecular machinery for life.

There are, basically, two types of gene products: (i) for {\em protein-coding genes} the gene product is a {\em protein}; (ii) {\em non-coding genes}, such as transfer RNA (tRNA) and small nuclear RNA (snRNA), the gene product is a functional {\em non-coding} RNA (ncRNA).

{\em Regulation of gene expression}, or simply {\em gene regulation}, is the range of mechanisms that are used by cells to increase or decrease the amount of specific gene products. 
Sophisticated schemes of gene expression are widely observed in biology, going from triggering developmental pathways, to responding to environmental stimuli.

A {\em gene} (or {\em genetic}) {\em regulatory network} (GRN) is a collection of molecular regulators that interact with each other and with other substances in the cell to govern the gene expression levels of mRNA and proteins. 
The {\em molecular regulators} can be DNA, RNA, protein or any combination of two, or more of these three that form a complex, such as a specific sequence of DNA and a transcription factor to activate that sequence. 
The interaction can be direct or indirect (through transcribed RNA or translated protein).
When a protein acts as a regulator it is called a {\em transcription factor}, which is one of the main players in regulatory networks.  
By binding to the promoter region of a coding gene they turn them on, initiating the production of another protein, and so on. 
Transcription factors can be {\em excitatory} ({\em activators}) or {\em inhibitory} ({\em repressors}).

The development of advanced experimental techniques in molecular biology is producing increasingly large amounts of experimental data on gene regulation. 
This, in turn, demands the development of mathematical modeling methods for the study and analysis of gene regulation.
Mathematical models of GRNs describe both gene expression and regulation, and in some cases generate predictions that support experimental observations.
Formally, a GRN is represented by a directed graph.
Nodes represent the variables associated to genes (e.g., mRNA and/or protein concentration) and directed links represent couplings between genes (e.g., effect of one gene product on other genes).
In this review we will focus on the mathematical modeling of GRNs using coupled systems of ordinary differential equations.
As we will see in a moment, in the coupled ODE setting there is an important issue concerning the number of variables / equations associated to each node.

Even though the notion of homeostasis is often associated with regulating global physiological parameters, such as temperature, hormone levels, or concentrations of molecules in the bloodstream in complex multicellular organisms, it also can be applied to unicellular organisms.
However, the issue here is how some internal cell state of interest (such as the concentration of some gene product) responds to changes in the intra-cellular and/or extra-cellular environment \cite{ma2009,tang2016,shi2017}.
For instance, Antoneli \etal~\cite{antoneli2018} study the occurrence of homeostasis in a \emph{feedforward loop} motif from the GRN of \emph{S.~cerevisiae} (see Figure \ref{F:reg_net_1a}).

In \cite{antoneli2018} the authors use theorem \ref{lem:irreducible} and bare hands calculations to find infinitesimal homeostasis in a small $3$-node GRN called `feedforward loop motif' (see Figure \ref{F:reg_net_1a}).
Assuming that the regulation of the three genes is inhibitory (repression) and that only gene SPF1 is regulated by upstream transcription factors (the input parameter), they show that the protein concentration of gene GAP1 (the output node) robustly exhibits infinitesimal homeostasis with respect to variation on the regulation level of SPF1 over a wide range.
Moreover, SPF1 and GZF3 (i.e., their protein concentrations) are not homeostatic for any value of the input parameter.
Here, `robustly' means that the occurrence of infinitesimal homeostasis (or not) on the protein concentrations of the three genes described above is persistent under variation of kinetic parameters of the defining differential equations (the rates of synthesis and degradation of the mRNA and protein concentrations).
In order to obtain compatibility of the infinitesimal homeostasis formalism with the GRN structure they use the protein-mRNA network (PRN) representation.
Moreover, in this particular setting, they were able to explicitly compute the homeostasis point by assuming a special functional form for the equations.

\begin{figure}[!htb]
\centerline{\includegraphics[width=0.3\textwidth]{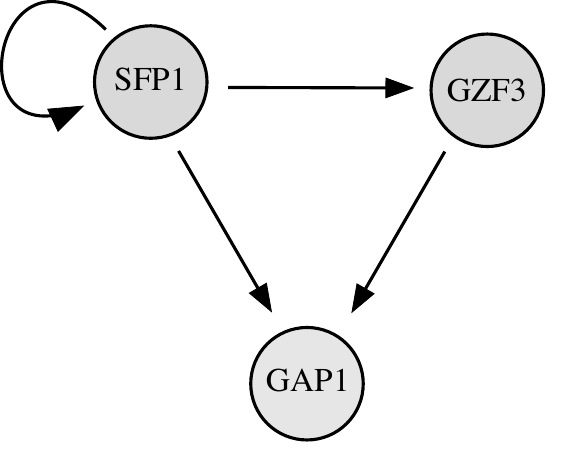}}
\caption{Feedforward loop motif from the GRN of \emph{S.~cerevisiae} (see Antoneli \etal \cite{antoneli2018}).  Note that autoregulation implies that the protein concentration associated with the gene SFP1 affects directly the mRNA concentration associated with that gene.} 
\label{F:reg_net_1a}
\end{figure}

Antoneli \etal~\cite{antoneli2023} employ the theory for input-output networks explained before to deal with infinitesimal homeostasis in general gene regulatory networks (GRN).
As we will explain in a moment a GRN is not exactly a biochemical network as in \cite{rbgsn2017}.
There is a `mismatch' between the number of nodes in the network and the number of state variables of the underlying system of ODEs.
In order to resolve this `mismatch',
we generalize the approach of \cite{antoneli2018}, used to analyze feedforward loops, to arbitrary GRNs.
The idea there was to replace each `gene' node of a GRN by a pair of `protein' and `mRNA' nodes to obtain a {\em protein-mRNA network} (PRN).
Now, PRNs have a mathematical structure similar to that of biochemical networks and hence the theory of \cite{wang2021,duncan2024} can be readily applied.
More importantly, since the classification results have a purely combinatorial side, we can consider the GRN and its associated PRN simultaneously, and work out a correspondence between the classifications of homeostasis types and homeostasis patterns in both of them.
Even though infinitesimal homeostasis only makes sense on the PRN (dynamical) level, its purely combinatorial aspects can be transferred to the GRN.

The main result of \cite{antoneli2023} is a complete characterization of homeostasis types and homeostasis patterns on the PRN that have correspondent on the GRN.
A byproduct of this characterization is the discovery of homeostasis types and homeostasis patterns on the PRN that do not have counterparts on the GRN.
The novelty of the approach of \cite{antoneli2023} is the simultaneous use of two networks, the GRN and the PRN, in the analysis of gene expression homeostasis, and the lack of assumptions about the functional form of the differential equations.

At the abstract level a GRN is a directed graph whose nodes are the genes and a directed link from a source gene to a target gene indicates that the gene product of the source gene acts as a molecular regulator of the target gene.
{\em Autoregulation} occurs when a link connects a gene to itself, that is, the gene product is a molecular regulator of the gene itself.

At the dynamical level, a gene (a node in the GRN) represents the collection of processes that ultimately lead to the making of the gene product.
A protein-coding gene should have at least two processes: (i) {\em transcription}, that is, the synthesis of a mRNA copy from a DNA strand and (ii) {\em translation}, that is, the synthesis of protein from a mRNA. 
The `output' of the corresponding node in the GRN is the protein concentration at a given time.

Lut us assume the simplest scenario, namely, a GRN containing only protein-coding genes.
Then, each gene (node) represents two concentrations (the concentration of mRNA and the concentration of protein). 
That is, there is a mismatch between the number of nodes $N$ and the number of state variables $2N$.
There are two ways to deal with this mismatch.
\begin{enumerate}[(1)]
\item {\it Protein-protein formalism.} This approach is based on the fact that, in some cases, the changes to mRNA concentrations occur much faster than the changes to the concentrations of the associated proteins \cite{phb2009}. 
More specifically, the mRNA concentration quickly reaches a steady-state value before any protein is translated from it.
Formally, this technique is called {\em quasi steady-state approximation} (QSSA) \cite{snowden2017}.
Then, we can solve the steady-state mRNA equations and plug the result in the protein equations.
This procedure effectively reduces the number of state variables by half, thus matching of the number of nodes in the GRN (see \cite{ks2008,rs2017a,rs2017b,shi2017}).

\item {\it Protein-mRNA formalism.}
In this approach we keep the mRNA and protein concentrations for each gene and double the number of nodes of the network, leading to the notion of {\em protein-mRNA network} (PRN) \cite{phb2009,mszt2016}.
Now the network has two `types' of nodes (mRNA and protein) and two `types' of arrows (mRNA $\to$ protein and protein $\to$ mRNA).
As we will see below there is a correspondence between the two networks that allows us to transfer some properties back and forth. 
For instance, this is the approach adopted in Antoneli \etal~\cite{antoneli2018} for the particular example of the feedforward loop motif (see also \cite{kebc2005}).
\end{enumerate}

There are mathematical and biological reasons to prefer the second possibility.
From the mathematical point of view it is more convenient to work with the PRN \cite{mrs2023}. 
More specifically, it allows us to use the general theory of network dynamics \cite{gs2023} to associate a natural class of ODEs to a PRN, which contains virtually all models discussed in the literature. 
Moreover, it makes it possible to apply the techniques developed in Wang \etal~\cite{wang2021} and Duncan \etal~\cite{duncan2024} to classify the homeostasis types in PRN and GRN.

Biologically, the use of protein-protein networks is more appropriate to model prokaryotic gene regulation, because: (i) transcription and translation occur, almost simultaneously, within the cytoplasm of a cell due to the lack of a defined nucleus, (ii) the coding regions typically take up $\sim 90\%$ of the genome, whereas the remaining $\sim 10\%$ does not encode proteins, but most of it still has some biological function (e.g., genes for transfer RNA and ribosomal RNA).
In this case, it is reasonable to assume that gene expression is regulated primarily at the transcriptional level.
On the other hand, gene expression in eukaryotes is a much more complicated process, with several intermediate steps: transcription occurs in the nucleus, where mRNA is processed, modified and transported, and translation can occur in a variety of regions of the cell.
In particular, several non-coding genes that are transcribed into functional non-coding RNA molecules, such as, microRNAs (miRNAs), short interfering RNAs (siRNAs) and long non-coding RNAs (lncRNAs), play an important role in eukaryotic gene regulation.
As we will explain later, it is very easy to incorporate these regulatory elements into the framework of PRNs (see Remark \ref{RMK:non_coding}).
For the moment we make the simplifying assumption that the GRN consists of protein-coding genes with transcription and translation.

Since we are interested in gene expression homeostasis, we consider {\em input-output} GRNs. 
They are supplied with an external parameter $\II$ -- e.g., environmental disturbance or transcription activity (a function of the concentration of transcription factors) -- that affects the mRNA transcription of one gene, called the {\em input gene} $\iota$ of the GRN.
The protein concentration of a second gene, called the {\em output gene} $o$ of the GRN, is the concentration where we expect to exhibit homeostasis.
These two distinguished nodes are fixed throughout the analysis.
We assume that only the input node is affected by the external parameter $\II$.

\begin{definition}
Let $\mathcal{G}$ be an input-output GRN.
We define the associated PRN $\mathcal{R}$ as follows.
\begin{enumerate}[(a)]
\item Every node $\rho$ in the GRN $\mathcal{G}$ corresponds to two nodes in the associated PRN $\mathcal{R}$: $\rho^R$ (the mRNA concentration of gene $\rho$) and $\rho^P$ (the protein concentration of gene $\rho$).
Since there is no intermediary process, the protein concentration $\rho^P$ is affected only by the mRNA concentration $\rho^R$.
\item There is a PRN arrow from $\rho^R\to\rho^P$ and no other PRN arrow has head node $\rho^P$. 
In addition, each GRN arrow $\sigma\to\rho$ leads to a PRN arrow from the protein concentration $\sigma^P$ to the mRNA concentration $\rho^R$ (that is, $\sigma$ is a transcription factor of $\rho$).
Note that each arrow in the GRN leads to a single arrow in the PRN. 
In particular, autoregulation in $\rho$ leads to an arrow of $\rho^P\to\rho^R$.
\item Finally, if the GRN has an input gene $\iota$ and output gene $o$, then the associated PRN $\mathcal{R}$ is an input-output network with input node $\iota^R$ and output node $o^P$.
\end{enumerate}
In particular, a PRN $\mathcal{R}$ is always a \emph{bipartite digraph (directed graph)} (see \cite{pkpbmb2018}), where the two distinguished subsets of nodes are the $\rho^R$ nodes and the $\rho^P$ nodes.
\END
\end{definition}

For example, the abstract GRN corresponding to the feedforward loop motif shown in Figure~\ref{F:reg_net_1a} is the $3$-node input-output network shown in Figure~\ref{F:reg_net_1b}(a). 
Its associated PRN is the $6$-node input-output network shown in Figure \ref{F:reg_net_1b}(b).

\begin{figure}[!htb]
\begin{subfigure}[b]{0.45\textwidth}
\centering
\includegraphics[width=\textwidth]{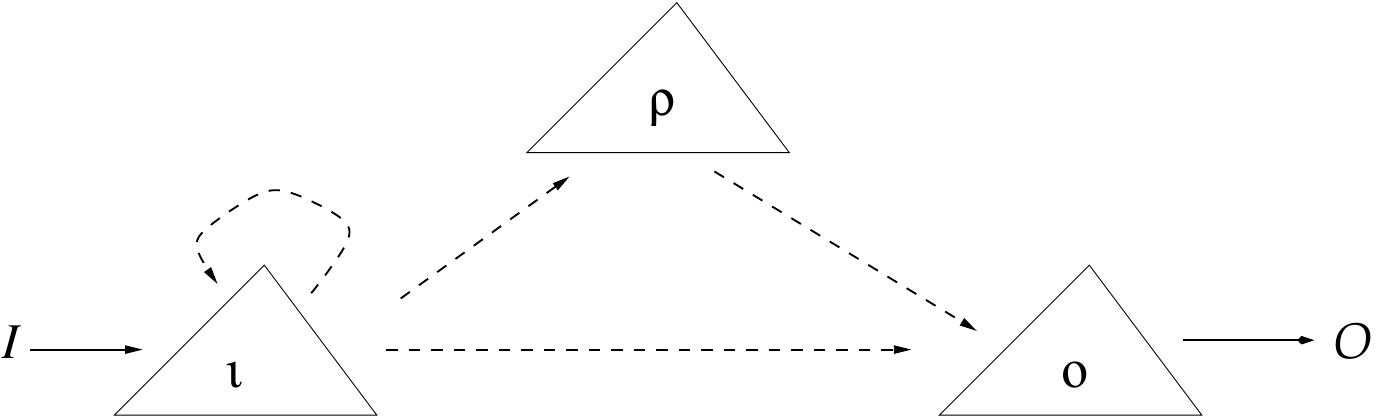}
\caption{GRN}
\label{F:3node_grn2}
\end{subfigure} \qquad
\centering
\begin{subfigure}[b]{0.45\textwidth}
\centering
\includegraphics[width=\textwidth]{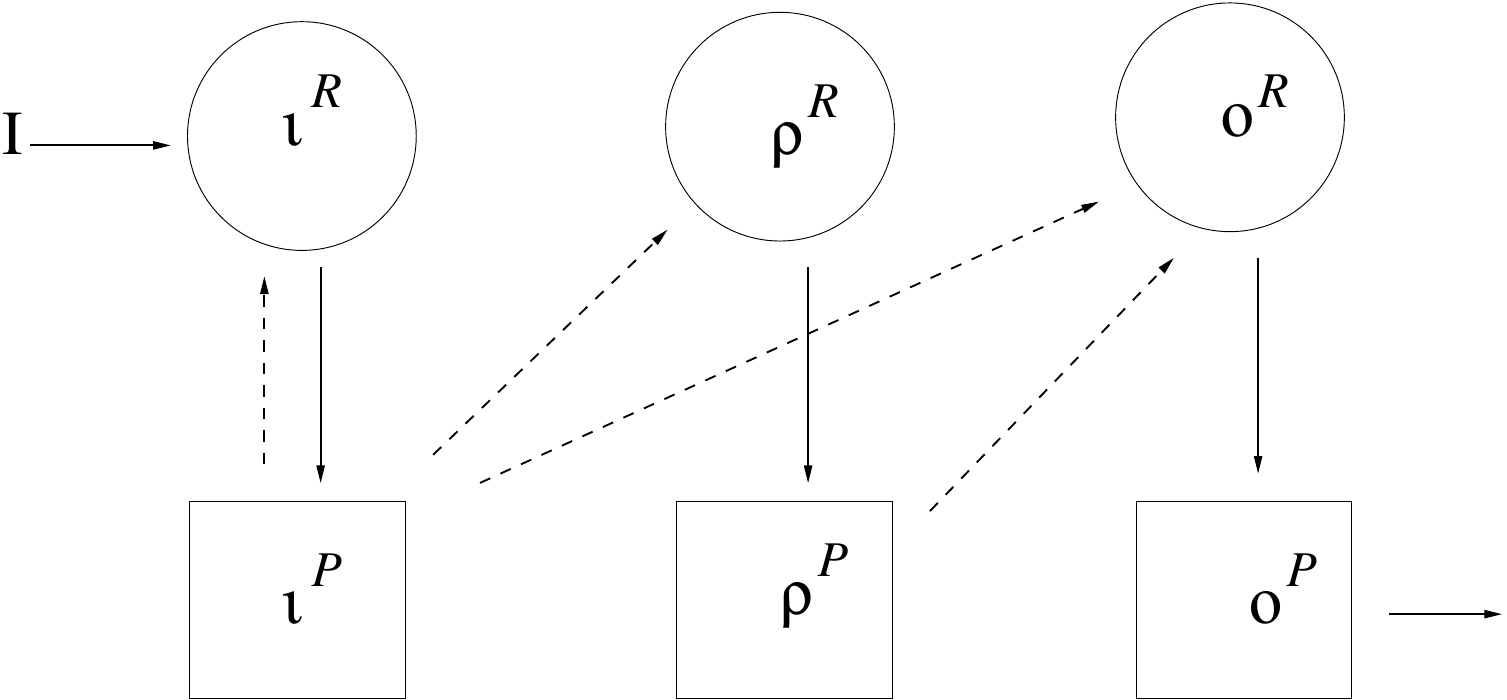}
\caption{PRN}
\label{F:3node_prn2}
\end{subfigure} \qquad
\caption{\label{F:reg_net_1b}
\textbf{Feedforward loop.}
(a) The $3$-node input-output GRN. Triangles designate genes and dashed arrows designate either gene coupling or auto-regulation. 
(b) The corresponding $6$-node PRN. Circles designate mRNA concentrations and squares designate protein concentrations. 
Solid lines stand for ${}^R\longrightarrow {}^P$ coupling inside a single gene and dashed lines for ${}^P\dashrightarrow {}^R$ coupling between genes, that is, the couplings (arrows) inherited from the GRN.}
\end{figure}

In general, the special bipartite structure of the PRN imposes restrictions on the functional form of the admissible vector fields.
For each gene $\rho$ there is a pair of 
of PRN nodes $\rho^R$, $\rho^P$ that yields a pair of $1$-dimensional state variables and corresponding differential equations have the form
\begin{equation} \label{EQ:GEN_FORM_ADM}
\begin{split}
\dot{\rho}^R & = f_{\rho^R}(\rho^R,\rho^P,\tau^P_1,\ldots,\tau^P_k) \\
\dot{\rho}^P & = f_{\rho^P}(\rho^R, \rho^P) \\
\end{split}    
\end{equation}
Here, $f_{\rho^R}$ and $f_{\rho^P}$ are smooth functions.
The variables $\tau^P_1,\ldots,\tau^P_k$ are the transcription factors (TFs), that is, the corresponding protein concentrations of the genes $\tau_1,\ldots,\tau_k$ that regulate gene $\rho$.
They are determined by the GRN arrows $\tau_i\to\rho$ and the corresponding PRN arrows $\tau^P_i\to\rho^R$.
The presence of the variable $\rho^P$ in the function $f_{\rho^R}$ occurs if and only if gene $\rho$ has a self-coupling in the GRN (see Remark \ref{RMK:self-coupling} below).
If $\rho$ is the input node then the function $f_{\rho^R}$ depends explicitly on the input parameter $\II$, as well.
From now on we will assume only the general form \eqref{EQ:GEN_FORM_ADM} for each protein-coding gene, since our classification results depend only on the GRN coupling structure, not on the particular form of the equations (see Remark \ref{RMK:act_rep}).

\begin{remark}[{\bf Activation and Repression}] \normalfont \label{RMK:act_rep}
Very often in the literature GRNs are drawn with two types of arrows:
(i) $\tau\rightarrow\rho$ to indicate that gene $\tau$ (more specifically, its protein) acts as an {\em activator}, or {\em excitatory} transcription factor, of gene $\rho$, and (ii) $\tau\flatrightarrow\rho$ to indicate that gene $\tau$ (more specifically, its protein) acts as a {\em repressor}, or {\em inhibitory} transcription factor, of gene $\rho$.
In terms of the associated differential equations this information is encoded in the $\tau$-dependence of the function $f_{\rho^R}$ in \eqref{EQ:GEN_FORM_ADM}.
Typically, the dependence of $f_{\rho^R}$ on $\tau$ is defined by the so called {\em gene input function}. Well-known examples of gene input functions are the classical {\em Michalis-Menten} and {\em Hill} functions \cite{santillan2008}, and their multi-variate versions \cite{kbzda2008}.
Since we do not specify the functional form of the differential equations, we will not use distinct arrow types in the GRN to indicate activation/repression of genes.
However, we do employ distinct arrow types (see \cite{ch2010}) in the PRN to distinguish mRNA to protein (${}^R\longrightarrow {}^P$)  and protein to mRNA (${}^P\dashrightarrow {}^R$) couplings.
\END
\end{remark}

\begin{remark}[{\bf Autoregulation}] \normalfont \label{RMK:self-coupling}
Self-coupling, or autoregulation, is a peculiar feature of GRNs.
It means that the gene product acts as a transcription factor of the gene itself.
The dynamical interpretation of autoregulation is revealed by the associated PRN.
It is a coupling from the protein node to the mRNA node of the {\em same} gene (see Figure \ref{F:reg_net_1b}).
Moreover, it should be clear that the self-coupling representing autoregulation is not the same as `self-interaction'.
In fact, all nodes of the PRN are {\em self-interacting}, in the sense that the right-hand side of each differential equation explicitly depends on the state variable on the left-hand side.
The clarification of the dynamical interpretation of autoregulation is another advantage of the PRN formalism.
\END
\end{remark}

Another simple network motif is shown in Figure \ref{F:3node_grn}.
Figure \ref{F:3node_prn} shows the associated PRN.
It is called \emph{feedback inhibition} and it plays an important role in the GRN of {\em E. coli}.
Feedback inhibition appears twice in the regulatory cascade of carbohydrate catabolism of {\em E. coli} \cite[Figs. 1 and 2(a)]{mjt2008}, with $(\iota,\tau,o)=(\textrm{IHF},\textrm{CRP},\textrm{FIS})$ and $(\iota,\tau,o)=(\textrm{ARC-A},\textrm{HNS},\textrm{GAD-X})$.

\begin{figure}[!htb]
\begin{subfigure}[b]{0.3\textwidth}
\centering
\includegraphics[width=\textwidth]{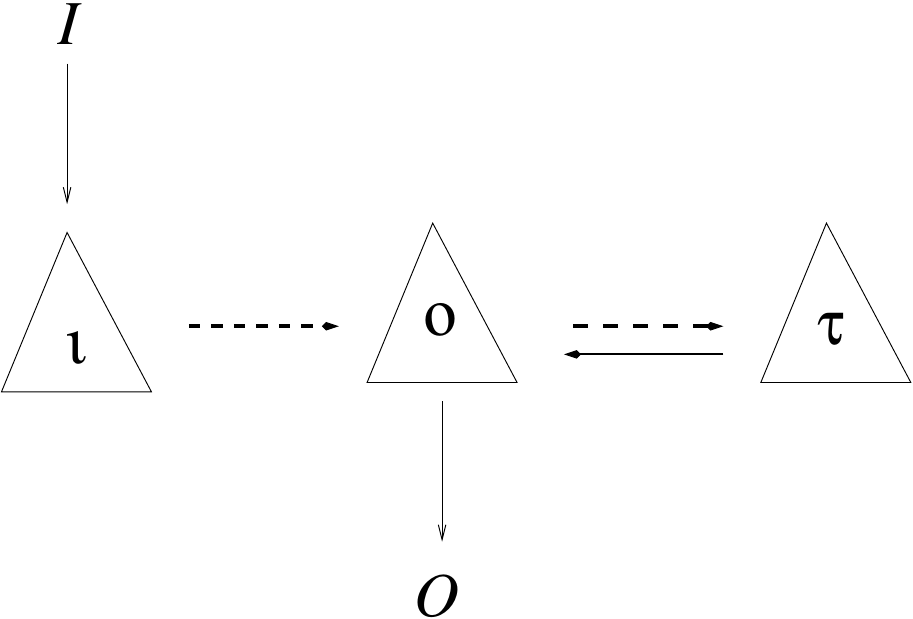}
\caption{GRN}
\label{F:3node_grn}
\end{subfigure} \qquad\qquad\qquad
\centering
\begin{subfigure}[b]{0.25\textwidth}
\centering
\includegraphics[width=\textwidth]{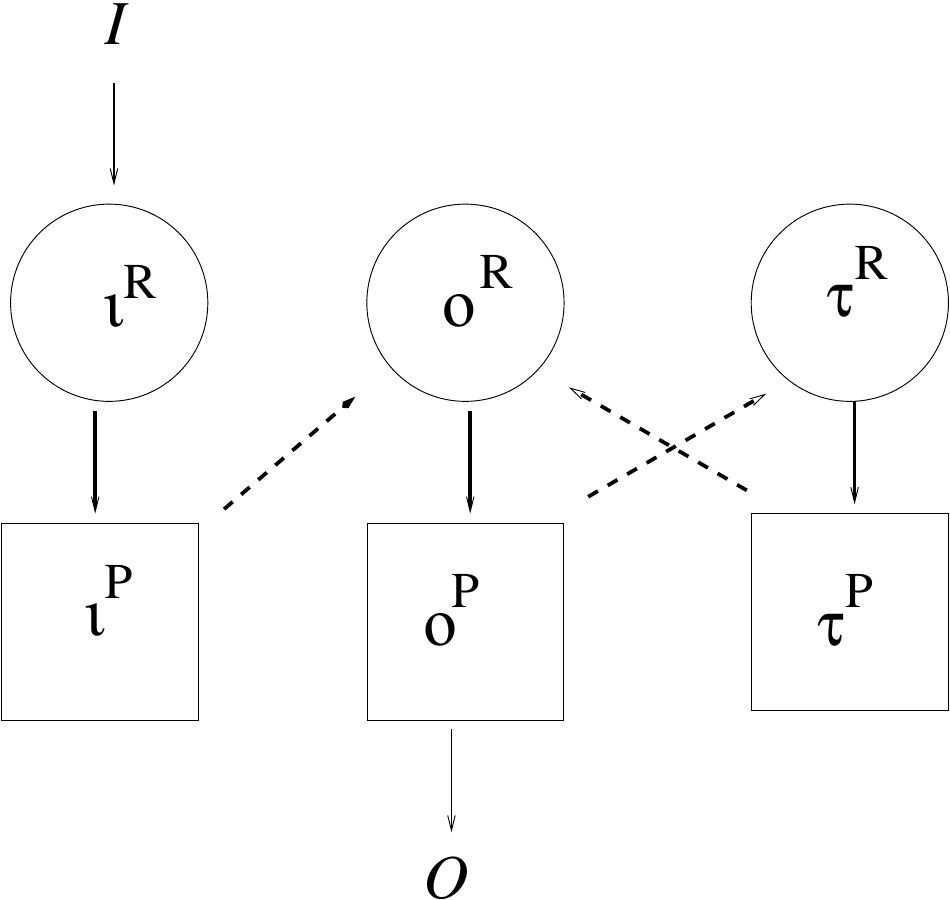}
\caption{PRN}
\label{F:3node_prn}
\end{subfigure} \qquad
\caption{\label{F:3node}
\textbf{Feedback inhibition.}
(a) The $3$-gene input-output GRN.
Triangles designate genes
and dashed arrows designate gene coupling.
(b) The corresponding $6$-node input-output PRN.
Circles designate mRNA concentrations and squares designate protein concentrations. 
Solid lines ${}^R\longrightarrow {}^P$ stand for coupling inside a single gene and dashed lines ${}^P\dashrightarrow {}^R$ for coupling between genes, that is, the couplings inherited from the GRN.}
\end{figure}

Antoneli \etal~\cite{antoneli2023} show that it is straightforward to use the PRN construction together with the theory of \cite{wang2021} to obtain all the possible homeostasis types and corresponding homeostasis patterns on the PRN.
Now it remains to explain how the results obtained for the PRN can be `lifted back' to the GRN.
In other words, we first need to define what it means for a node in a GRN to be homeostatic.
Then, we use the general theory to determine in a purely combinatorial fashion, the `formal homeostasis subnetworks' of the GRN.
Lastly, we show how the homeostasis subnetworks of the GRN and the associated PRN relate to each other.
Since only the homeostasis subnetworks of the PRN have a `dynamical' interpretation of infinitesimal homeostasis, we use the relation obtained before reinterpreting infinitesimal homeostasis on the GRN level.

\begin{definition} \normalfont
\label{def:special_component}
Let $\mathcal{G}$ be a GRN with associated PRN $\mathcal{R}$.
Let $\tau\in\mathcal{G}$ be a node with
$\tau^R\in\mathcal{R}$ be the mRNA node and $\tau^P\in\mathcal{R}$ be the protein node.
\begin{enumerate}[(a)]
\item An appendage node $\tau\in\mathcal{G}$ is a {\em single appendage node} if $\{\tau\}$ is an appendage subnetwork of $\mathcal{G}$ with no self-coupling.
\item If $\{\tau^R\}$ is an appendage subnetwork of $\mathcal{R}$ then it is called {\em $R$-null-degradation}.
\item If $\{\tau^P\}$ is an appendage subnetwork of $\mathcal{R}$ then it is called {\em $P$-null-degradation}.
\item If $\langle\tau^R,\tau^P\rangle$ is a structural subnetwork of $\mathcal{R}$ then it is called {\em $\mathcal{R}$-Haldane}.
\END
\end{enumerate}
\end{definition}

Now we state the first main result of \cite{antoneli2023}, regarding the relation between the homeostasis subnetworks of a GRN and its associated PRN.

\begin{theorem}[{\cite[Theorem 2.5]{antoneli2023}}] \label{thm:main1}
The homeostasis subnetworks of a GRN $\mathcal{G}$ and its associated PRN $\mathcal{R}$ correspond uniquely to each other, except in the following cases:
\begin{enumerate}[{\rm (a)}]
\item The $\mathcal{R}$-Haldane subnetworks of $\mathcal{R}$ correspond uniquely to the super-simple nodes of $\mathcal{G}$ (super-simple nodes are not homeostasis subnetworks of $\mathcal{G}$).
\item For every single appendage node $\tau\in\mathcal{G}$ the appendage subnetwork $\{\tau\}$ of $\mathcal{G}$ yields 2 appendage subnetworks $\{\tau^R\}$ ($R$-null-degradation) and $\{\tau^P\}$ ($P$-null-degradation) of $\mathcal{R}$.
\end{enumerate}
\end{theorem}
\begin{proof}
See \cite[Sec. B]{antoneli2023}.
The proof relies on the results of \cite{wang2021}.
\qed
\end{proof}

Now we consider the following question: {\em Can we combinatorially determine the homeostasis patterns of a GRN and its associated PRN?}
In principle this can be done since the classification of homeostasis patterns is purely combinatorial.
We start by introducing the notion of a homeostasis pattern in a GRN that is derived from the associated PRN.

\begin{definition} \normalfont
\label{def: pattern in GRN}
Consider a GRN $\mathcal{G}$ and its associated PRN $\mathcal{R}$.
Suppose that infinitesimal homeostasis occurs in the PRN at $\II_0$.
A node $\rho\in\mathcal{G}$ is said to be {\em GRN-homeostatic} if both associated PRN-nodes $\rho^R$ and $\rho^P$ are simultaneously homeostatic at $\II_0$.
A \emph{GRN-generating homeostasis pattern} is a homeostasis pattern $\mathcal{P}$ on $\mathcal{R}$ such that, for every PRN-node in $\mathcal{P}$ corresponds to GRN-homeostatic node.
\END
\end{definition}

In a GRN-generating homeostasis pattern all PRN-nodes appear in mRNA-protein pairs.
That is, the set of GRN-homeostatic nodes match perfectly the PRN homeostasis pattern.
It is useful to introduce the following terminology.
Let $\mathcal{K}$ be a homeostasis subnetwork of the GRN $\mathcal{G}$.
Then we define a map $\mathcal{K}\to\mathcal{K}^R$ from the set of homeostasis subnetworks of $\mathcal{G}$ to the set of homeostasis subnetworks of the associated $\mathcal{R}$ as follows.
If $\mathcal{K}\neq\{\tau\}$, where $\tau$ is a single appendage node, the $\mathcal{K}^R$ is the unique subnetwork given by Theorem \ref{thm:main1}.
If $\mathcal{K}=\{\tau\}$, where $\tau$ is a single appendage node, then $\mathcal{K}^R=\{\tau^R\}$.
With this definition, the map $\mathcal{K}\to\mathcal{K}^R$ is injective and the complement of its image in the set of homeostasis subnetworks of the PRN is exactly the set of $\mathcal{R}$-Haldane and $P$-null-degradation subnetworks.

\begin{theorem}[{\cite[Theorem 2.7]{antoneli2023}}] \label{thm:main2}
Let $\mathcal{G}$ be a GRN and $\mathcal{R}$ its associated PRN.
Then the homeostasis patterns on $\mathcal{G}$ correspond exactly to the GRN-generating homeostasis patterns on $\mathcal{R}$.
The homeostasis patterns on $\mathcal{R}$ associated to the $\mathcal{R}$-Haldane and the $P$-null-degradation subnetworks do not correspond to homeostasis patterns on $\mathcal{G}$.
\end{theorem}
\begin{proof}
See \cite[Sec. C]{antoneli2023}.
The proof relies on the results of \cite{duncan2024}.
\qed
\end{proof}

\begin{remark}[{\bf Non-coding genes}] \normalfont \label{RMK:non_coding}
As mentioned before, in eukaryotic cells there are several regulatory mechanisms modulating transcription and translation.
Almost all regulatory modulation is performed by non-coding genes, i.e., genes that are transcribed into RNA, but the RNA is not translated into protein.
The PRN formalism can be extended to include non-coding genes thanks to the following observation: the gene product of a non-coding gene is an mRNA whose regulatory activity is performed by direct interaction with other mRNAs
\cite{shimoni2007,lai2016}.
Thus, unlike a protein-coding gene, a non-coding gene $\nu$ yields only one scalar state variable and one differential equation 
\[
\dot{\nu} = f_{\nu}(\nu, \underbrace{\tau^P_1,\ldots,\tau^P_k}_{\text{TFs}}, \underbrace{\rho^R_1,\ldots,\rho^R_\ell}_{\text{mRNAs}})
\]
Here, $f_{\nu}$ is a smooth function.
The variables $\tau^P_1,\ldots,\tau^P_k$ are the protein concentrations associated to the transcription factors (TFs) that regulate gene $\nu$.
The variables $\rho^R_1,\ldots,\rho^R_\ell$ are the
mRNA concentrations associated to the protein-coding genes that interact with $\nu$.
Finally, for each $\rho^R_j$ above, the corresponding mRNA equation must now depend on $\nu$:
\[
\dot{\rho}^R_j = f_{\rho^R_j} (\rho^R_j,\ldots,\nu)
\]
The consequence for the PRN diagram is that a non-coding gene: (i) gives rise to a single PRN-node $\nu=\nu^R$, instead of two, (ii) receives arrows from protein nodes $\tau^P_i\to\nu^R$, (iii) has a bidirectional connection with the mRNA nodes that it interacts with $\rho^R_j \biarrow \nu$.
With these new requirements the PRN is no longer a bipartite digraph, now it is a {\em tripartite digraph}.
\END
\end{remark}

An interesting byproduct of the approach expounded above is the discovery of homeostasis types and homeostasis patterns on the PRN without GRN counterpart.
The `new' PRN homeostasis types are degree one homeostasis types, namely, they are related to one dimensional irreducible factors of the homeostasis determinant.
They are: (i) \emph{$\mathcal{R}$-Haldane}, that occurs when the linearized coupling between the mRNA and protein of the same gene changes from excitation to inhibition as the input parameter varies, and (ii) \emph{$P$-null-degradation}, that occurs when the linearized self-interaction of a protein changes from degradation to production as the input parameter varies.
Although the existence of $\mathcal{R}$-Haldane and $P$-null-degradation is mathematically established, their occurrence in biological models is unlikely.
$\mathcal{R}$-Haldane homeostasis is related to the \emph{synthesis rate} of the protein from the mRNA template and $P$-null-degradation is related to the \emph{degradation rate} of the protein.
Both these rates are {\em constant} (the first is positive and the second is negative) in specific model equations for GRN modeling \cite{phb2009,mszt2016}.

Generally speaking, all model equations for gene expression found in the literature have an explicit functional form \cite{antoneli2018,kebc2005,phb2009,mszt2016}. 
Here, we assume only the general form \ref{EQ:GEN_FORM_ADM}, forced by the admissibility of vector fields with respect to the PRN.
Hence, our classification results apply to virtually any model equation for gene expression.
Even more importantly, this leaves open the possibility to use `higher-order' terms to model more complicated interactions \cite{bghs2023}.
In the terminology of \cite{gs2023} our results are called \emph{model independent}.
This means that the classification results obtained here provide a complete list of possible behaviors, with respect to homeostasis, that is \emph{independent} of the model equations -- the list depends only on the topology of the network.
Which of those behaviors will be observed in a particular realization of the dynamics (e.g., a model equation) \emph{depends} on the specific form of the dynamics.

\subsection{Intracellular Metal Ion Regulation}
\label{SS:METAL}

Metal ions such as iron, copper, zinc, calcium, etc. are involved in many crucial biological processes and are necessary for the survival of all living organisms. 
They are ubiquitously found in all organisms, nearly exclusively as constituents of proteins, including enzymes, storage proteins and transcription factors \cite{hood2012}.
Due to the unique redox potential of some of these transition metals, many serve important roles as co-factors in enzymes and it is estimated that 30\%-45\% of known enzymes are metalloproteins whose functions require a metal co-factor \cite{klein2011}. 
However, transition metals are toxic at high intracellular concentrations, as they perturb the cellular redox potential and produce highly reactive hydroxyl radicals. 
Therefore, all organisms require mechanisms for sensing small fluctuations in metal levels to maintain a controlled balance of uptake, efflux, and sequestration and to ensure that metal availability is in accordance with physiological needs.

The key task of an intracellular metal homeostasis regulatory mechanism is to tightly control the concentration level of the metal, namely, the collection of weakly bound metal ion in the cell, which is available for a variety of interactions with other molecules.
Cells have available mechanisms to import metal from the circulation to increase the saturation level, sequester metal using a storage protein, or exporting it through an efflux pump in order to decrease the saturation level.

In the following, we describe two published mathematical models for intracellular metal ion homeostasis, calcium and zinc, that exhibit the interesting characteristic that in the corresponding network, the input node is the same as the output node.

\subsubsection{Copper}

Copper is an inorganic element essential to many physiological process, including neurotransmission, gastrointestinal uptake, lactation, transport to the developing brain  and growth. However, its concentration must be tightly regulated, as intracellular copper excess is associated to cellular damage and protein folding disorders \cite{lutsenko2007b, kaplan2016}.
In addition to cytosolic copper concentration, copper in intramitochondrial space must be also strictly regulated, as it is paramount for the function of copper dependent enzymes, but it may cause oxidative stress in excessive levels \cite{baker2017}.

Expanding briefly on the physiological implications of defective copper regulation, anomalies in the ATP7B gene generate a disorder known as Wilson disease (WD), in which dysfunctional ATP7B proteins implicate WD carriers to accumulate abnormal levels of copper in the liver and in the brain \cite{kaler2008}. 
Variations in the ATP7A gene result in dysfunctional ATP7A proteins that cause three separate illnesses: Menkes disease, a severe early-onset neuro-degenerative condition in which carriers usually die by 3 years of age \cite{kaler1994a}; occipital horn syndrome, a connective disorder with typical skeleton deformations which is also clinically resembling to Menkes disease, while less aggressive in its neurological manifestation \cite{kaler1994b}; and a recently found distal motor neuropathy, marked by frequent onset at adulthood and with no apparent signs of copper metabolic abnormalities, although still poorly studied \cite{kennerson2010, yi2012}.

A simplified version of the intracellular copper regulation mechanism was proposed in \cite{andrade2022}.
We will consider the minimal model obtained from the full model of \cite{andrade2022}.
The minimal model considers the concentration of cytosolic copper $[\text{Cu}_{\text{cyt}}]$ and
and the trans-Golgi copper $[\text{Cu}_{\text{TG}}]$ and its interaction with the metallochaperone $[\text{ATOX1}]$.
The ATOX1 protein takes the cytosolic copper to the Cu-ATPases ATP7A and ATP7B, which use ATP to pump copper ions to vesicles of the \emph{trans-Golgi} network, where copper will be incorporated in Cu-dependent enzymes and secreted. 
This is called the secretory pathway and it is responsible for decreasing the cytosolic copper concentration. 
However, when cytosolic copper levels are low, ATP7A and ATP7B take copper from the trans-Golgi network and give it to ATOX1, leading to an increase on the cytosolic copper concentration~\cite{yu2017}.

The model is given by a system of four nonlinear ordinary differential equations. For convenience we represent the
concentration of each component by:
$x_1=[\text{Cu}_{\text{cyt}}]$, $x_2=[\text{ATOX1}]$ and $x_3=[\text{Cu}_{\text{TG}}]$.
That is, $x_1$, $x_2$, $x_3$ are state variables and the equations are:
\begin{equation} \label{simplified_eqs_newvar}
\begin{aligned}
& \dot{x}_1 = \mathcal{I} - k_1 x_1 (1 + x_1) + w_1 G(x_3) \\
& \dot{x}_2 = - k_2 x_2 + k_4 x_1 - w_2 H(x_2) (x_3 - x_2) \\
& \dot{x}_3 = - k_3 x_3 + k_5 x_2 + w_2 H(x_2) (x_3 - x_2)  \\
\end{aligned}
\end{equation}
Here, the constants $w_1$, $w_2$, $k_1$, $k_2$, $k_3$, $k_4,k5$ are positive parameters and $G$ and $H$ are Hill-type functions (for $x\geq 0$):
\begin{equation*}
 G(x)=-\frac{x^2}{1+x^2} \qquad\text{and}\qquad
 H(x)=\frac{x}{1+x}
\end{equation*}
The quantity of interest to be controlled is the \emph{cytosolic copper concentration} ($x_1=[\text{Cu}_{\text{cyt}}]$)
There is one external input parameter $\II$
representing the \emph{(normalized) concentration of extracelullar copper}. 
The network associated to the model equations \eqref{simplified_eqs_newvar}
is shown in Figure \ref{fig:copper_network}.

\begin{figure}[!ht]
\centering
\includegraphics[scale=0.5]{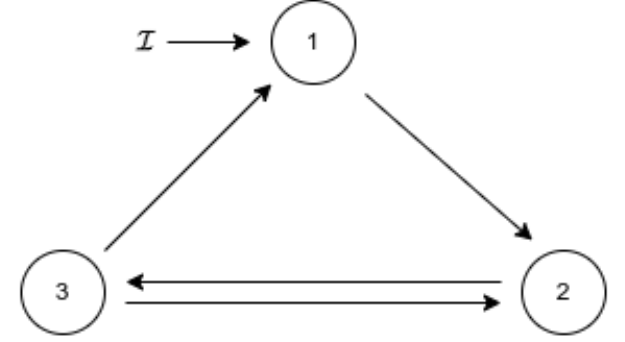}
\caption{\label{fig:copper_network}
Network for the copper homeostasis model.
}
\end{figure}

In \cite{andrade2022} the authors show that system \eqref{simplified_eqs_newvar} has a linearly stable equilibrium and the input-output function is well-defined for all $\II \geq 0$.

The general admissible system of differential equations for the network in Figure 1\ref{fig:copper_network}  is the following
\begin{equation} \label{simplified_eqs_general}
\begin{aligned}
& \dot{x}_1 = f_1(x_1,x_3,\II) \\
& \dot{x}_2 = f_2(x_1,x_2,x_3) \\
& \dot{x}_3 = f_3(x_2,x_3) \\
\end{aligned}
\end{equation}
The jacobian at an equilibrium is 
\begin{equation}
\begin{aligned}
J_{\mathrm{Cu}} & = \begin{pmatrix}
f_{1,x_1} & f_{1,x_2} & 0 \\
f_{2,x_1} & f_{2,x_2} & f_{2,x_3} \\
0 & f_{3,x_2} & f_{3,x_3} 
\end{pmatrix} 
\end{aligned}
\end{equation} 
The homeostasis matrix is obtained by removing the \emph{first row} and the \emph{first column}, i.e.
\begin{equation}
\begin{aligned}
H_{\mathrm{Cu}} & = \begin{pmatrix}
f_{2,x_2} & f_{2,x_3} \\
f_{3,x_2} & f_{3,x_3} 
\end{pmatrix} 
\end{aligned}
\end{equation} 
with
\begin{equation}
\det(H_{\mathrm{Cu}}) = f_{2,x_2} f_{3,x_3}
- f_{2,x_3} f_{3,x_2}
\end{equation} 
Hence, for a general admissible system  \eqref{simplified_eqs_general} it is possible to have infinitesimal homeostasis caused by the $2$-node appendage subnetwork $\{2 \rightleftharpoons 3\}$.

\subsubsection{Calcium}

In eukaryotic cells, calcium (an alkaline earth metal) functions as a ubiquitous intracellular messenger by which extracellular factors induce a variety of physiological responses.
For instance, the intracellular $\mathrm{Ca}^{2+}$ concentration in normal yeast cell (\emph{Saccharomyces cerevisiae}) is maintained in the range of 50–200 $n$M in the presence of environmental $\mathrm{Ca}^{2+}$ concentrations ranging from 1 $\mu$M to 100 $m$M \cite{miseta1999}. 

In \cite{cui2009} the authors propose a mathematical model of calcium homeostasis in normally growing yeast cells which is consistent with experimental observations available
at the time of its publication.

The model is given by a system of four nonlinear ordinary differential equations.
For convenience we represent the concentration of each
component by: $x_1=[\mathrm{Ca}^{2+}]$, $x_2=[\mathrm{CaM}]$, $x_3=[\mathrm{Crz1p}]$ and $x_4=[\mathrm{CaN}]$.
That is, $x_1,\ldots,x_4$ are state variables and the equations are
\begin{equation} \label{EQ:CALCIUM}
\begin{aligned}
\dot{x}_1 & = \mathcal{I} - x_3 \psi(x_4) 
\left(\frac{k_{11} x_1}{k_{12}+x_1} + 
\frac{k_{21} x_1}{k_{22}+x_1}\right) \\
& \qquad - \frac{1}{1+a_1 x_4}\frac{k_{31} x_1}{k_{32}+x_1} - 
c_1 x_1 \\
\dot{x}_2 & = a_2 x_1^3 (v_2 - x_2) - c_2 x_2 \\
\dot{x}_3 & = a_3 \phi(x_4)(1-x_3) - c_3 (1-\phi(x_4)) x_3 \\
\dot{x}_4 & = a_4 x_2 (v_4 - x_4) - c_4 x_4 
\end{aligned}
\end{equation}
Here, $a_i$, $c_j$, $v_l$, $k_{mn}$ are positive parameters and $\phi$, $\psi$ are sigmoid (rational) functions.
The quantity of interest to be controlled is the \emph{concentration of intracellular calcium} ($x_1=[\mathrm{Ca}^{2+}]$).
There is one external input parameter $\mathcal{I}$ representing the \emph{(normalized) concentration of extracelullar calcium}.
The network associated to the model equations \eqref{EQ:CALCIUM} is shown in Figure \ref{FIG:CALCIUM}.

\begin{figure}[!htb]
  \centering
  \includegraphics[scale=0.5]{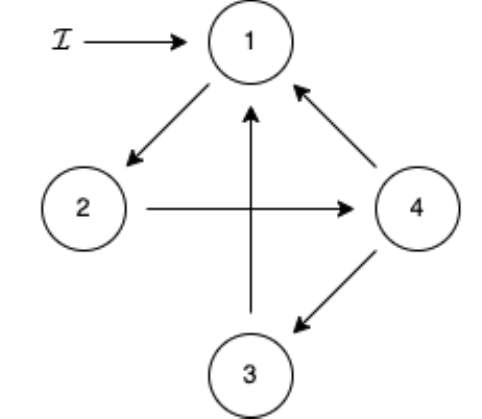}
  \caption{Network for the calcium homeostasis model.}
  \label{FIG:CALCIUM}
\end{figure}

In \cite{cui2009} the authors show, by numerical integration, that the model equation \eqref{EQ:CALCIUM} has an equilibrium point \cite[Fig. 1a]{cui2009}, for the parameters listed in \cite[Tab. 1]{cui2009}.
They numerically compute the input-output function of the model and plot \cite[Fig. 1b]{cui2009} the output node variable $x_1$ as a function of the input parameter $\mathcal{I}$. 
The simulated cytosolic calcium concentration of the model (output node variable $x_1$) rests within 73-159 $n$M (regardless of the initial conditions) when the simulated extracellular calcium concentration (the parameter $\mathcal{I}$) ranges from 1 $\mu$M to 100 $m$M.

The general admissible system of differential equations for the network in Figure \ref{FIG:CALCIUM} is the following
\begin{equation} \label{EQ:CALCIUM_GENERAL}
\begin{aligned}
\dot{x}_1 & = f_1(x_1,x_3,x_4,\mathcal{I}) \\
\dot{x}_2 & = f_2(x_1,x_2) \\
\dot{x}_3 & = f_3(x_3,x_4) \\
\dot{x}_4 & = f_4(x_2,x_4)
\end{aligned}
\end{equation}
The jacobian at an equilibrium is
\[
J_{\mathrm{Ca}}  = \begin{pmatrix}
f_{1,x_1}& 0 & f_{1,x_3} & f_{1,x_4}   \\
f_{2,x_1} & f_{2,x_2} & 0 & 0 \\
0 & 0 & f_{3,x_3} & f_{3,x_4} \\
0 & f_{4,x_2} & 0 & f_{4,x_4}  
\end{pmatrix}
\]
The homeostasis matrix is obtained by removing the \emph{first row} and the \emph{first column}, i.e.
\[
H_{\mathrm{Ca}} = \begin{pmatrix}
f_{2,x_2} & 0 & 0 \\
0 & f_{3,x_3} & f_{3,x_4} \\
f_{4,x_2} & 0 & f_{4,x_4}  
\end{pmatrix}
\]
with 
\[
\det(H_{\mathrm{Ca}}) = f_{2,x_2} f_{3,x_3} f_{4,x_4}
\]
Hence, for a general admissible system  \eqref{EQ:CALCIUM_GENERAL} it is possible to have infinitesimal homeostasis caused by null-degradation on nodes $2$, $3$ and $4$.

\subsubsection{Zinc}

Zinc is an essential micro-nutrient for plants, because it plays an important role in many enzymes catalyzing vital cellular reactions. 
In higher doses, however, zinc is toxic. 
Therefore, plants have to strictly control and adjust the uptake of zinc through their roots depending on its concentration in the surrounding soil.
This is achieved by a complicated control system consisting of sensors, transmitters and zinc transporter proteins.

In \cite{claus2015} the authors propose a mathematical model for regulation of zinc uptake in roots of \emph{Arabidopsis thaliana} based on the uptake of zinc, expression of a transporter protein and the interaction between an activator and inhibitor.
The equations in \cite{claus2015} are obtained by an \emph{ad hoc} dimensional reduction of a model proposed in \cite{claus2012}.

The model is given by a system of four nonlinear ordinary differential equations.
For convenience we represent the concentration of each component by: $x_1=[\mathrm{Zn^{2+}}]$, $x_2=[\mathrm{mRNA}]$, $x_3=[\mathrm{ZIP}]$ and $x_4=[\mathrm{Dimer}]$.
That is, $x_1,\ldots,x_4$ are state variables and the equations are
\begin{equation} \label{EQ:ZINC}
\begin{aligned}
\dot{x}_1 & = \mathcal{I} x_4  - c_1 x_1 \\
\dot{x}_2 & = 1 - a_2 x_2 x_3  - c_2 x_2\\
\dot{x}_3 & = a_3 x_1 (v_1 - x_3) -  a_5 x_2 x_3  
- c_3 x_3 \\
\dot{x}_4 & = a_4 x_2^3 (v_2 - x_4) - c_4 x_4 
\end{aligned}
\end{equation}
Here, $a_i$, $c_j$  and $v_l$ are positive parameters.
The quantity of interest to be controlled is the \emph{concentration of intracellular zinc} ($x_1=[\mathrm{Zn}^{2+}]$).
There is one external control parameter, $\mathcal{I}$ representing the \emph{(normalized) concentration of extracelullar zinc}.
The network associated to the model equations \eqref{EQ:ZINC} is shown in Figure \ref{FIG:ZINC}.

\begin{figure}[!htb]
  \centering
  \includegraphics[scale=0.5]{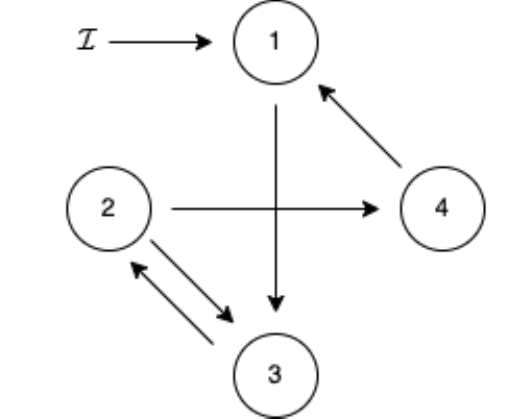}
  \caption{Network for the zinc homeostasis model.}
  \label{FIG:ZINC}
\end{figure}

In \cite{claus2015} the authors show that the model equation \eqref{EQ:ZINC} has a unique positive equilibrium point for any set of parameter values.
They show an illustration of input-output function of the model and plot \cite[Fig. 2]{claus2015} the output node variable $x_1$ as a function of the input parameter $\mathcal{I}$. 

The general admissible system of differential equations for the network in Figure \ref{FIG:ZINC} is the following
\begin{equation} \label{EQ:ZINC_GENERAL}
\begin{aligned}
\dot{x}_1 & = f_1(x_1,x_4,\mathcal{I}) \\
\dot{x}_2 & = f_2(x_2,x_3) \\
\dot{x}_3 & = f_3(x_1,x_2,x_3) \\
\dot{x}_4 & = f_4(x_2,x_4)
\end{aligned}
\end{equation}
The jacobian at an equilibrium is
\[
J_{\mathrm{Zn}}  = \begin{pmatrix}
f_{1,x_1} & 0 & 0 & f_{1,x_4}   \\
0 & f_{2,x_2} & f_{2,x_3} & 0 \\
f_{3,x_1} & f_{3,x_2} & f_{3,x_3} & 0 \\
0 & f_{4,x_2} & 0 & f_{4,x_4}
\end{pmatrix}
\]
The homeostasis matrix is obtained by removing the \emph{first row} and the \emph{first column}, i.e.
\[
H_{\mathrm{Zn}} = \begin{pmatrix}
f_{2,x_2} & f_{2,x_3} & 0 \\
f_{3,x_2} & f_{3,x_3} & 0 \\
f_{4,x_2} & 0 & f_{4,x_4}
\end{pmatrix}
\]
with 
\[
\det(H_{\mathrm{Zn}}) = f_{4,x_4} 
(f_{2,x_2} f_{3,x_3} - f_{2,x_3}f_{3,x_2})
\]
Hence, for a general admissible system  \eqref{EQ:ZINC_GENERAL} it is possible to have infinitesimal homeostasis caused by null-degradation on node $4$ and the $2$-node appendage subnetwork $\{2 \rightleftharpoons 3\}$.

\subsubsection{Iron Homeostasis}

Iron is an essential component for any aerobic organism and is required for oxygen transport, cellular respiration, and many other metabolic processes.
Dysregulation of iron homeostasis has been implicated in a wide variety of diseases, ranging from cancer to inflammatory and neuro-degenerative diseases.
But iron is also a redox active element that can facilitate the formation of dangerous molecules, such as the hydroxyl radical, a highly reactive radical that can damage DNA, lipids, and proteins. 
In \cite{chifman2012} the authors propose an explicit ODE model for the so called \emph{core control system} of intracellular iron homeostasis in the form found in breast epithelial cells (see also \cite{omholt1998}). 
The core control system of intracellular iron homeostasis is a regulatory network that contains several intertwined feedback loops and includes the key regulatory mechanisms of iron homeostasis. It was based on an intracellular network consisting of 151 chemical species and 107 reactions and transport steps, including cell type specific subnetworks found in \cite{hower2009}.

The model of \cite{chifman2012} is given by a system of five nonlinear ordinary differential equations describing changes in concentrations with respect to time for specific components of the model.
We represent the concentration of each component by: $x_1=[\mathrm{LIP}]$, $x_2=[\mathrm{TfR1}]$, $x_3=[\mathrm{Fpn}]$, $x_4=[\mathrm{Ft}]$, and $x_5=[\mathrm{IRP}]$. 
That is, $x_1,\ldots,x_5$ are state variables and the equations are
\begin{equation} \label{EQ:IRON}
\begin{aligned}
\dot{x}_1 & = a_1 x_2 \mathcal{I} - a_6 x_1 x_3 + c_4 x_4 - a_4 x_1 \frac{k_{54}}{k_{54}+x_5} \\
\dot{x}_2 & = a_2 \frac{x_5}{k_{52}+x_5} - c_2 x_2 \\
\dot{x}_3 & = a_3 \frac{k_{53}}{k_{53}+x_5} - c_3 x_3 \\
\dot{x}_4 & = a_4  x_1 \frac{k_{54}}{k_{54}+x_5} - c_4 x_4 \\
\dot{x}_5 & = a_5 \frac{k_{15}}{k_{15}+x_1} - c_5 x_5
\end{aligned}
\end{equation}
Here, $a_i$, $c_j$ and $k_{mn}$ are positive parameters.
The quantity of interest to be controlled is the \emph{liable iron pool} ($x_1=[\mathrm{LIP}]$). 
There is one external input parameter $\mathcal{I} =[\mathrm{Fe_{ex}}]$ \emph{the concentration level of extracelullar iron}.
The network associated to the model equations \eqref{EQ:IRON} is shown in Figure \ref{FIG:IRON}(a) (see \cite[Fig. 1]{chifman2012}), with the addition of arrows $5 \to 1$ and $1 \to 4$, corresponding couplings are in the equations.

We will consider a slightly different form of this model.
Since the input parameter is the amount of iron that can be transported into the cytosol, we can consider the whole term $a_1 x_2 \mathcal{I}$ as the input to the system and discard the equation for $\dot{x}_2$.
Moreover, since there is already a connection from node $5$ to node $1$, the removal of node $2$ will not make it necessary to add any new link.
In this case the equations become:
\begin{equation} \label{EQ:IRON_NEW}
\begin{aligned}
\dot{x}_1 & = \mathcal{I} - a_1 x_1 x_2 + c_3 x_3 - a_3 x_1 \frac{k_{43}}{k_{43}+x_4} \\
\dot{x}_2 & = a_2 \frac{k_{42}}{k_{42}+x_4} - c_2 x_2 \\
\dot{x}_3 & = a_3  x_1 \frac{k_{43}}{k_{43}+x_4} - c_3 x_3 \\
\dot{x}_4 & = a_4 \frac{k_{14}}{k_{14}+x_1} - c_4 x_4
\end{aligned}
\end{equation}
With $x_1=[\mathrm{LIP}]$, $x_2=[\mathrm{Fpn}]$, $x_3=[\mathrm{Ft}]$, and $x_4=[\mathrm{IRP}]$.
Again, the quantity of interest to be controlled is the \emph{cytosolic liable iron pool} ($x_1=[\mathrm{LIP}]$). 
There is one external input parameter $\mathcal{I} =[\mathrm{Fe_{ex}}]$ the \emph{concentration level of extracelullar iron}.
The network associated to the model equations \eqref{EQ:IRON} is shown in Figure \ref{FIG:IRON}(b). 

\begin{figure}[!htb]
\centering
\begin{subfigure}[b]{0.4\textwidth}
\centering
\includegraphics[width=0.7\textwidth]{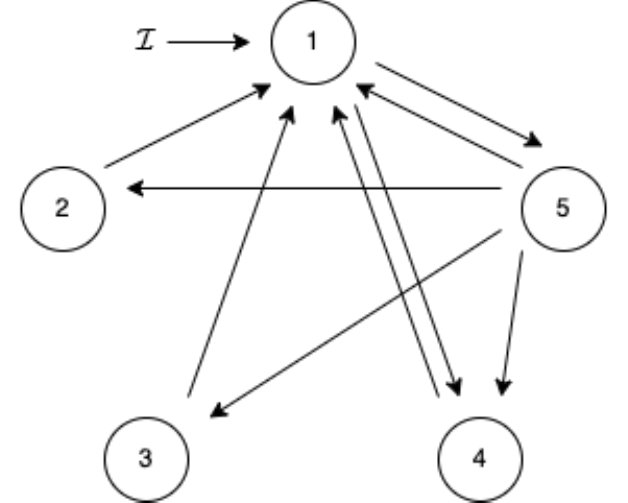}
\caption{Iron network from Chifman \etal~\cite[Fig. 1]{chifman2012}}
\label{F:IRON_NET}
\vspace{2mm}
\end{subfigure}
\centering
\begin{subfigure}[b]{0.4\textwidth}
\centering
\includegraphics[width=0.7\textwidth]{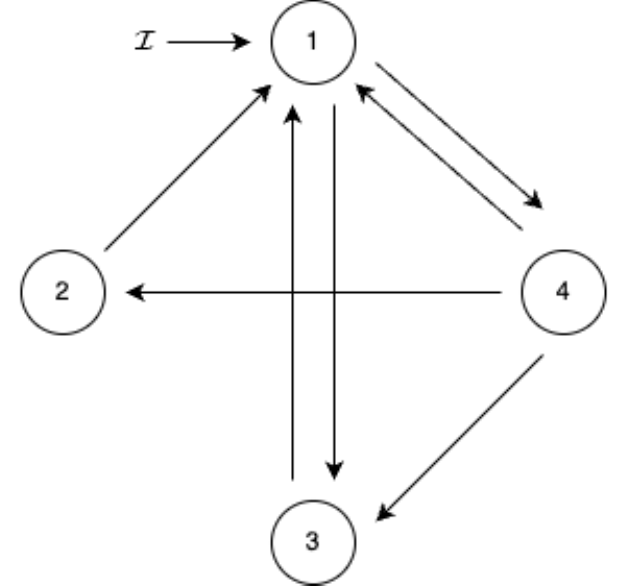}
\caption{Reduced iron network obtained from \cite[Fig. 1]{chifman2012} by discarding node $2$}
\label{F:IRON_NET_NEW}
\end{subfigure}
\caption{Networks for the iron model.}
\label{FIG:IRON}
\end{figure}

In \cite{chifman2012} the authors show that the model equations \eqref{EQ:IRON} have a unique positive equilibrium point for any set of parameter values.
They perform extensive simulations suggesting that the positive equilibrium point might be asymptotically stable in general and that \eqref{EQ:IRON} does not have limit cycle oscillations.
Furthermore, they have validated the model with experimental data that show that it correctly predicts the effect of perturbing one of the network nodes on the steady state levels of two of the other nodes. 

The general system of differential equations for the network in Figure \ref{FIG:IRON}(b) is the following
\begin{equation} \label{EQ:IRON_GENERAL}
\begin{aligned}
\dot{x}_1 & = f_1(x_1,x_2,x_3,x_4,\mathcal{I}) \\
\dot{x}_2 & = f_2(x_2,x_4) \\
\dot{x}_3 & = f_3(x_3,x_4) \\
\dot{x}_4 & = f_4(x_1,x_4) 
\end{aligned}
\end{equation}
The jacobian at an equilibrium is
\[
J_{\mathrm{Fe}}  = \begin{pmatrix}
f_{1,x_1} & f_{1,x_2} & f_{1,x_3} & f_{1,x_4} \\
0 & f_{2,x_2} & 0 & f_{2,x_4} \\
0 &  0 & f_{3,x_3} & f_{3,x_4} \\
f_{4,x_1} & 0 & 0 & f_{4,x_4} 
\end{pmatrix}
\]
The homeostasis matrix is obtained by removing the \emph{first row} and the \emph{first column}, i.e.
\[
H_{\mathrm{Fe}} = \begin{pmatrix}
f_{2,x_2} & 0 & f_{2,x_4} \\
0 & f_{3,x_3} & f_{3,x_4} \\
0 & 0 & f_{4,x_4}
\end{pmatrix}
\]
with 
\[
\det(H_{\mathrm{Fe}}) = f_{2,x_2} f_{3,x_3} f_{4,x_4} 
\]
Hence, for a general admissible system  \eqref{EQ:IRON_GENERAL} it is possible to have infinitesimal homeostasis by null-degradation on nodes $2$, $3$ and $4$.

\subsubsection{Input=Output Networks}

Antoneli \etal~\cite{antoneli2024} provides a careful adaptation of the infinitesimal homeostasis formalism to the case where the input and the output nodes are the same. 
 
When the network $\mathcal{G}$ has the same node as the input and output nodes the corresponding variables coincide $x_\iota=x_o$ and we call the network an \emph{input=output network} and $\iota$ the input $=$ output node.

In this case the vector of state variable is $X=(x_{\iota},x_{\rho})\in\mathbb{R}\times\mathbb{R}^N$ and the system of ODE's \eqref{eq:system} becomes
\begin{equation}  \label{admissible_systems_ODE_IO}
\begin{aligned}
\dot{x}_{\iota} & = f_{\iota}(x_{\iota}, x_{\rho}, \mathcal{I}) \\
\dot{x}_{\rho} & = f_{\rho}(x_{\iota}, x_{\rho})\\
\end{aligned}
\end{equation}
Finally, we assume the the \emph{genericity condition} \eqref{eq:genericity_condition}, which in this case reads:
\[
 f_{\iota,\II} \neq 0
\]
generically.

Let $J$ be the $(N+1)\times (N+1)$ Jacobian matrix of an admissible vector field $F=(f_{\iota},f_{\sigma})$, that is,
\begin{equation} \label{jacobian_general}
J = \begin{pmatrix}
  f_{\iota, x_{\iota}}   &  f_{\iota, x_\rho} \\
  f_{\rho, x_{\iota}}   &  f_{\rho, x_\rho} 
\end{pmatrix}
\end{equation}
As we have seen in the examples, the $N\times N$ \emph{homeostasis matrix} $H$ obtained from $J$ by removing the \emph{first row}:
\begin{equation}
\label{homeostasis_matrix_definition}
H = 
\begin{pmatrix}
f_{\rho, x_\rho}
\end{pmatrix}
\end{equation}
In both \eqref{jacobian_multiple} and \eqref{homeostasis_matrix_definition} partial derivatives $f_{\ell,x_j}$ are evaluated at the equilibrium $\big(X(\mathcal{I}),\mathcal{I}\big)$.

The main difference between the homeostasis matrix \eqref{homeostasis_matrix_definition} and the homeostasis matrix of network with distinct input and output nodes is that the former contains only the partial derivatives associated with the regulatory nodes, while the latter contains partial derivatives involving the input and the output nodes, as well.
In fact, the matrix $H$ in eq. \eqref{homeostasis_matrix_definition} is the Jacobian matrix of the subnetwork generated by the regulatory nodes.

Let us consider some small abstract networks.

\begin{example}[{$1$-node Networks}] \normalfont
Trivially, there is only one network with one node.
The system of ODEs associated to a one-node network is
\[
 \dot{x}_{\iota} = f_{\iota}(x_{\iota},\mathcal{I})
\]
The Jacobian is $f_{\iota,x_{\iota}}$ and so in order to have a stable equilibrium we must assume that $f_{\iota,x_{\iota}}<0$.
The homeostasis matrix is empty and the derivative of the input-output function is
\[
 x_{\iota}'=-\frac{f_{\iota,\mathcal{I}}}{f_{\iota,x_{\iota}}}
\]
Since we have assumed that $f_{\iota,\mathcal{I}}\neq 0$, generically, it follows that $x_{\iota}'\neq 0$, generically.
Hence, infinitesimal homeostasis is not  generically expected to occur in a one-node network.
\END
\end{example}

\begin{example}[{$2$-node Networks}] \normalfont
\label{EX:TWO_NODE}
There is only one possibility, up to isomorphism, in this case: $\{\iota \rightleftharpoons \rho\}$.
The system of ODEs associated to this network is
\[
\begin{aligned}
 \dot{x}_{\iota} & = f_{\iota}(x_{\iota},x_{\rho},\mathcal{I}) \\
 \dot{x}_{\rho} & = f_{\rho}(x_{\iota},x_{\rho})
\end{aligned}
\]
The Jacobian is 
\[
J=\begin{bmatrix}
  f_{\iota, x_{\iota}} & f_{\iota, x_\rho} \\
  f_{\rho, x_{\iota}} & f_{\rho, x_\rho} 
\end{bmatrix}
\]
and the homeostasis matrix is a scalar actually, $H=f_{\rho, x_\rho}$. 
Hence, infinitesimal homeostasis occurs, generically, by null-degradation with respect to regulatory node $\rho$.
\END
\end{example}

Wang \etal~\cite{wang2021} have shown that in order to analyze if an input-output network exhibits infinitesimal homeostasis, it is enough to study an associated ``core subnetwork'' (see Theorem \ref{T:coreA}). 
The same definition of core network and core-equivalence apply to an input $=$ output network and Theorem \ref{T:coreA} extends to the input $=$ output case
However, in the input $=$ output case, the condition that a node $\rho$ that is both upstream from the output node and downstream from the input node takes a special form.

\begin{lemma}
Let $\mathcal{G}$ be an input $=$ output network, with input $=$ output node $\iota$.
A regulatory node $\rho$ belongs to the core subnetwork $\mathcal{G}_c$ if and only if $\rho$ belongs to a cycle that contains the input $=$ output node $\iota$.
\end{lemma}
\begin{proof}
A directed path starting and ending at the same node is a cycle.
\end{proof}

\begin{corollary}
An input $=$ output network $\mathcal{G}$ is a core network if and only if every regulatory node of $\mathcal{G}$ belongs to a cycle that contains the input $=$ output node.
\end{corollary}

In a core network $\mathcal{G}$ where the input node is distinct from the output node, a simple node is always downstream from the input node and upstream from the output node, but not `the other way around'. 
That is, if a node is \emph{downstream} the \emph{output node} and/or \emph{upstream} the \emph{input node} then it must be an \emph{appendage node}. 
Indeed, if $\mathcal{G}$ is a core network the every node is downstream from the input node and upstream from the output node. 
Then a node that satisfies `the other way around' condition above must be on an $\iota o$-path that cycles around the input and/or the output node. 
We note the not every appendage node is of this type, that is, it may be on an $\iota o$-path that cycles around other regulatory nodes.
Hence, we have the following.

\begin{lemma} \label{THM:APPENDAGE}
Let $\mathcal{G}$ be a core input $=$ output network.
Then every node of $\mathcal{G}$ is appendage.
\end{lemma}
\begin{proof}
Since every node in a core input $=$ output network forms a cycle with the input $=$ output node it follows that every node is \emph{downstream} the \emph{output node} and \emph{upstream} the \emph{input node}.
Thus, every node is appendage.
\qed
\end{proof}

\begin{corollary}
Structural homeostasis does not exist in input $=$ output networks.
\end{corollary}

The next theorem summarizes the classification of homeostasis subnetworks of an input $=$ output network

\begin{theorem} \label{THM:APPENDAGE_NET}
If $\mathcal{G}$ is an input $=$ output network the the appendage subnetwork of $\mathcal{G}$ is exactly the subnetwork of $\mathcal{G}$ generated by all regulatory nodes.
Moreover, the irreducible factors of $\det(H)$ correspond to the appendage path components $\mathcal{A}_k$ and are given by $\det(J({\mathcal{A}_k}))$. 
Here $J({\mathcal{A}_k})$ is the jacobian matrix associated with the network $\mathcal{A}_k$.
\end{theorem}
\begin{proof}
By lemma \ref{THM:APPENDAGE} it follows that the set of nodes of $\mathcal{A}_\mathcal{G}$ is exactly the set of regulatory nodes. Hence they generate the same network.
The, it follows from \cite[Thm. 7.1]{wang2021},
that all irreducible blocks of the homeostasis matrix correspond to the appendage path components. Moreover, since there are no simple nodes the ``no cycle condition'' is trivially satisfied by all appendage path components.
Finally, it follows from \cite[Thm. 5.4]{wang2021} that the corresponding irreducible factor of $\det(H)$ is of the form $\det(J({\mathcal{A}_k}))$.
\qed
\end{proof}

\begin{theorem} \label{THM:ACYCLIC}
Let $\mathcal{G}$ be a directed graph.
Then its condensation graph $\bar{\mathcal{G}}$ is a \emph{directed acyclic graph}, i.e., $\bar{\mathcal{G}}$ has no directed cycle.
Moreover, the nodes of $\bar{\mathcal{G}}$ have an \emph{acyclic ordering}. That is, the path-components of $\mathcal{G}$ can be labeled $\mathcal{A}_1,\ldots,\mathcal{A}_m$ such that there is no arrow from $\mathcal{A}_j$ to $\mathcal{A}_i$ unless $j < i$.
\end{theorem}
\begin{proof}
See \cite[pgs. 18-19 and prop. 3.1.2]{bang-jensen2018}.
\end{proof}

Theorem \ref{THM:ACYCLIC} essentially says that each component of the condensation graph represents a feedforward network.
The acyclic order is a partial order on the set of nodes (the path-components of $\mathcal{G}$) of the condensation graph $\mathcal{G}^c$.

The condensation graph $\mathcal{A}_{\mathcal{G}}^c$ of the appendage network $\mathcal{A}_{\mathcal{G}}$ of an input $=$ output network $\mathcal{G}$ is called the \emph{condensation appendage network} of $\mathcal{G}$.
We say that a path-component $\mathcal{A}_j$ is \emph{downstream} from a path-component $\mathcal{A}_i$ if there is a directed path from $\mathcal{A}_j$ to $\mathcal{A}_i$ in the condensation appendage network $\bar{\mathcal{A}}_{\mathcal{G}}$.
A path-component $\mathcal{A}_j$ is \emph{upstream} from a path-component $\mathcal{A}_i$ if $\mathcal{A}_i$ is downstream from $\mathcal{A}_j$.
We say that $\mathcal{A}_i$ and $\mathcal{A}_j$ not comparable if there is no directed path between them.

\begin{theorem} \label{THM:HOMEO_INDUC}
Let $\mathcal{G}$ be an input $=$ output network and let $\mathcal{A}_{\mathcal{G}}$ be its appendage subnetwork.
Let $\mathcal{A}_1,\ldots,\mathcal{A}_m$ be the path-components of $\mathcal{A}_{\mathcal{G}}$ labeled according to the acyclic ordering induced by its condensation appendage network.
Suppose that homeostasis is induced by $\mathcal{A}_i$. Then $\mathcal{A}_j$ is induced by $\mathcal{A}_i$ ($\mathcal{A}_i \Rightarrow \mathcal{A}_j$) if and only if $\mathcal{A}_j$ is not downstream from $\mathcal{A}_i$ in the condensation appendage network $\bar{\mathcal{A}}_{\mathcal{G}}$.
In other words, $\mathcal{A}_i \Rightarrow \mathcal{A}_j$ if and only if
\begin{enumerate}[(i)]
\item $\mathcal{A}_j$ is not in the same $\bar{\mathcal{A}}_{\mathcal{G}}$-path-component as $\mathcal{A}_i$.
\item $\mathcal{A}_j$ is not comparable to $\mathcal{A}_i$ with respect to the acyclic order.
\item $\mathcal{A}_j$ is upstream from $\mathcal{A}_i$.
\end{enumerate}
\end{theorem}
\begin{proof}
It follows from Theorem \ref{thm:app_to_app}.
In \cite{antoneli2024} we provide another proof for the input $=$ output case.
\qed
\end{proof}

As we have shown in Theorem  \ref{THM:HOMEO_INDUC}, given an input $=$ output network $\mathcal{G}$ it is possible to classify all patterns of homeostasis supported by $\mathcal{G}$.
Moreover, each homeostasis pattern corresponds to exactly one homeostasis subnetwork.

\begin{theorem} \label{thm:bijective_corresp_i=o}
Let $\mathcal{A}_1$ and $\mathcal{A}_2$ be two distinct homeostasis subnetworks of $\mathcal{G}$. 
Then the set of subnetworks induced by $\mathcal{A}_1$ and the set of subnetworks induced by $\mathcal{A}_2$ are distinct.
\end{theorem}
\begin{proof}
It follows from Theorem \ref{thm:app_to_app}.
In \cite{antoneli2024} we provide another proof for the input $=$ output case.
\qed
\end{proof}

\subsubsection{$3$-node Input $=$ Output Networks}
\label{SSS:classification}

Now we consider the classification, up to isomorphism of the $3$-node input $=$ output core networks.

\begin{theorem}
There are $10$ isomorphism classes of $3$-node input $=$ output core networks, shown in Figure \ref{FIG:TEMPLATE}.
\end{theorem}
\begin{proof}
Consider the network of Figure \ref{FIG:SUB5} as a template for all other potential examples.
We can see that there are $2^5=32$ possibilities for choosing the remaining $5$ arrows (A, B, C, D, E) between the three nodes $\iota$, $\rho$, $\sigma$ (the arrow from $\iota$ to $\rho$ is fixed throughout the proof).
The following conditions exclude a potential example: (i) the network is \emph{disconnected}, (ii) the network is \emph{non-core}, that is , one (or more) of the nodes does not belong to cycle with the $\iota$ node, and so can be removed reducing to a network with $<3$ nodes.
Now, the potential examples can be partitioned according to the number of additional arrows. 
The number of potential examples with $k=0,\ldots,5$ additional arrows is $\binom{5}{k}$.
It is easy to see that $k=0$ and $k=1$ only give non-connected or non-core networks. The case $k=5$ corresponds to all $5$ arrows present and is a actual example, shown in  Figure \ref{FIG:SUB5}.
The three remaining cases $k=2,3,4$ can be analyzed by direct inspection.
See \cite{antoneli2024} for details.
\qed
\end{proof}

\begin{figure*}[!htb]
\begin{center}
\begin{subfigure}{.4\textwidth}
  \centering
  \includegraphics[scale=0.35]{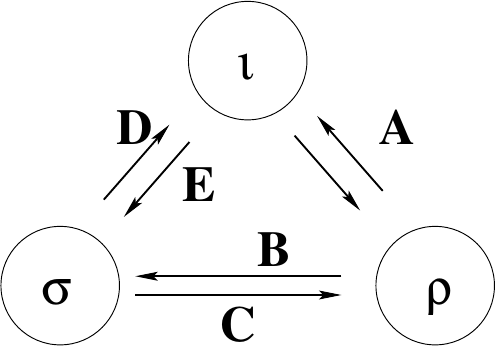}
  \caption{Networks with $5+1$ arrows}
  \label{FIG:SUB5}
\end{subfigure}
\end{center}
\begin{center}
\begin{subfigure}{\textwidth}
  \centering
  \includegraphics[scale=0.35]{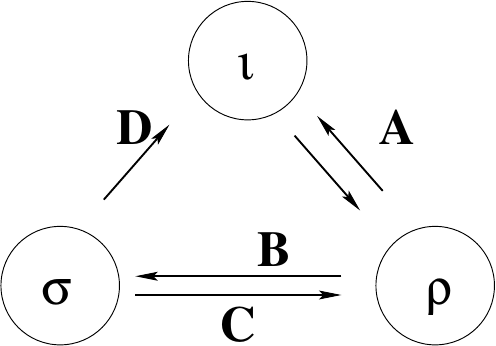} \quad
  \includegraphics[scale=0.35]{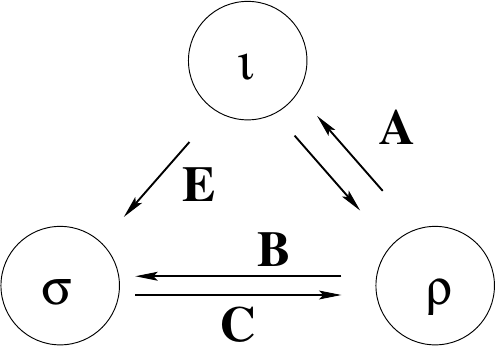} \quad
  \includegraphics[scale=0.35]{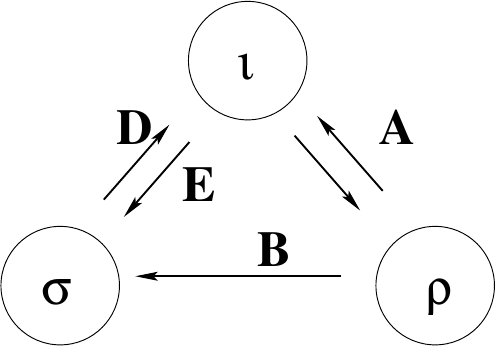}
  \caption{Networks with $4+1$ arrows.}
  \label{FIG:SUB4}
\end{subfigure}
\end{center}
\begin{center}
\begin{subfigure}{\textwidth}
  \centering
  \includegraphics[scale=0.35]{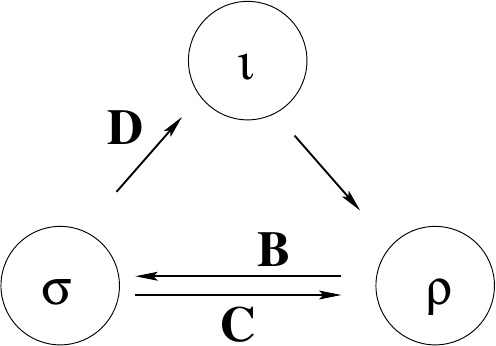} \quad
  \includegraphics[scale=0.35]{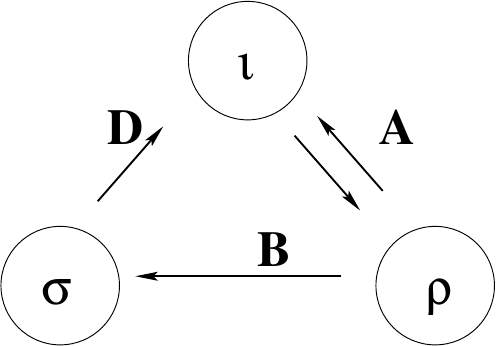} \quad
  \includegraphics[scale=0.35]{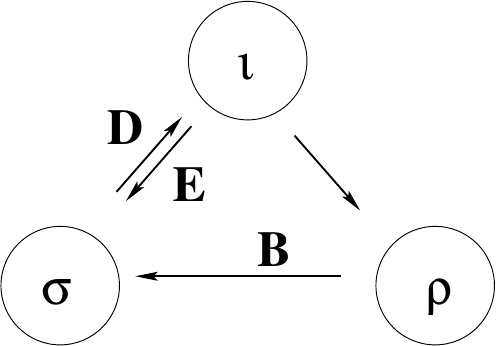} \quad
  \includegraphics[scale=0.35]{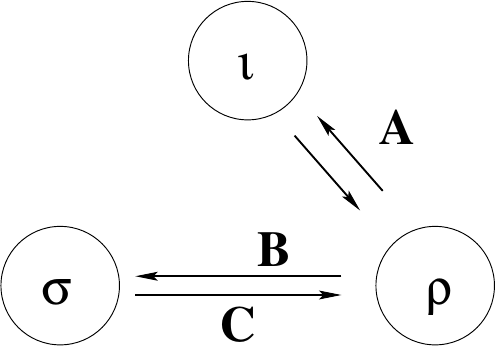} \quad
  \includegraphics[scale=0.35]{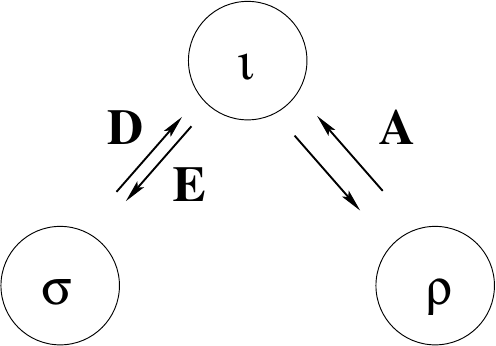}
  \caption{Networks with $3+1$ arrows.}
  \label{FIG:SUB3}
\end{subfigure}
\end{center}
\begin{center}
\begin{subfigure}{.4\textwidth}
  \centering
  \includegraphics[scale=0.35]{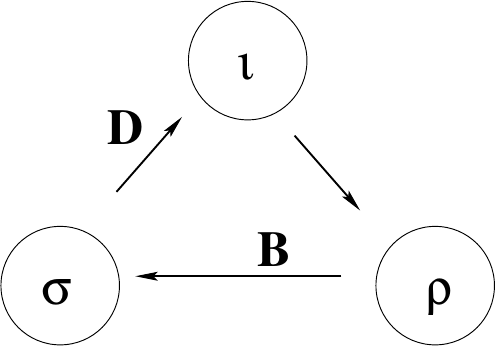}
  \caption{Networks with $2+1$ arrows.}
  \label{FIG:SUB2}
\end{subfigure}
\end{center}
\caption{\label{FIG:TEMPLATE} Three node input $=$ output networks, up to isomorphism.}
\end{figure*}

Theorem \ref{THM:APPENDAGE_NET} implies that the kinds of homeostasis that can occur in an input $=$ output network correspond to the appendage path components.
Thus it is natural to collect the isomorphism classes of networks into families according to the number of path components of the corresponding condensation appendage network.

\begin{definition} \normalfont
Define the following two families of $3$-node input $=$ output core networks according to the number of path components of the corresponding condensation appendage network:
\begin{description}
\item[Family 1:] the condensation appendage network has $1$ path component. 
In this case the appendage network is of the form $\{\sigma \rightleftharpoons \rho\}$ and the condensation appendage network is a singleton.
The $3$-node networks of Figure \ref{FIG:TEMPLATE} that belong to family 1 are: (a), (1b), (2b), (1c), (4c).
\item[Family 2:] the condensation appendage network has $2$ path components. 
In this case the the appendage network is of the form $\{\sigma \leftarrow \rho\}$ (subfamily 2A) or $\{\sigma \; \rho\}$ (subfamily 2B).
The $3$-node networks of Figure \ref{FIG:TEMPLATE} that belong to family 2 are: (3b), (2c), (3c), (d) (subfamily 2A) and (5c) (subfamily 2B).
In both cases, the condensation appendage network is equal to the appendage network and the path components are given by two singletons $\{\sigma\}$ and $\{\rho\}$. 
\END
\end{description}
\end{definition}

Now we can classify the homeostasis subnetworks of the $3$-node networks by direct application of theorem \ref{THM:APPENDAGE_NET}.

\begin{theorem} \label{thm:det_i=o}
Consider a $3$-node input $=$ output core network $\mathcal{G}$.
Then the determinant of the homeostasis matrix $H$ is as follows:
\begin{description}
\item[Family 1:] $\det(H)=f_{\sigma,x_\sigma}\, f_{\rho,x_{\rho}}-f_{\sigma,x_{\rho}} \, f_{\rho,x_{\sigma}}$.
\item[Family 2:] $\det(H)=f_{\sigma,x_\sigma}\, f_{\rho,x_{\rho}}$.
\end{description}
That is, homeostasis in a $3$-node input $=$ output core network is caused (i) by null-degradation on one of the two regulatory nodes or (ii) by the $2$-node appendage subnetwork consisting of the two regulatory nodes.
\end{theorem}
\begin{proof}
Apply theorem \ref{THM:APPENDAGE_NET} to each network in Figure \ref{FIG:TEMPLATE}.
\qed
\end{proof}

\begin{corollary} \label{thm:core_equiv_3_i=o}
Two $3$-node input $=$ output core networks are core equivalent if and only their condensation appendage networks have the same number of path components.
\end{corollary}

Combining corollary \ref{thm:core_equiv_3_i=o}
with theorem \ref{thm:class_3_node_core_i_o}
gives a complete classification of all $3$-node core networks, up to core equivalence.

\begin{corollary}[{All $3$-node core networks}]
\label{thm:class_3_node_core_i=o}
There are $5$ core equivalence clas\-ses of $3$-node input-output core networks.
Representatives are given by: (i) the Feedforward Excitation motif, (ii) the Substrate Inhibition motif, (iii) the Negative Feedback motif, (iv) the $3$-node Loop motif, Figure \ref{FIG:TEMPLATE}(d) and (v) the Antithetic Feedback motif, Figure \ref{FIG:TEMPLATE}(1c).
Any other $3$-node input-output network can be obtained from these $5$ representatives by adding and /or removing backward arrows.
\end{corollary}

\begin{remark}[Backward arrows] \normalfont
The notion of backward arrow in the input $=$ output case is more subtle.
Let $\mathcal{G}$ be an input $=$ output core network.
A \emph{backward arrow} in $\mathcal{G}$ is any arrow coming into or out of the input node such that after its removal  $\mathcal{G}$ remains a core network.
The problem here arises when we add and / or remove several backward arrows.
The order in which the arrows are added and / or removed matters.
For example, networks (1c) and (4c) in Figure \ref{FIG:TEMPLATE} are core equivalent and so each one can be converted into the other by adding and /or removing backward arrows.
To convert (1c) and (4c) one first adds the arrow $\rho\to \iota$ and then removes the arrow $\sigma\to\iota$.
To convert (4c) to (1c) one first adds the arrow $\sigma\to\iota$ and then removes the arrow $\rho\to \iota$.
In this last case, if one reverses the procedure and first removes the arrow $\rho\to \iota$ the resulting network becomes non-core.
The same happens with the first case, if one first removes the arrow $\sigma\to\iota$ the resulting network becomes non-core.
\END
\end{remark}

\begin{example}[{$3$-Node Feedback Loop}] \normalfont
\label{EX:THREE_NODE_LOOP}
Consider the $3$-Node Feedback Loop motif in Figure \ref{FIG:TEMPLATE}(d).
It is the only three node input $=$ output with three arrows.
The system of ODEs associated to this network is
\[
\begin{aligned}
 \dot{x}_{\iota} & = f_{\iota}(x_{\iota},x_{\sigma},\mathcal{I}) \\
 \dot{x}_{\rho} & = f_{\rho}(x_{\rho},x_{\iota}) \\
 \dot{x}_{\sigma} & = f_{\sigma}(x_{\sigma},x_{\rho})
\end{aligned}
\]
The Jacobian is 
\[
J=\begin{pmatrix}
  f_{\iota, x_{\iota}} & 0 & f_{\iota, x_\sigma} \\
  f_{\rho, x_{\iota}} & f_{\rho, x_\rho} & 0 \\
  0 & f_{\sigma, x_{\rho}} & f_{\sigma,x_{\sigma}}
\end{pmatrix}
\]
The homeostasis matrix is 
\[
H=\begin{pmatrix}
  f_{\rho, x_\rho} & 0 \\
  f_{\sigma, x_{\rho}} & f_{\sigma,x_{\sigma}}
\end{pmatrix}
\]
Thus
\[
\det(H)=f_{\rho, x_\rho}\,f_{\sigma,x_{\sigma}}
\]
That is, infinitesimal homeostasis can occur by null-de\-gra\-da\-tion with respect to regulatory nodes $\rho$ or $\sigma$.
\END
\end{example}

\begin{example}[{Antithetic Feedback}] \normalfont
\label{EX:ANTITHETIC_FB}
Consider the $3$-node Loop motif in Figure \ref{FIG:TEMPLATE}(1c).
This is the network for the copper homeostasis model.
This network have been considered by several authors \cite{briat2016,ferrell2016,aoki2019,tyson2023} because of its capability to exhibit perfect homeostasis.
The system of ODEs associated to this network is
\[
\begin{aligned}
 \dot{x}_{\iota} & = f_{\iota}(x_{\iota},x_{\sigma},\mathcal{I}) \\
 \dot{x}_{\rho} & = f_{\rho}(x_{\rho},x_{\iota},x_{\sigma}) \\
 \dot{x}_{\sigma} & = f_{\sigma}(x_{\sigma},x_{\rho})
\end{aligned}
\]
The Jacobian is 
\[
J=\begin{pmatrix}
  f_{\iota, x_{\iota}} & 0 & f_{\iota, x_\sigma} \\
  f_{\rho, x_{\iota}} & f_{\rho, x_\rho} & f_{\rho,x_{\sigma}} \\
  0 & f_{\sigma, x_{\rho}} & f_{\sigma,x_{\sigma}}
\end{pmatrix}
\]
The homeostasis matrix is 
\[
H=\begin{pmatrix}
  f_{\rho, x_\rho} & f_{\rho,x_{\sigma}} \\
  f_{\sigma, x_{\rho}} & f_{\sigma,x_{\sigma}}
\end{pmatrix}
\]
Thus
\[
\det(H)=
f_{\rho,x_\rho}\, f_{\sigma,x_{\sigma}}
- f_{\rho,x_{\sigma}}\, f_{\sigma,x_{\rho}}
\]
That is, infinitesimal homeostasis can occur by the $2$-node appendage subnetwork $\{\rho \rightleftharpoons \sigma\}$.
The model considered in \cite{briat2016,ferrell2016,aoki2019,tyson2023} has an additional feature that makes it capable of exhibiting perfect homeostasis.
It satisfies the \emph{antithetic identities}
\begin{equation} \label{e:antithetic_condition}
\begin{split}
f_{\rho,x_{\rho}} & \equiv f_{\sigma,x_{\rho}} \\
f_{\sigma,x_{\sigma}} & \equiv f_{\rho,x_{\sigma}}
\end{split}
\end{equation}
By plugging these identities in the expression for $\det(H)$ one easily obtains
\[
\det(H) \equiv 0
\]
Interestingly, the model equations for copper homeostasis do not satisfy the antithetic identities \eqref{e:antithetic_condition}.
Nevertheless, it has been shown in \cite{andrade2022} that this model can exhibit infinitesimal homeostasis.
\END
\end{example}

In Example \ref{ex:NF_loop} we analyzed the Negative Feedback Loop and showed that null-degradation together with linear stability implies negative feedback between the regulatory node that causes homeostasis and the output node.
It would be interesting to know if it makes sense to ask the same for the core equivalence class of the antithetic feedback, namely, is it true that appendage homeostasis together with stability implies `negative feedback'?
The first step to answer this question would be to make it more precise what it means `negative feedback' in the context of $3$-node input $=$ output networks.

\subsection{Bacterial Chemotaxis}
\label{SS:BAC_CHEMO}

The mathematical modeling of chemotaxis can be roughly divided into two types: single cell models and bacterial population models~\cite{tindall2008a,tindall2008b}. 
Single cell models consider the activation of the flagellar motor by detection of attractants and repellents in the extracellular medium. 
The flagellar motor activity of bacteria is regulated by a signal transduction pathway, which integrates changes of environmental chemical concentrations into a behavioral response. 
Assuming mass-action kinetics, the reactions in the signal transduction pathway can be modeled mathematically by ODEs. 
The bacterial population models describe evolution of bacterial density by parabolic PDEs involving an anti-diffusion `chemotaxis term' proportional to the gradient of the chemoattractant, thus allowing movement up-the-gradient, the most prominent feature of chemotaxis. 
The most extensively studied bacterial population models are the Patlak-Keller-Segel (PKS) type models.

Understanding the response of \textit{E. coli} cells to external attractants has been the subject of experimental work and mathematical models for nearly 40 years. 
In fact, many models of the chemotaxis have been formulated and developed to provide a comprehensive description of the cellular processes and include details of receptor methylation, ligand-receptor binding and its subsequent effect on the biochemical signaling cascade, along with a description of motor driving CheY/CheY-P levels, the main output of the chemotaxis system  (see~\cite{tindall2008a} for a survey).

However, including such detail has often led to very complex mathematical models consisting of tens of governing differential equations, making mathematical analysis of the underlying cellular response difficult, if not in many cases, impossible. 
A model proposed by Clausznitzer \etal~\cite{clausznitzer2010} has sought to provide a comprehensive description of the \textit{E. coli} response, by coupling a simplified statistical mechanical description of receptor methylation and ligand binding, with the signaling cascade dynamics. 
The model consists of five nonlinear ordinary differential equations (ODEs) and is parametrized using data from the literature. The authors were able to show that the model is in good agreement with experimental findings. 
However, being a fifth-order nonlinear ODE model, it is difficult to treat analytically. 
More recently, Edgington and Tindall \cite{edgington2018} undertook a comprehensive mathematical analysis of a number of simplified forms of the model of~\cite{clausznitzer2010} and proposed a fourth-order reduction of this model that has been used previously in the theoretical literature~\cite{edgington2015}.

In the following we shall consider the model proposed by \cite{edgington2015,edgington2018}. 
It has four variables for the concentrations of CheA/CheA-P ($a_p$), CheY/CheY-P ($y_p$), CheB/CheB-P ($b_p$) and the receptor methylation ($m$) and is given by the following system of ODEs (in non-dimensional form):
\begin{equation} \label{original_e_coli}
\begin{aligned}
& \frac{d m}{d t} = \gamma_{R}\, (1 - \phi(m,L)) - \gamma_{B} \, \phi(m,L) \, b_{p}^{2} \\
& \frac{d a_{p}}{d t} = \phi(m,L) \, k_{1} \, (1 - a_{p}) \\
& \qquad\quad - k_{2} \, (1 - y_{p})\, a_{p} - k_{3} \, (1 - b_{p}) \, a_{p} \\
& \frac{d y_{p}}{d t} = \alpha_{1} \, k_{2} \, (1 - y_{p})a_{p} - k_{4} \, y_{p} \\
& \frac{d b_{p}}{d t} = \alpha_{2} \, k_{3} \, (1 - b_{p}) \, a_{p} - k_{5} \, b_{p}
\end{aligned}
\end{equation}
where $\gamma_B$, $\gamma_R$, $k_1,\ldots,k_5$ are non-dimensional parameters, the extracellular ligand concentration $L$ is the \emph{external parameter} and the function $\phi$ is determined by a Monod\--Wyman\--Changeux (MWC) description of receptor clustering
\begin{equation} \label{definition_phi}
    \phi (m,L) = \frac{1}{1 + e^{F(m,L)}} 
\end{equation}
with
\begin{equation} \label{definition_F}
    F(m,L) = N\left[ 1 - \frac{m}{2} + \log \left( \frac{1 + \frac{L}{K_{a}^{\textrm{off}}}}{1 + \frac{L}{K_{a}^{\textrm{on}}}}\right) \right]
\end{equation}

The abstract input-output network obtained from system \ref{original_e_coli} is shown in Figure \ref{f:e_coli_network}.
Here, the variables in the original system correspond to the nodes as follows: $m \leftrightarrow x_{\iota_1}$, $a_p \leftrightarrow x_{\iota_2}$, $b_p \leftrightarrow x_{\sigma}$, $y_p \leftrightarrow x_{o}$, $L \leftrightarrow \mathcal{I}$.

\begin{figure}[!ht]
\centering
\includegraphics[width=0.6\linewidth]{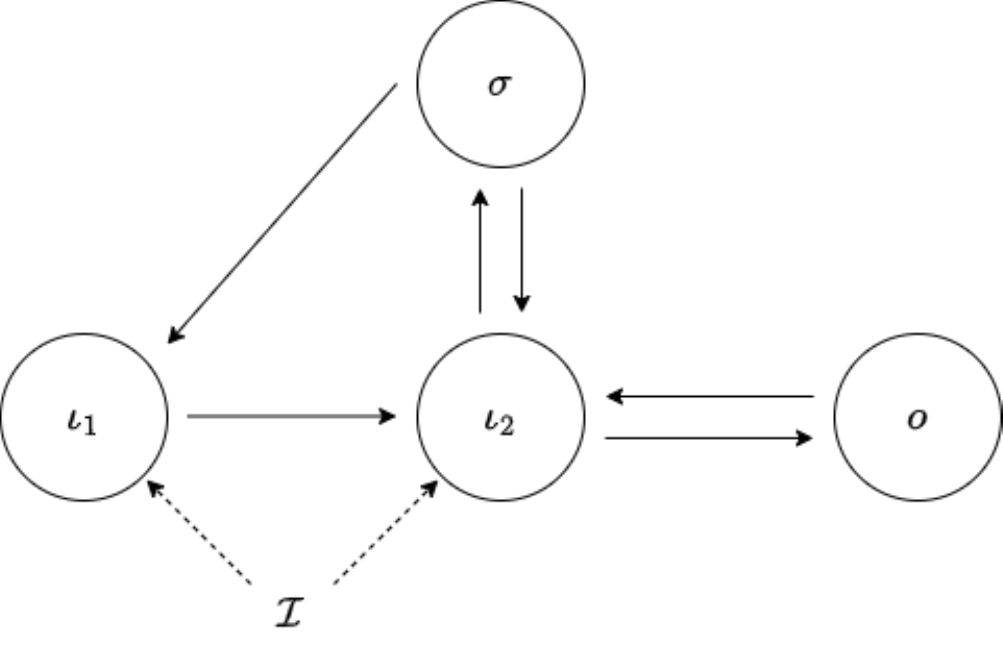}
\caption{\label{f:e_coli_network} Network corresponding to the model equations \eqref{original_e_coli} for \emph{E. coli chemotaxis.}}
\end{figure}

The problem here is that we cannot apply the theory of \cite{wang2021} to this network since it has two input nodes.
The extension of the theory to this new situation has been obtained by Madeira and Antoneli \cite{madeira2022}.

The key observation of~\cite{edgington2018} is that system \eqref{original_e_coli} has a unique asymptotically stable equilibrium 
\[
X^*=(m^*,a_p^*,y_p^*,b_p^*)
\]
with $a_p^*$, $y_p^*$ and $b_p^*$ positive and $m^{*}$ is a real number, for the non-dimensional parameters obtained from the parameter values originally used in~\cite{clausznitzer2010}. 
Furthermore, \cite{edgington2018} were able to show that some pairs of parameters might yield oscillatory behavior, but in regions of parameter space outside that observed experimentally. 
This was done by carrying out the stability analysis for pair-wise parameter variations, whereby for each case the occurrence of at least two non-zero imaginary parts was recorded as indicating possible oscillatory dynamics. 
The steady-state $X^*$ can be easily found by numerical integration for parameter values that are experimentally valid, although some combinations of parameters produce a large stiffness coefficient~\cite{clausznitzer2010}.

More importantly, the stability of $X^*$ persists as $L$ is varied in the range $D=(0,+\infty)$. 
By standard arguments this implies that there is a well-defined smooth mapping $L\mapsto X^*(L)=\big(m^*(L),a_p^*(L),y_p^*(L),b_p^*(L)\big)$. 
Since the values of $a_p^*$, $y_p^*$ and $b_p^*$ are independent of $L$ (see~\cite{edgington2018}), it follows that the individual component functions $a_p^*(L)$, $y_p^*(L)$ and $b_p^*(L)$ are actually constant functions with respect to $L$.

\begin{figure}[!ht]
\centering
\includegraphics[width=\linewidth,trim=0cm 1cm 0cm 2.5cm,clip=true]{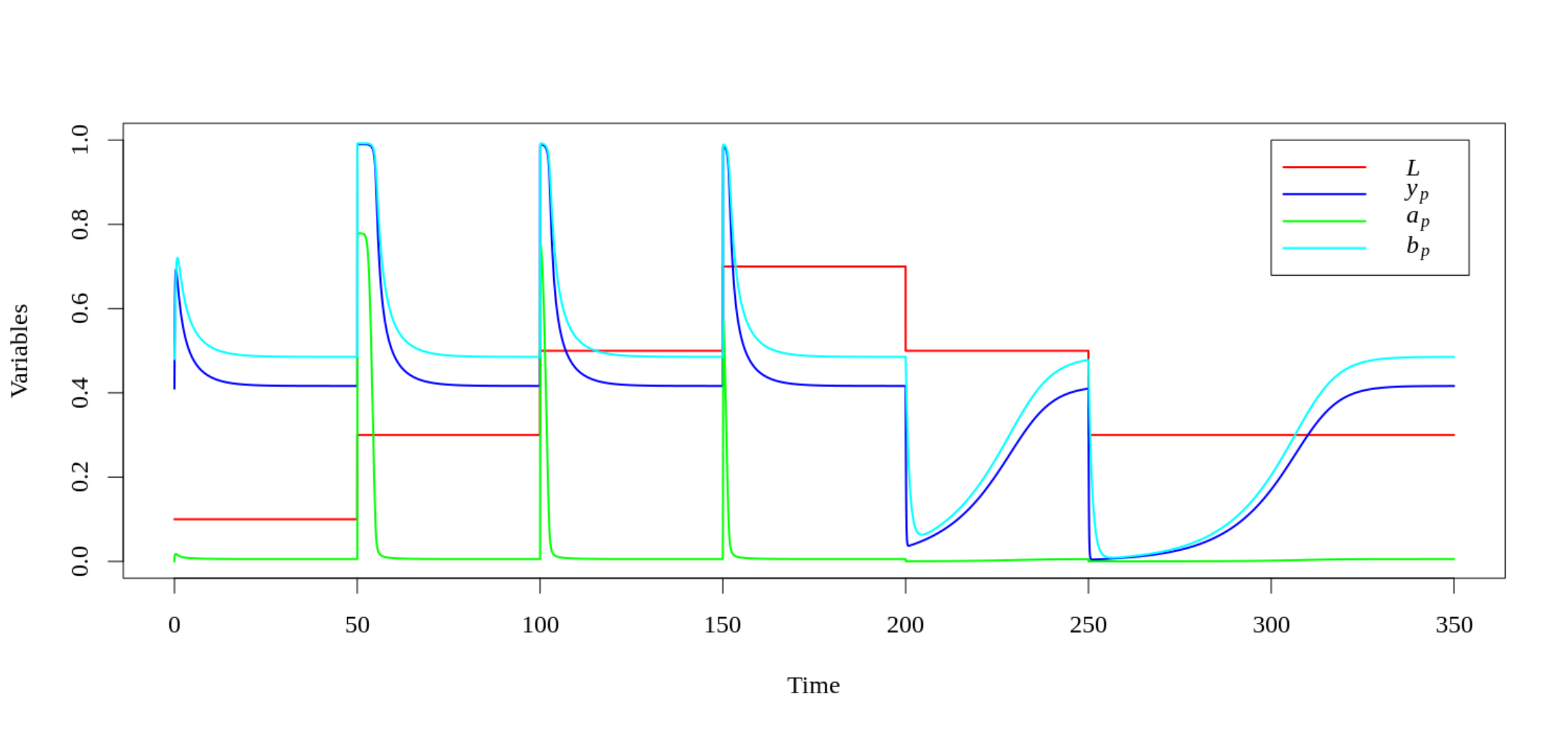}
\caption{\label{e_coli_plot1} Time series of the model \eqref{original_e_coli}, showing perfect homeostasis of the three variables $a_p$ (green), $b_p$ (cyan), $y_p$ (blue) at the non-dimensional equilibrium given by $a_p^* = 5.5 \times 10^{-3}$, $b_p^* = 0.48$, $y_p^* = 0.41$. Input parameter $L$ is given by a step function (red curve). The parameters were set to non-dimensional values of \cite[Table~2]{edgington2018}.  Time series were computed using the software \textsc{XPPAut} \cite{bard2002}.}
\end{figure}

In Figure~\ref{e_coli_plot1} we show the time series of the three variables $a_p$, $y_p$ and $b_p$ and how they are perturbed when the parameter $L$ is varied by sudden jumps. After a transient, which depends on the contraction rates at the equilibrium point, each variable returns to its corresponding steady-state value. Moreover, when the parameters $\gamma_B$, $\gamma_R$, $k_1,\ldots,k_5$ are changed the equilibrium $X^*$ also changes, but the new functions $a_p^*(L)$, $y_p^*(L)$ and $b_p^*(L)$ remain constant with respect to $L$, possibly with different values.  This is exactly the robust perfect homeostasis that was shown to occur in other models of chemotaxis \cite{barkai1997,alon1999,mello2003,yi2000,kollmann2005,hansen2008}.

\subsubsection{Multiple input-node networks}

A \emph{multiple input-node input-output network} is a network $\mathcal{G}$ with $n$ distinguished \emph{input nodes} $\iota=\{\iota_{1}, \iota_{2}, \ldots, \iota_{m}\}$, all of them associated to the same input parameter $\mathcal{I}$, one distinguished \emph{output node} $o$, and $M$ \emph{regulatory nodes} $\rho=\{\rho_1,\ldots,\rho_M\}$.
Here, we explicitly assume that the output node is distinct from all the input nodes.
The associated network systems of differential equations have the form
\begin{equation} \label{admissible_systems_ODE_multiple_input_nodes}
\begin{aligned}
\dot{x}_{\iota} & = f_{\iota}(x_{\iota}, x_{\rho}, x_{o}, \mathcal{I}) \\
\dot{x}_{\rho} & = f_{\rho}(x_{\iota}, x_{\rho}, x_{o})\\
\dot{x}_{o} & = f_{o}(x_{\iota}, x_{\rho}, x_{o})
\end{aligned}
\end{equation}
where $\mathcal{I}\in\mathbb{R}$ is an \emph{external input parameter} and $X=(x_{\iota},x_{\rho},x_o)\in\mathbb{R}^m\times\mathbb{R}^M\times\mathbb{R}$ is the vector of state variables associated to the network nodes.

An admissible vector field associated with the system \eqref{admissible_systems_ODE_multiple_input_nodes} can be witten as
\[
F(X,\mathcal{I})=(f_{\iota}(X,\mathcal{I}),f_\rho(X),f_o(X))
\]
Finally, we assume the \emph{genericity condition}
\eqref{eq:genericity_condition}, which in this case reads
\begin{equation*} \label{e:f_iota_I}
 f_{\iota_k,\mathcal{I}} \neq 0
\end{equation*}
generically, for all $1 \leq k \leq m$.

Let $J$ be the $(m+M+1)\times (m+M+1)$ Jacobian matrix of an admissible vector field $F=(f_{\iota},f_{\sigma},f_{o})$, that is,
\begin{equation} \label{jacobian_multiple}
J = \begin{pmatrix}
  f_{\iota, x_\iota}   &  f_{\iota, x_\rho} & f_{\iota, x_o} \\
  f_{\rho, x_\iota}   &  f_{\rho, x_\rho} & f_{\rho, x_o} \\
  f_{o, x_\iota} &  f_{o, x_\rho} & f_{o, x_o} 
\end{pmatrix}
\end{equation}
The $(m+M+1)\times (m+M+1)$ matrix $\langle H \rangle$ obtained from $J$ by replacing the last column by $(-f_{\iota,\mathcal{I}},0,0)^t$, is called \emph{generalized homeostasis matrix}:
\begin{equation} \label{weighted_homeostasis_matrix_definition}
\langle H \rangle = 
\begin{pmatrix}
f_{\iota, x_\iota} &  f_{\iota, x_\rho} & -f_{\iota, \mathcal{I}} \\
f_{\rho, x_\iota}&  f_{\rho, x_\rho} & 0 \\
f_{o, x_\iota} &  f_{o, x_\rho} & 0
\end{pmatrix}
\end{equation}
Here all partial derivatives $f_{\ell,x_j}$ are evaluated at $\big(\tilde{X}(\mathcal{I}),\mathcal{I}\big)$.

\begin{lemma} \label{cramer_rule}
The input-output function $x_o(\mathcal{I})$ satisfies
\begin{equation} \label{xo'}
x_o'(\mathcal{I}) = \frac{\det\!\big(\langle H \rangle\big) }{\det(J)}
\end{equation}
Here, $\det(J)$ and $\det\!\big(\langle H \rangle\big)$ are evaluated at $\big(X(\mathcal{I}),\mathcal{I}\big)$. 
Thus, $\mathcal{I}_0$ is a point of infinitesimal homeostasis if and only if
\begin{equation} \label{xo'_reduced}
\det\!\big(\langle H \rangle\big) = 0
\end{equation}
at the equilibrium $\big(X(\mathcal{I}_0),\mathcal{I}_0\big)$.
\end{lemma}
\begin{proof}
Implicit differentiation of the equation $f(\tilde{X}(\mathcal{I}),\mathcal{I})=0$ with respect to $\mathcal{I}$ yields the linear system
\begin{equation} \label{imp_diff}
J \begin{pmatrix} x_i' \\ x_\rho' \\ x_o'\end{pmatrix} =
-\begin{pmatrix}f_{\iota, \mathcal{I}} \\ 0 \\ 0 \end{pmatrix}
\end{equation}
Since $\tilde{X}(\mathcal{I})$ is assumed to be a linearly stable equilibrium, it follows that $\det(J)\neq 0$. On applying Cramer's rule to \eqref{imp_diff} we can solve for $x_o'(\mathcal{I})$ obtaining \eqref{xo'}.
\qed
\end{proof}

The notion of core network can be extended to the multiple input node context and a similar result to theorem \ref{T:coreA} holds \cite[Thm. 3.2]{madeira2022}.

By expanding $\det(\langle H \rangle)$ with respect to the last column and each $\iota_k$ (input) row one obtains 
\begin{equation} \label{xo'_reduced_expand}
\det\!\big(\langle H \rangle\big) = \sum_{k=1}^m \pm f_{\iota_k,\mathcal{I}} \det(H_{\iota_k})
\end{equation}
Note that when there is a single input node, i.e. $m=1$, Lemma \ref{cramer_rule} gives the corresponding result obtained in~\cite{wang2021}. 
In this case, there is only one matrix $H_{\iota_m}=H$, the \emph{homeostasis matrix}, that played a fundamental role in the theory developed in \cite{wang2021}. Hence, it is expected that the matrices $H_{\iota_m}$ should play a similar role in the generalization of \cite{wang2021} to the multiple input node case.

\begin{definition} \label{definition_text_parcels_homeostasis_matrix}
Let $\mathcal{G}$ be an input-output network with $n$ input nodes and
$F$ be an admissible vector field.
The \emph{partial homeostasis matrix} $H_{\iota_k}$ of $F$ is obtained from the Jacobian matrix $J$ of $F$ by dropping the last column and the $\iota_k$ row
\begin{equation} \label{definition_parcels_homeostasis_matrix}
H_{\iota_{k}} = \begin{pmatrix} f_{\iota_{1}, x_{\iota_{1}}} & \cdots & f_{\iota_{1}, x_{\iota_{m}}} & f_{\iota_{1}, x_{\sigma}} \\
\vdots & \ddots & \vdots & \vdots \\ f_{\iota_{k-1}, x_{\iota_{1}}} & \cdots & f_{\iota_{k-1}, x_{\iota_{m}}} & f_{\iota_{k-1}, x_{\sigma}} \\ f_{\iota_{k+1}, x_{\iota_{1}}} & \cdots & f_{\iota_{k+1}, x_{\iota_{m}}} & f_{\iota_{k+1}, x_{\sigma}} \\  \vdots & \ddots & \vdots & \vdots \\
f_{\iota_{n}, x_{\iota_{1}}} & \cdots & f_{\iota_{n}, x_{\iota_{m}}} & f_{\iota_{n}, x_{\sigma}} \\
f_{\sigma, x_{\iota_{1}}} & \cdots & f_{\sigma, x_{\iota_{m}}} & f_{\sigma, x_{\sigma}} \\
f_{o, x_{\iota_{1}}} & \cdots & f_{o, x_{\iota_{m}}} & f_{o, x_{\sigma}} \end{pmatrix}
\end{equation}
When the network $\mathcal{G}$ has only one input node, eq. \eqref{definition_parcels_homeostasis_matrix} gives exactly the homeostasis matrix defined in \cite{wang2021}.
\END
\end{definition}

The classification of homeostasis types proceeds as in \cite{wang2021}. 
The first step is to apply Frobenius-K\"onig theory \cite{schneider1977,brualdi1991} to the generalized homeostasis matrix $\langle H \rangle$. 
More precisely, Frobenius-K\"onig theory implies that there exist (constant) permutation matrices $P$ and $Q$ such that
\begin{equation} \label{weighted_homeostasis_normal_form}
P \langle H \rangle Q = \begin{pmatrix} B_{1} & * & \cdots & * & * \\
0 & B_{2} & \cdots & * & * \\
\vdots & \vdots & \ddots & \vdots & \vdots \\
0 & 0 & \cdots & B_{s} & * \\
0 & 0 & \cdots & 0 & C
\end{pmatrix}
\end{equation}
where each diagonal block $B_{1}$, \ldots, $B_{s}$ and $C$ is irreducible (in the sense of \cite{brualdi1991}). 
Hence, $\det(B_{1})$, \ldots, $\det(B_{s})$ and $\det(C)$ are irreducible polynomials. 
As $P$ and $Q$ are constant permutation matrices, we have that
\begin{equation} \label{FK_factorization}
\det\!\big(\langle H \rangle\big) = \pm \det(B_{1}) \cdots \det(B_{s}) \cdot \det(C)
\end{equation}
In order to simplify nomenclature, we will call $B_{1}$, \ldots, $B_{s}$ and $C$ \emph{irreducible homeostasis blocks}, although in the literature the term irreducible matrix may have a different meaning (see \cite{schneider1977}).

A direct comparison of factorization \eqref{FK_factorization} with expansion 
\eqref{xo'_reduced_expand} suggests that the irreducible factors $\det(B_j)$ are the common factors of $\det(H_{\iota_m})$ and $\det(C)$ is a weighted alternating sum of $f_{\iota_{1},\mathcal{I}}, \ldots, f_{\iota_{n},\mathcal{I}}$. 
Indeed, as we show in \cite{madeira2022} the matrix $C$ in \eqref{weighted_homeostasis_normal_form} contains all the functions $f_{\iota_{\ell}, \mathcal{I}}$ as entries, that is, it is a homogeneous polynomial of degree $1$ on $f_{\iota_{1}, \mathcal{I}}, \ldots, f_{\iota_{n}, \mathcal{I}}$, whereas the matrices $B_{1}, \ldots, B_{s}$ do not contain any of them.

The next step is to classify the irreducible homeostasis blocks of $\langle H \rangle$ according to their number of self-couplings. 
Indeed, we show that each block $B_{j}$ of order $k_{j}$ has exactly $k_{j}$ or $k_{j}-1$ self-couplings \cite[Sec. 3.4]{madeira2022}. 
But, unlike \cite{wang2021}, in the multiple input nodes case we find \emph{three} classes (instead of two classes) of irreducible homeostasis blocks that may occur in core networks with multiple input nodes.

\begin{definition} \label{definition_appendage_and_structural} \rm
Let $B_{j}$ be an irreducible homeostasis block of order $k_{j}$ that does not contain any partial derivatives $f_{\iota_{m}, \mathcal{I}}$ with respect to $\mathcal{I}$.
The homeostasis class of $B_{j}$ is \emph{appendage} if $B_{j}$ has $k_{j}$ self-couplings and \emph{structural} if $B_{j}$ has $k_{j} - 1$ self-couplings.
\end{definition}
 
\begin{definition} \label{types_homeostasis} \rm
The network $\mathcal{G}$ exhibits \emph{appendage homeostasis} if there is an appendage irreducible homeostasis block $B_{j}$ such that $\det(B_{j}) = 0$. 
In an analogous way, $\mathcal{G}$ exhibits \emph{structural homeostasis} if there is a structural irreducible homeostasis block $B_{j}$ such that $\det(B_{j}) = 0$.
\end{definition}

As shown in \cite{wang2021}, appendage and structural homeostasis occur in core networks with one input node. 
Nevertheless, networks with multiple input nodes also exhibit a new class of homeostasis that is not found in networks with only one input node.

\begin{definition} \rm
\label{definition_input_counterweight_homeostasis_block}
Let $C$ be an irreducible homeostasis block whose determinant $\det(C)$ is a homogeneous polynomial of degree $1$ on the variables $f_{\iota_{1}, \mathcal{I}}, \ldots, f_{\iota_{n}, \mathcal{I}}$. We say that the homeostasis class of $C$ is \emph{input counterweight}. Moreover, we say that $\mathcal{G}$ exhibits \emph{input counterweight homeostasis} when $\det(C) = 0$.
\end{definition}

Although input counterweight homeostasis does not occur in networks with only one input node, it is interesting to note that in these networks there is a corresponding matrix $C$. In fact, from equation \eqref{weighted_homeostasis_matrix_definition} and considering networks with only one input node $\iota$, we have
\begin{equation}
    C = [- f_{\iota, \mathcal{I}}] \quad\Rightarrow\quad 
    \det(C) = - f_{\iota, \mathcal{I}} \neq 0
\end{equation}
As, by hypothesis, $f_{\iota, \mathcal{I}} \neq 0$, the input counterweight homeostasis is never present in networks with only one input node.
On the other hand, the block $C$ can always vanish in any network with more than one input node \cite[Cor. 3.9]{madeira2022}.

The final step in this theory is to identify a homeostasis subnetwork of $\mathcal{G}$ corresponding to each homeostasis block of $\langle H \rangle$ in such a way that each class of homeostasis corresponds to a distinguished class of subnetworks. 
Because of the appearance of a third homeostasis class, the extension of the results of \cite{wang2021} to the multiple input nodes is not a trivial extension.
See \cite{madeira2022} for more details.

\subsubsection{The \emph{E. coli} Model}

Since the network corresponding to the model equations \ref{original_e_coli} has two input nodes it cannot be analyzed using the methods of \cite{wang2021}. 
In fact this can be seen directly from the equations \ref{original_e_coli}, since both the equation for $m$ and $a_p$ depend on the external parameter $L$.

In order to apply the theory developed in \cite{madeira2022} to the model for the \textit{Escherichia coli} chemotaxis of \cite{edgington2018}, we first rewrite the system \eqref{original_e_coli} in the standard form \eqref{admissible_systems_ODE_multiple_input_nodes}.
Using the correspondence between variables defined before, we obtain the following system of ODEs
\begin{equation} \label{e_coli_network_notation}
\begin{aligned}
\dot{x}_{\iota_1} & = \gamma_{R}  \, (1 - \phi(x_{\iota_1}, \mathcal{I})) - \gamma_{B} \, x_{\sigma}^{2} \, \phi(x_{\iota_1}, \mathcal{I}) \\
\dot{x}_{\iota_{2}} & = \phi(x_{\iota_1}, \mathcal{I}) \, k_{1} \, (1 - x_{\iota_2}) \\
& \quad - k_{2} \, (1 - x_{o}) \, x_{\iota_2} - k_{3} \, (1 - x_{\sigma}) \, x_{\iota_2} \\
\dot{x}_{\sigma} & = \alpha_{2} \, k_{3} \, (1 - x_{\sigma}) \, x_{\iota_2} - k_{5} \, x_{\sigma} \\
\dot{x}_{o} & = \alpha_{1} \, k_{2} \, (1 - x_{o}) \, x_{\iota_2} - k_{4} \, x_{o}
\end{aligned}
\end{equation}
with the function $\phi$ given by \eqref{definition_phi} and input parameter $\mathcal{I}$.
The ODE system \eqref{e_coli_network_notation} is an admissible system for the abstract input-output network shown in Figure \ref{f:e_coli_network}. 
In fact, the general form of an admissible system for this network is
\begin{equation} \label{e_coli_general_admissible}
\begin{aligned}
\dot{x}_{\iota_1} & = f_{\iota_1}(x_{\iota_1},x_{\sigma},\mathcal{I}) \\
\dot{x}_{\iota_2} & = f_{\iota_2}(x_{\iota_1}, x_{\iota_2},x_{\sigma},x_{o},\mathcal{I}) \\
\dot{x}_{\sigma} & = f_{\sigma}(x_{\iota_2},x_{\sigma}) \\
\dot{x}_{o} & = f_{o}(x_{\iota_2},x_{o})
\end{aligned}
\end{equation}

From the theory developed in this paper, it is clear that the network $\mathcal{G}$ is a core network with input nodes $\iota_{1}$ and $\iota_{2}$ and output node $o$.
The Jacobian matrix of network $\mathcal{G}$ is:
\begin{equation} \label{jacobian_matrix_e_coli}
J = 
\begin{pmatrix} 
f_{\iota_1, x_{\iota_1}} & 0 & f_{\iota_1, x_{\sigma}} & 0 \\
f_{\iota_2, x_{\iota_1}} & f_{\iota_2, x_{\iota_2}} & f_{\iota_2, x_{\sigma}} & f_{\iota_2, x_o} \\
0 & f_{\sigma, x_{\iota_2}} & f_{\sigma, x_{\sigma}} & 0 \\
0 & f_{o, x_{\iota_2}} & 0 & f_{o,x_o}
\end{pmatrix}
\end{equation}
The generalised homeostasis matrix of network $\mathcal{G}$ is:
\begin{equation} \label{homeostasis_matrix_e_coli}
\langle H \rangle = 
\begin{pmatrix} 
f_{\iota_1, x_{\iota_1}} & 0 & f_{\iota_1, x_{\sigma}} & -f_{\iota_1, \mathcal{I}} \\
f_{\iota_2, x_{\iota_1}} & f_{\iota_2, x_{\iota_2}} & f_{\iota_2, x_{\sigma}} & -f_{\iota_2, \mathcal{I}} \\
0 & f_{\sigma, x_{\iota_2}} & f_{\sigma, x_{\sigma}} & 0 \\
0 & f_{o, x_{\iota_2}} & 0 & 0
\end{pmatrix}
\end{equation}
It can be decomposed as
\begin{equation} \label{analysis_e_coli_part_i}
\det\big(\langle H \rangle\big) = -f_{\iota_1, \mathcal{I}} \, \det(H_{\iota_1}) + f_{\iota_2, \mathcal{I}} \, \det(H_{\iota_2})
\end{equation}
where
\[
H_{\iota_1} = \begin{pmatrix} 
f_{\iota_2, x_{\iota_1}} & f_{\iota_2, x_{\iota_2}} & f_{\iota_2, x_{\sigma}} \\
0 & f_{\sigma, x_{\iota_{2}}} & f_{\sigma, x_{\sigma}} \\
0 & f_{o, x_{\iota_2}} & 0 
\end{pmatrix} 
\]
and
\[
H_{\iota_2} = \begin{pmatrix} 
f_{\iota_1, x_{\iota_1}} & 0 & f_{\iota_1, x_{\sigma}} \\
0 & f_{\sigma, x_{\iota_2}} & f_{\sigma, x_{\sigma}} \\
0 & f_{o, x_{\iota_2}} & 0
\end{pmatrix}
\]
Hence, the complete factorization of $\det\!\big(\langle H \rangle \big)$ is
\begin{equation*} \label{analysis_e_coli_part_iii}
\begin{split}
& = f_{\sigma, x_{\sigma}} \, f_{o, x_{\iota_{2}}} \, (-f_{\iota_{1}, \mathcal{I}} \, f_{\iota_{2}, x_{\iota_{1}}} + f_{\iota_{2}, \mathcal{I}} \, f_{\iota_{1}, x_{\iota_{1}}})
\end{split}
\end{equation*}
Summarizing, network $\mathcal{G}$ (Figure \ref{f:e_coli_network}) generically supports three types of homeostasis: (1) appendage (null-degradation) homeostasis associated with the subnetwork $\{\sigma\}$, (2) structural (Haldane) homeostasis associated with the subnetwork $\{\iota_{2} \rightleftharpoons o\}$ and (3) input counterweight homeostasis associated with the subnetwork $\{\iota_{1} \rightarrow \iota_{2}\}$.
See \cite{madeira2022} for a more conceptual derivation of the determinant $\det\!\big(\langle H \rangle \big)$ using only the topology of the network.

Now, specializing to the model equations \eqref{e_coli_network_notation}, we 
observe that, although the abstract network supports appendage and structural homeostasis, the model equations \eqref{e_coli_network_notation} cannot exhibit these types of homeostasis.
In fact, at equilibrium we have that
\begin{equation}
f_{\sigma, x_{\sigma}} = -\alpha_{2}k_{3}x_{\iota_{2}} - k_{5} < 0
\end{equation}
and
\begin{equation}
f_{o, x_{\iota_{2}}} = \alpha_{1}k_{2}(1 - x_{o}) \neq 0
\end{equation}
The first inequality follows because all parameters are positive and $x_{\iota_{2}}$ is positive at equilibrium.
The last inequality follows because $\alpha_{1}k_{2}(1 - x_{o}) = 0 \Rightarrow x_{o} = 1$, at equilibrium, so one would have $\dot{x}_{o} = - k_{4} \neq 0$.

This leaves the only remaining possibility: input counterweight homeostasis. To verify that the model equations \eqref{e_coli_network_notation} indeed exhibit input counterweight homeostasis we compute, assuming that both $\phi_{x_{\iota_1}}$ and $\phi_{\mathcal{I}}$ are non-zero:
\begin{equation} \label{analysis_e_coli_part_iiii}
\begin{aligned}
f_{\iota_{1}, \mathcal{I}} & = - \phi_{\mathcal{I}} (x_{\iota_1},\mathcal{I})  \, (\gamma_{R} + \gamma_{B}x_{\sigma}^{2}) \\
f_{\iota_{2}, x_{\iota_1}} & = \phi_{x_{\iota_1}} (x_{\iota_1},\mathcal{I}) \, k_{1} \, (1 - x_{\iota_2}) \\
f_{\iota_{2}, \mathcal{I}} & = \phi_{\mathcal{I}} (x_{\iota_1},\mathcal{I})  \, k_{1} \, (1 - x_{\iota_2}) \\
f_{\iota_{1}, x_{\iota_1}} & = -\phi_{x_{\iota_1}} (x_{\iota_1},\mathcal{I}) \, (\gamma_{R} + \gamma_{B}x_{\sigma}^{2})
\end{aligned}
\end{equation}
Thus, for any $C^1$ function $\phi$, we have
\begin{equation} \label{analysis_e_coli_part_v}
f_{\iota_{1}, \mathcal{I}} \, f_{\iota_{2}, x_{\iota_{1}}} - f_{\iota_{2}, \mathcal{I}} \, f_{\iota_{1}, x_{\iota_{1}}} \equiv 0
\end{equation}
In particular, $x_o'(\mathcal{I})\equiv 0$ and so the model equations \eqref{e_coli_network_notation} exhibits, not only infinitesimal homeostasis (of the input counterweight type), but perfect homeostasis.

Moreover, it is easy to check that in the model equations \eqref{e_coli_network_notation}, both $\phi_{x_{\iota_1}}$ and $\phi_{\mathcal{I}}$ are non-zero.
Hence, all the partial derivatives of $f$ in \eqref{analysis_e_coli_part_iiii} are non-zero and so the occurrence of perfect homeostasis in this model is non-trivial.
Figure \ref{e_coli_plot2} (bottom-left) shows the graph of the output variable $x_o = y_p$ (blue) as a function of the parameter $L$.
As expected, it is a constant line over the whole range of $L$.

\begin{figure*}[!htp]
\centering
\includegraphics[width=\linewidth,trim=0cm 1cm 0cm 0.5cm,clip=true]{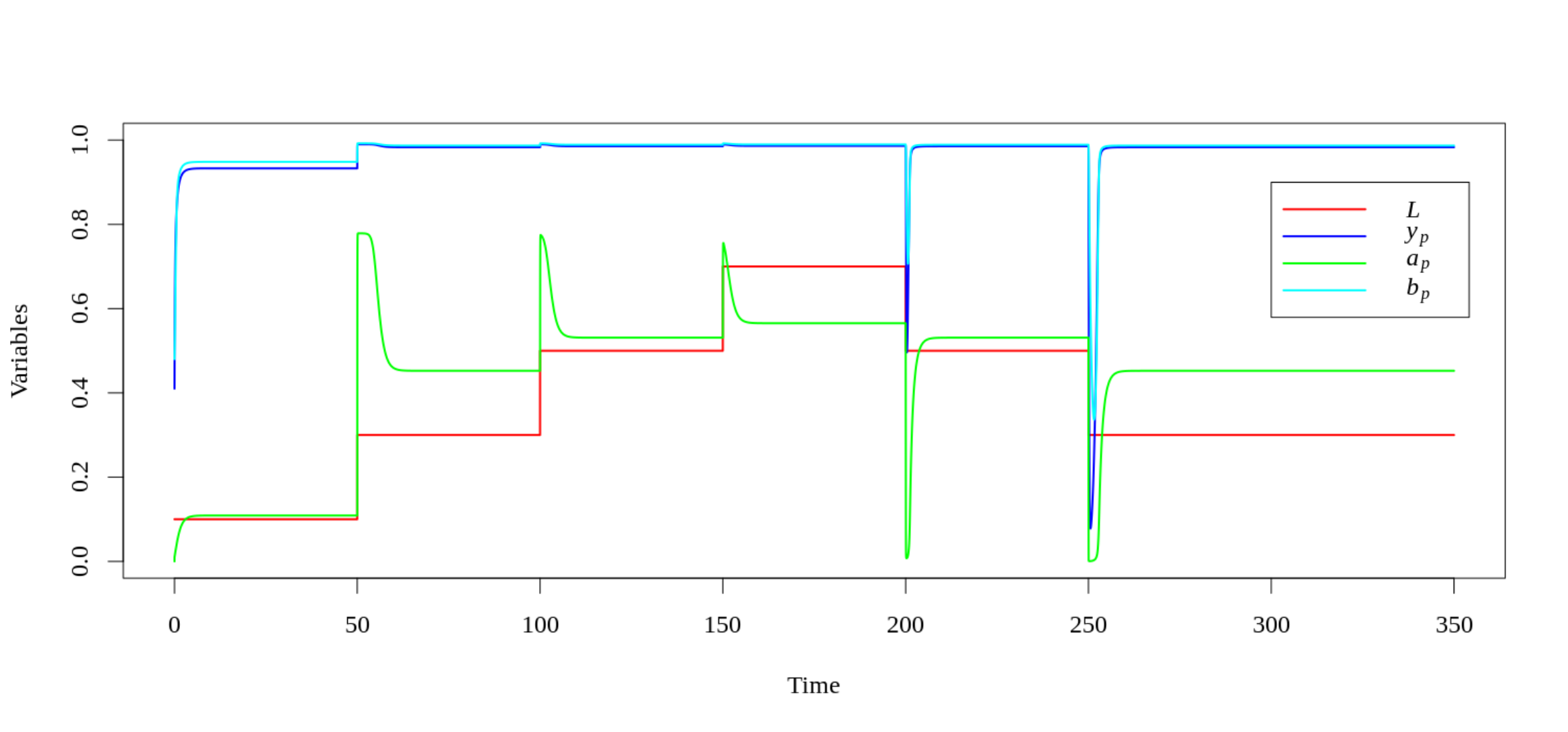} \\ [1ex]
\includegraphics[width=\linewidth,trim=0cm 0.5cm 0cm 1cm,clip=true]{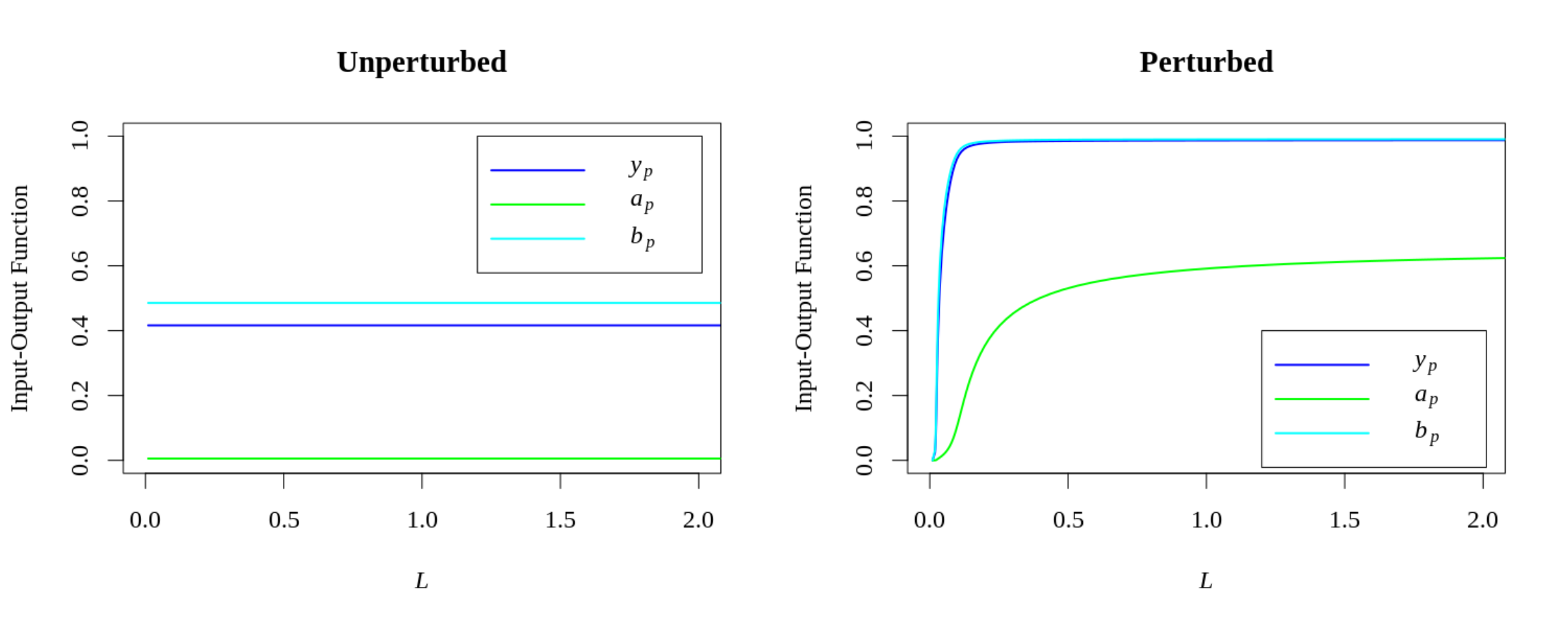}
\caption{\label{e_coli_plot2} (Top) Time series of the model \eqref{e_coli_general_perturbation}, with $\psi(z)=z$ and $\epsilon=0.05$, showing near perfect homeostasis of the original variables $a_p \leftrightarrow x_{\iota_2}$ (green), $b_p \leftrightarrow x_{\sigma}$ (cyan), $y_p \leftrightarrow x_{o}$ (blue) at the corresponding non-dimensional equilibrium. Input parameter $L \leftrightarrow \mathcal{I}$ is given by a step function (red curve). Other parameters were set to non-dimensional values of \cite[Table~2]{edgington2018}.
(Bottom) Input-output functions of the original variables $a_p \leftrightarrow x_{\iota_2}$ (green), $b_p\leftrightarrow x_{\sigma}$ (cyan), $y_p \leftrightarrow x_{o}$ (blue), as functions of input parameter $L \leftrightarrow \mathcal{I}$, for the model \eqref{e_coli_general_perturbation}.
(Left) For $\epsilon=0$, which reduces to the original model \eqref{original_e_coli}, we have perfect homeostasis (constant input-output functions) and (Right) for $\epsilon=0.05$ we have near perfect homeostasis. Time series were computed using the software \textsc{XPPAut} and input-output functions were computed by numerical continuation of an equilibrium point using \textsc{Auto} from \textsc{XPPAut} \cite{bard2002}.
}
\end{figure*}

A consequence of the above calculation is that the property of perfect homeostasis in the model \eqref{e_coli_network_notation} for \textit{E. coli} chemotaxis is \emph{persistent}, in the sense that it does not depend on the values of the parameters of the model. 
Even more, its persistence holds for a much larger set of perturbations than just parameter variation, since there is a functional dependence in the input-output function (see \cite{madeira2022} for the details).

As can be seen in Figure \ref{e_coli_plot1} it looks like the time series of all three variables $a_p$, $b_p$ and $y_p$ exhibit perfect homeostatic behavior.
This is the homeostasis pattern associated with homeostasis type the caused homeostasis.
A theory of homeostasis patterns for networks with multiple input nodes has not been worked yet and thus we cannot predict the homeostasis patterns from the knowledge of the homeostasis types alone.
The only way to determine the homeostasis patterns in this case is by numerical computation.

As we have seen, there is a symmetry between the two first equations of \eqref{e_coli_network_notation} that cause the identical vanishing of the determinant 
$\det\!\big(\langle H \rangle \big)$.
Therefore, it is reasonable to think that in order for a perturbation to disrupt the perfect homeostasis it is necessary that the perturbation breaks the symmetry above.
Consider the $1$-parameter family of perturbations of \eqref{e_coli_network_notation}
\begin{equation} \label{e_coli_general_perturbation}
\begin{aligned}
\dot{x}_{\iota_1} & = \gamma_{R}  \, (1 - \phi(x_{\iota_1}, \mathcal{I})) - \gamma_{B} \, x_{\sigma}^{2} \, \phi(x_{\iota_1}, \mathcal{I}) -\epsilon \, \psi(x_{\iota_{1}}) \\
\dot{x}_{\iota_{2}} & = \phi(x_{\iota_1}, \mathcal{I}) \, k_{1} \, (1 - x_{\iota_2}) \\
& \quad - k_{2} \, (1 - x_{o}) \, x_{\iota_2} - k_{3} \, (1 - x_{\sigma}) \, x_{\iota_2} \\
\dot{x}_{\sigma} & = \alpha_{2} \, k_{3} \, (1 - x_{\sigma}) \, x_{\iota_2} - k_{5} \, x_{\sigma} \\
\dot{x}_{o} & = \alpha_{1} \, k_{2} \, (1 - x_{o}) \, x_{\iota_2} - k_{4} \, x_{o}
\end{aligned}
\end{equation}
where $\epsilon \geqslant 0$ and $\psi$ is a fixed smooth function. 
It is clear that \eqref{e_coli_general_perturbation} is admissible for the network $\mathcal{G}$ for all $\epsilon \geqslant 0$ and any smooth function $\phi$.
Moreover, \eqref{e_coli_general_perturbation} reduces to \eqref{e_coli_network_notation} when $\epsilon=0$.
The expressions for $f_{\iota_{1}, \mathcal{I}}$, $f_{\iota_{2}, \mathcal{I}}$ and $f_{\iota_{2}, x_{\iota{1}}}$ are independent of $\epsilon$ and thus are the same as in the original system \eqref{analysis_e_coli_part_iiii}. 
The expression for $f_{\iota_{1}, x_{\iota{1}}}$ now is
\begin{equation} \label{new_formula_derivative}
f_{\iota_{1}, x_{\iota_1}} = -\phi_{x_{\iota_1}} (x_{\iota_1},\mathcal{I}) \, (\gamma_{R} + \gamma_{B}x_{\sigma}^{2}) - \epsilon \, \psi_{x_{\iota_{1}}}(x_{\iota_{1}})
\end{equation}

Since the equations of $f_{\sigma, x_{\sigma}}$ and of $f_{o, x_{\iota_{2}}}$ are independent of $\epsilon$, the same argument as before shows that system \eqref{e_coli_general_perturbation} does not exhibit appendage or structural homeostasis, for all $\epsilon \geqslant 0$. 
Now to show that the input counterweight homeostasis factor cannot be identically zero we compute
\begin{equation*} \label{final_analysis_e_coli_not_generic_part_i}
 f_{\iota_{1}, \mathcal{I}} \, f_{\iota_{2}, x_{\iota_{1}}} - f_{\iota_{2}, \mathcal{I}} \, f_{\iota_{1}, x_{\iota_{1}}} = \epsilon \, k_{1} \, \phi_{\mathcal{I}} (x_{\iota_1},\mathcal{I}) \, (1 - x_{\iota_2}) 
 \, \psi_{x_{\iota_{1}}}(x_{\iota_{1}})
\end{equation*}
Since $\phi_{\mathcal{I}} (x_{\iota_1},\mathcal{I}) \neq 0$ and $(1 - x_{\iota_2}) \neq 0$, generically, left-handed side vanishes if $\psi_z(z)$ vanishes, generically

Therefore, if the function $\psi_z(z)$ does not vanish on an interval (say, $\psi(z)=\text{constant}$), system \eqref{e_coli_general_perturbation} does not exhibit perfect homeostasis, generically, for $\epsilon > 0$.

Nevertheless, it is still possible that \eqref{e_coli_general_perturbation} exhibits infinitesimal homeostasis.
In fact, the function $\psi_z(z)$ determines the critical points of the input counterweight factor.
That is, $x_o'(\II_0)=0$ at some $\II_0$ if 
\[
 \frac{d}{d\II} \psi_{x_{\iota_1}}\big(x_{\iota_1}(\II)\big) \bigg|_{\II=\II_0} = 0
\]
In particular, if $\psi_z(z)$ has a critical point at $x_{\iota_1}(\II_0)$,
or equivalently, if $\psi_{zz}(z)$ vanishes on $x_{\iota_1}(\II_0)$, then $\II_0$ is a point of infinitesimal homeostasis.

More generally, numerical simulations suggests that, when $\epsilon>0$ is small, system \eqref{e_coli_general_perturbation} displays near perfect homeostasis, see Figure \ref{e_coli_plot2} (bottom-right). 

Finally, we should mention that another extension of the theory of \cite{wang2021} to the case of multiple input nodes and multiple input parameters has been developed \cite{madeira2024}.
Thus, one can apply essentially the same combinatorial tools to study homeostasis in any MISO system given in terms on an input-output network.

\subsection{Period Homeostasis and Circadian Rhythms}
\label{SS:PHCR}

Circadian clocks are endogenous oscillators that are prevalent in organisms from bacteria to humans and control $24$-hour physiological and behavioral processes in these organisms. 
Such oscillations are thought to work as pacemakers, and the maintenance of their characteristic period to a constant value against changes in external conditions is important. 
To date, it is known that some of these oscillators display `period homeostasis' against external changes \cite{bell2005}.
The most prominent example of period homeostasis is called \emph{temperature compensation}, i.e., persistence of the period against temperature changes \cite{hk2012}.
In addition to temperature compensation, circadian clocks display persistence against changes in trophic conditions, which is called \emph{nutrient compensation} \cite{phong2013}. 
Such homeostatic features (i.e., temperature and nutrient compensation) are known to be universally
conserved across a wide range of organisms \cite{hk2014}.

The mammalian circadian system is composed of a hierarchical multi-oscillator structure, with the central clock located in the suprachiasmatic nucleus (SCN) of the hypothalamus regulating the peripheral clocks found throughout the body. 
At the cellular level, each cell is governed by its own independent clock; and yet, these cellular circadian clocks in the SCN form regional oscillators that are further coupled to one another to generate a single rhythm for the tissue. 
The oscillatory coupling within and between the regional oscillators appears to be critical for the extraordinary stability and the wide range of adaptability of the circadian clock \cite{dibner2010}.

In the past two decades, key clock genes have been discovered in mammals and shown to be interlocked in transcriptional and translational feedback loops.
Interestingly, the mammalian circadian clock also seems to be persistent to global changes in transcription rates and gene dosage \cite{dibner2009}.

To investigate the dynamics and mechanisms of
the intracellular feedback loops in circadian clocks, a number of mathematical models based on systems of ordinary differential equations have been developed. 
The majority of models use Hill functions to describe transcriptional repression in a way similar to the Goodwin model \cite{kim2016}.

In order to understand what period homeostasis means in these models let us consider a parametrized system of ODEs as in \ref{general_dynamics} with $\II\in\mathcal{R}$ for simplicity.
But now we suppose that $(X(t)^*,\II^*)$ is a stable periodic solution of \eqref{general_dynamics} with minimal period $\tau^*$.

By the \emph{continuation of hyperbolic periodic solutions} theorem \cite[Thm. 1.1]{hale2013}, there is a smooth family of periodic solutions $X(t,\II)$ defined in a neighborhood of $\II^*$, with smooth minimal period $\tau(\II)$, such that $X(t,\II^*) = X^*(t)$, $\tau(\II^*)=\tau^*$ and $\dot{X}(t,\II) = F(X(t,\II),\II)$ for all $t\in\mathbb{R}$. 

There are several proofs of this theorem, mostly often using the method of Poincar\'e sections. 
Hale and Raugel \cite{hale2013} considers a simple method using the Fredholm alternative and the Lyapunov-Schmidt procedure.

Let us call the mapping $\II\to\tau(\II)$ the \emph{period input-output function}.
Then we can adapt Definition \ref{D:inf_homeo} to give a notion of infinitesimal period homeostasis.

\begin{definition} \label{D:period_homeostasis}
Let $\tau(\II)$ be the period input-output function. 
We say that $\tau(\II)$ exhibits \emph{infinitesimal period homeostasis} at the point $\II_0 \in \operatorname{dom}(z)$ if
\begin{equation} \label{period_homeostasis_condition}
  \tau'(\II_0) = 0
\end{equation}
That is, $\II_0$ is a \emph{critical point} of the the period input-output function $\tau$. 
\END
\end{definition}

Of course one can easily formulate analog definitions for \emph{perfect period homeostasis} and \emph{near-perfect period homeostasis}.
Accordingly, we will refrain from writing down these definitions.

Furthermore, it is possible to adapt the proof the continuation theorem given in \cite{hale2013} to obtain an explicit formula for the derivative of $\tau(\II)$.
Let $\II_0$ be fixed and set $X(t,\II_0)=X_0(t)$ and $\tau(\II_0)=\tau_0$
\begin{equation} \label{e:formula_tau_prime}
 \tau'(\II_0) = \frac{1}{\tau_0}
 \int_{0}^{\tau_0} \!\!\!\!
 \langle F_{\II_0}(X_0(t)) , w^*(t) \rangle \,dt
\end{equation}
where $w^*$ is the unique solution of the \emph{adjoint Floquet operator} of the solution $X_0(t)$, with periodic boundary condition $w^*(t+\tau_0)=w^*(t)$ and the normalization condition $\langle X(t,\II_0),w^*(t)\rangle = 1$.
Here, $\langle \cdot , \cdot\rangle$ is the standard scalar product on $\mathcal{R}^n$.

The adjoint Floquet operator of the solution $X_0(t)$ is defined as follows.
Let $A(t)=D_Xf_{(X_{0}(t),0)}$ be the linearization of $f$ about the periodic solution $X_0(t)$, then $A(t)$ is a $\tau_0$-periodic function.
Then $\mathcal{L}=\frac{d}{dt}-A(t)$ and $\mathcal{L}^*=\frac{d}{dt}+A^t(t)$ are the \emph{Floquet operator} and the \emph{adjoint Floquet operator} of the solution $X_{0}(t)$, respectively.
See also \cite[]{yu2022}, where the function $w^*(t)$ is called \emph{infinitesimal phase response curve (iPRC)} of the periodic solution $X_0(t)$.

Now that we have a precise notion of period homeostasis and a means to compute it we could try to apply it to some example.
Although, formula \eqref{e:formula_tau_prime} is quite nice conceptually, it seems very hard to compute in concrete examples.
An alternative is to use numerical methods for continuation periodic solutions, such as \textsc{Auto} from \textsc{XPPAut} \cite{bard2002}.
We will illustrate this idea of numerical continuation periodic solutions with a simple model for the mammalian circadian system.

\begin{example}[Mammalian Clock Model] \normalfont
Kim and Forger \cite{kim2012} propose to use a mechanism of protein sequestration-based (PS) transcriptional repression rather than Hill-type (HT) repression.
Since Goodwin’s and Goldbeter’s pioneering studies, Hill functions have been widely used to model the negative feedback loops (NFL) in circadian clocks of diverse organisms, including \emph{Neurospora}, \emph{Drosophila} and mammals \cite{kim2016}.
In \cite{kim2012} the authors formulate three simple mathematical models for the mammalian clock (based on a more detailed model), that reproduces a surprising amount of experimental data on mammalian circadian rhythms.
See also Yao \etal~\cite{tyson2022} for a careful mathematical analysis.

The Simple Negative Feedback (SNF) loop model is generated by modifying the well-studied Goodwin model to include an activator, which can be inactivated when bound in complex with the repressor. 
The state variables are: the mRNA concentration ($M$), the cytoplasmic protein concentration ($P_c$) and the nucleous protein concentration $P$.
The mRNA transcription is proportional to the \% of unbound free activator $f(P, A, \kappa)$,
which indicates the activity of the promoter.
Here, the parameter $A$ is the upstream transcriptional factor that interacts with $P$ through binding with dissociation constant $\kappa$.
That is, the transcription of the mRNA is repressed by $P$ by inhibition of the activator $A$ 
See \cite{kim2012,tyson2022} for details.

The other two models can be obtained from the SNF through the addition of negative or positive feedback loops controlling the production of activator $A$, which then becomes another state variable.
The Positive-Negative Feedback (PNF) loop model and the Negative-Negative Feedback (NNF) loop mode are obtained from the (SNF) adding two state variables:
$A$ is the protein that activates the transcription of $M$ and $R$ is the mRNA the is associated to protein $A$.
See \cite{kim2012,tyson2022} for details.

It is possible to write down the non-dimensional equations for the three models at same time
\begin{equation}
\begin{split}
 \dot{M} & = \alpha_1 f(P,A,\kappa) - \beta_1 M \\
 \dot{P_{\mathrm{c}}} 
 & = \alpha_2 M - \beta_2 P_{\mathrm{c}} \\
 \dot{P} & = \alpha_3 P_{\mathrm{c}} - \beta_3 P \\
 \dot{R} & = \gamma_1 f(P,A,\kappa) -\delta_1 R \\
 \dot{A} & = \gamma_2 R^{\pm 1} - \delta_2 A
\end{split}
\end{equation}
where
\begin{equation*}
f(A,P,\kappa)  = \frac{1}{2A}
 \bigg(A-P-\kappa+\sqrt{(A-P-\kappa)^2-4A\kappa}\bigg)
\end{equation*}
and $\alpha_1$, $\alpha_2$, $\alpha_3$, $\beta_1$,
$\beta_2$, $\beta_3$, $\gamma_1$, $\gamma_2$, $\delta_1$, $\delta_2$, $\kappa$ are parameters.

The SNF model is obtained by setting $\gamma_i=\delta_i=0$. 
Then the equations for $\dot{R}$ and $\dot{A}$ disappear and $A$ becomes another parameter.
The PNF model is given by picking the positive ($+1$) exponent in $R^{\pm 1}$ on the right-handed side of the equation for $\dot{A}$ and the NNF model is given by picking the negative ($-1$) exponent in $R^{\pm 1}$ on the right-handed side of the equation for $\dot{A}$.

It is shown in \cite{kim2012} that the three models have stable periodic solutions for a large range of the parameters. See also \cite{tyson2022} for an analysis of the Hopf bifurcation that generates these solutions.

In Figure \ref{F:kim-forger} we show numerical computation of the period input-output function for the three models.
In the SNF model we use $A$ as the input parameter and the PNF and NNF models we use $\gamma_2$ as the input parameter.
The output is the period $\tau$ of the periodic solution in all three cases.

In all three cases we see that the period input-output function $\tau$ has an infinitesimal homeostasis point at the left end of each curve and exhibits near-perfect homeostasis at the right end of the graph.
Moreover, in the PNF case the periodic orbit undergoes a saddle-node bifurcation.
\END
\end{example}

\begin{figure}[!htb]
\begin{subfigure}[b]{0.45\textwidth}
\centering
\stackinset{r}{8pt}{t}{9pt}{\includegraphics[width=0.4\textwidth,trim=1.25cm 2cm 1.5cm 3cm, clip=true]{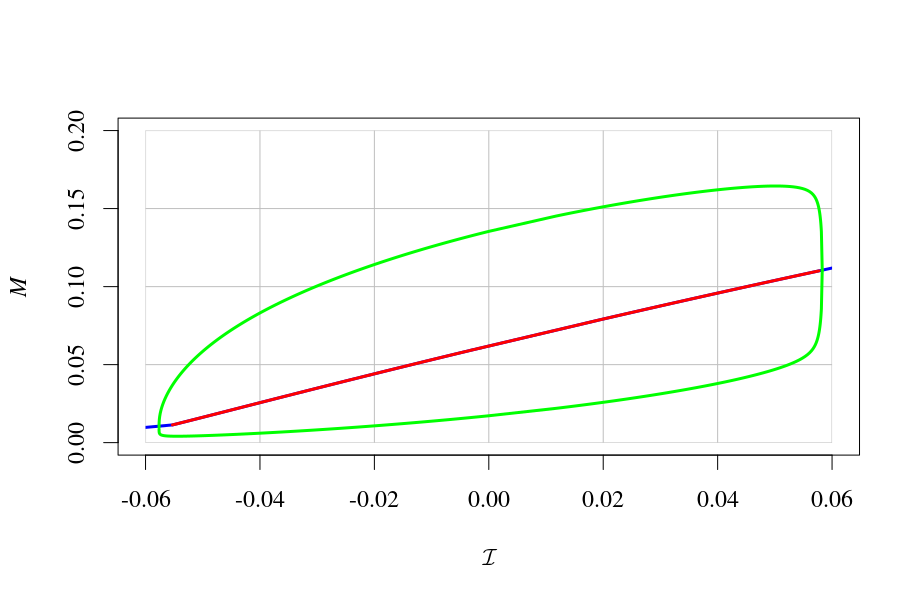}}{%
\includegraphics[width=\textwidth,trim=0.25cm 1cm 1.5cm 2.5cm, clip=true]{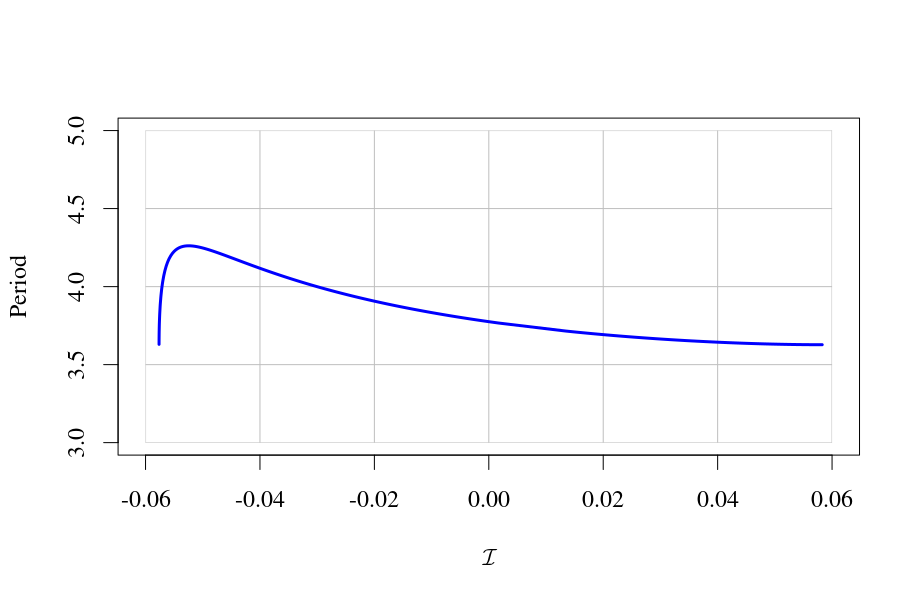}}
\caption{SNF model (Input = $A$)}
\label{F:SNF}
\end{subfigure} \quad
\centering
\begin{subfigure}[b]{0.45\textwidth}
\centering
\stackinset{r}{8pt}{t}{9pt}{\includegraphics[width=0.45\textwidth,trim=1.25cm 2cm 1.5cm 3cm, clip=true]{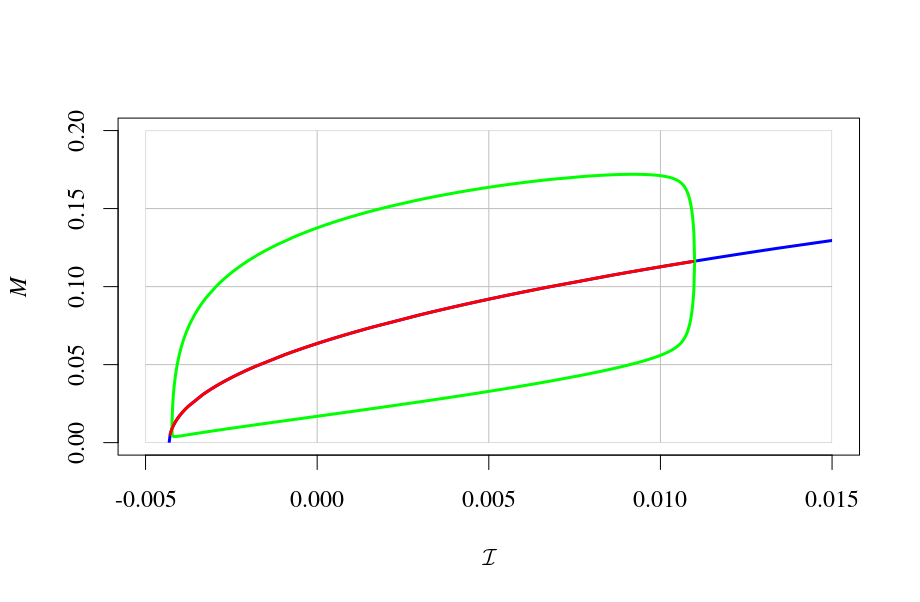}}{%
\includegraphics[width=\textwidth,trim=0.25cm 1cm 1.5cm 2.5cm, clip=true]{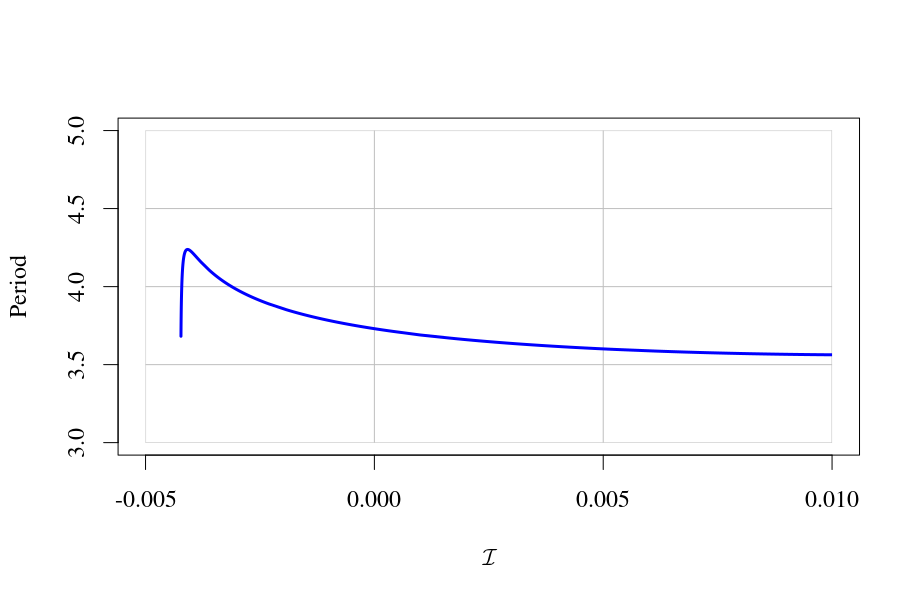}}
\caption{NNF model (Input = $\gamma_2$)}
\label{F:NNF}
\end{subfigure} \quad
\centering
\begin{subfigure}[b]{0.45\textwidth}
\centering
\stackinset{r}{8pt}{t}{9pt}{\includegraphics[width=0.35\textwidth,trim=1.25cm 2cm 1.5cm 3cm, clip=true]{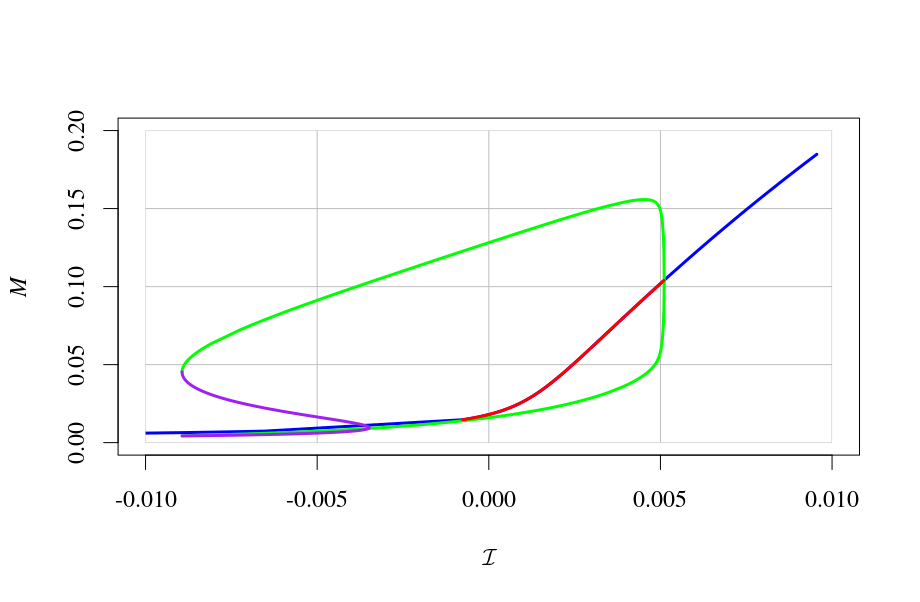}}{%
\includegraphics[width=\textwidth,trim=0.25cm 1cm 1.5cm 2.5cm, clip=true]{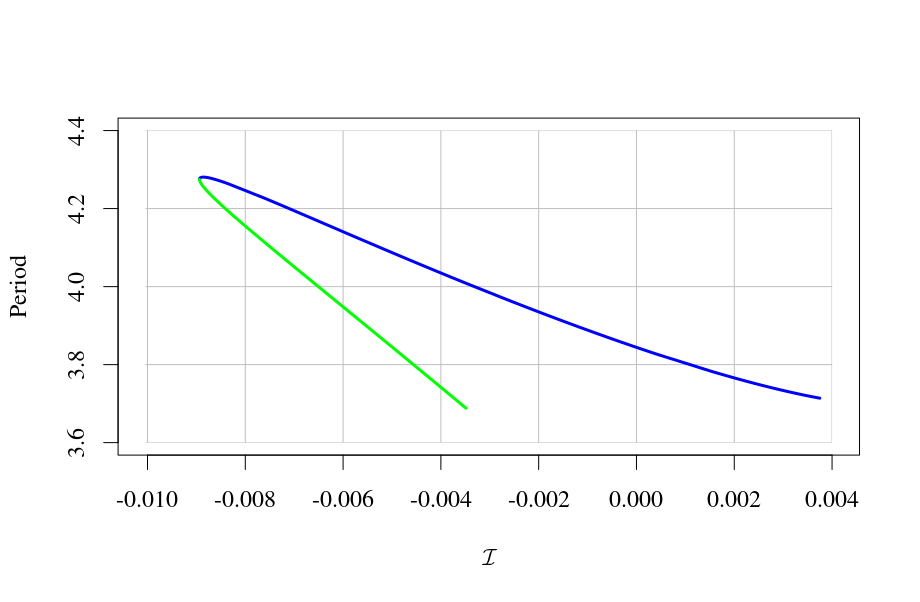}}
\caption{PNF model (Input = $\gamma_2$)}
\label{F:PNF}
\end{subfigure}
\caption{\label{F:kim-forger}
Kim-Forger Models \cite{kim2012,kim2016,tyson2022}.
(Outer panels) Graph of the period input-output function $\tau$ as function of the input parameter;
(blue) stable periodic solution, (green) unstable periodic solution.
(Inner panels) Projection of the bifurcation diagram onto the $M$ coordinate showing the range were the periodic solution exists; (blue) stable equilibrium, (red) unstable equilibrium, (green) stable periodic solution, (violet) unstable periodic orbit.
Period input-output functions and bifurcation diagrams were computed by numerical continuation of a periodic orbit using \textsc{Auto} from \textsc{XPPAut} \cite{bard2002}.}
\end{figure}

Notwithstanding its biological relevance, period homeostasis seems to be quite different from the steady-state homeostasis studied before.
For instance, it does not need a network structure or an observable on the state space to be defined.

However, it is possible to define a notion of homeostasis associated with periodic orbits that is analogous to the steady-state theory.
Motivated by problems in neuroscience Yu and Thomas \cite{yu2022} provided the first step towards such a theory by introducing a new class of input-output functions.

Yu and Thomas \cite{yu2022} start, as before, with a parametrized system of ODEs as in \ref{general_dynamics} possessing a
stable periodic solution $(X(t)^*,\II^*)$ with minimal period $\tau^*$ and consider its continuation $X(t,\II)$ defined in a neighborhood of $\II^*$, with smooth minimal period $\tau(\II)$, such that $X(t,\II^*) = X^*(t)$, $\tau(\II^*)=\tau^*$ and $\dot{X}(t,\II) = F(X(t,\II),\II)$ for all $t\in\mathbb{R}$.
Now, given a smooth observable $\phi$ defined on the state space consider the input-output function given by \emph{averaging $\phi$ along $X(t,\II)$}
\begin{equation} \label{e:average_IO}
 \overline{\!\phi}(\II) = \frac{1}{\tau(\II)}
 \int_{0}^{\tau(\II)} \!\!\!\!\!
 \phi(X(t,\II)) \,dt
\end{equation}
Note that the minimal period function $\tau$ now appears in the definition of the averaging operator.

For example, if the system of ODEs has an underlying input-output network $\mathcal{G}$ with output node $x_o$ the we can define the input-output function
\begin{equation} \label{e:average_x_o}
\overline{\!x}_o(\II) = \frac{1}{\tau(\II)}
 \int_{0}^{\tau(\II)} \!\!\!\!\!
 x_o(t,\II) \,dt
\end{equation}
where $x_o(t,\II)$ is the $o$-coordinate of $X(t,\II)$.
It is clear that this new class of input-output functions is provides the right counterpart in the periodic case to the steady-state input-output functions.

The main result of \cite{yu2022}  is a formula for the derivative of the input-output function $\,\overline{\!\phi}(\II)$ with respect to $\II$.
To accomplish this they propose ageneralization of the \emph{infinitesimal shape response curve (iSRC)} recently introduced in \cite{wgct2021}.
The iSRC complements the well-known infinitesimal phase response curve (iPRC). 
The latter quantifies, to linear order, the effect of an instantaneous perturbation on the timing of the periodic solution, and also captures the
cumulative effect of a sustained perturbation its period. 
In contrast, the iSRC captures the effect of a sustained or parametric perturbation on the shape of the periodic solution, to linear order.

The formula for the derivative of $\,\overline{\!\phi}(\II)$ obtained by \cite{yu2022} is very similar to formula \eqref{e:formula_tau_prime}
Let $\II_0$ be fixed and set $X(t,\II_0)=X_0(t)$ and $\tau(\II_0)=\tau_0$
\begin{equation} \label{e:formula_barphi_prime}
 \overline{\!\phi}^{\,\prime}(\II_0) = \frac{1}{\tau_0}
 \int_{0}^{\tau_0} \!\!\!\!
 \langle \nabla\phi(X_0(t)), \omega^*(t) \rangle \,dt
\end{equation}
where $\omega^*(t)$ is the generalized iSRC of the periodic solution $X_0(t)$ \cite{yu2022}.
That, is $\omega^*(t)$ is a $\tau_0$-periodic solution of an inhomogeneous linear equation similar to the adjoint Floquet operator.

The similarity with formula \eqref{e:z_prime} is noteworthy.
In both cases, the derivative is given by the scalar product (the average operation in formula \eqref{e:formula_barphi_prime} can be written as scalar product in the space of $\tau_0$-periodic functions) between the gradient $\nabla\phi$ of the observable and a solution of a linear equation obtained by linearization of the parametrized family of nonlinear equations at the appropriate solutions.
As we noted before, the period input-output function $\tau$ is an exceptional observable that do not have a counterpart in the steady-state context.

The work of \cite{yu2022} is an important advance in the extension of homeostasis to the periodic case.
However, the theory is much less developed in relation to the steady-state case.
For instance, an extension of the combinatorial structure theory for homeostasis in input-output networks with input-output functions of the from \eqref{e:average_x_o} remains elusive. 
Extending the classification theory to this context would be an interesting future direction.

\section{Conclusion and Outlook}

Undoubtedly there are many other directions for research on infinitesimal homeostasis that we have not explored in this review.

For example, we have mentioned that Duncan and Golubitsky~\cite{duncan2019} have investigated forms of combining infinitesimal homeostasis and bifurcations.
In fact, there are two generalizations of homeostasis theory that are motivated by codimension arguments in bifurcation theory.  
The first is related to chair homeostasis, where the plateau in the input-output function is flatter than expected.
See Nijhout~\etal~\cite{nbr2014}, Golubitsky and Stewart~\cite{gs2018}, and Reed~\etal~\cite{rbgsn2017}. 
The second is related to the existence of `homeostasis mode interaction', where two infinitesimal homeostasis types occur at the same equilibrium. 
See Duncan~\etal~\cite{duncan2025}. 
Interestingly, the simultaneous appearance of different homeostasis types leads to bifurcation in the family of equilibria that generates the homeostasis.  
An early example of this phenomenon is discussed in Duncan and Golubitsky~\cite{duncan2019}.  
A related biochemical example of multiple types of infinitesimal homeostasis occurring on variation of just one parameter is found in Reed~\etal~\cite{rbgsn2017}.

Another line of research is the role of homeostasis in \emph{chemical reaction network theory} \cite{feinberg2019,craciun2018}. 
Craciun and Deshpande~\cite{craciun2022} initiated the investigation of infinitesimal homeostasis in chemical reaction networks by relating it to the notion of network injectivity.
More precisely, the authors describe a procedure for checking whether a reaction network may admit infinitesimal homeostasis by constructing a modified network and checking if it is injective \cite[Thm. 3.4]{craciun2022}.
They also obtain a sufficient condition for perfect homeostasis \cite[Thm. 3.5]{craciun2022}. 
See also \cite{araujo2023} for the occurrence of perfect homeostasis in chemical reaction networks from a control theoretic point of view.

Finally, we should mention the possibility to extend the theory the stochastic realm.
In \cite{briat2016,aoki2019} the authors study the occurrence of perfect homeostasis in a stochastic version of the antithetic model (see Example \ref{EX:ANTITHETIC_FB} of subsection \ref{SS:METAL}).
More generally, the typical stochastic processes that are considered in \cite{briat2016,aoki2019} and many other modeling approaches to biological systems are called \emph{stochastic chemical reaction networks} \cite{anderson2015}.

In the stochastic setup a model is given by a continuous-time ergodic Markov process $X_t$ on a state space $\mathcal{S}$ with stationary (equilibrium) measure $\mu$.
It is assume that the model depends on an input parameter $\II$ and that the stationary measure is a function of the input parameter: $\mu_{\II}$.
Given an observable,i.e. an integrable function $\phi:\mathcal{S} \to \mathbb{R}$, its expectation with respect to the stationary measure $\mu_{\II}$ is
\[
 \mathcal{Z}(\II) = \mathbf{E}_{\mathcal{I}}[\phi] = \int \phi \; d\mu_\mathcal{I}
\]
is a function of the parameter $\mathcal{I}$. 
We say that the system exhibits \emph{linear response} for $\phi$ at $\mathcal{I}_0$ if the derivative
\[
 \mathcal{Z}'(\II_0) = \frac{d}{d\mathcal{I}} \mathbf{E}_{\mathcal{I}}[\phi] \bigg|_{\mathcal{I}=\mathcal{I}_0}
\]
exists.
Linear response can be rigorously justified for a large class of stochastic processes \cite{hairer2010,mackay2011}.
Taking $\mathcal{Z}$ as the input-output function, with $\phi$ fixed, the definitions of perfect, near-perfect and infinitesimal homeostasis readily go through this new context.
Now the basic question is: Which parts of the `machinery' developed for the deterministic setting (e.g. singularity structure, input-output networks and the combinatorial classification of homeostasis types, homeostasis patterns, etc.) can be extended to the stochastic version?
Even more so, in the stochastic setting there are many problems and questions that have no counterpart in the deterministic setting.






\section*{Acknowledgments}
We thank Janet Best, William Duncan, Jo\~ao Lu\-iz Ma\-dei\-ra, Fred Nijhout, Michael Reed, John Tyson and Yangyang Wang for helpful discussions.
We thank the organizers of the Workshop ``Dynamical Systems in the Life Sciences'', Harsh Jain,
Wenrui Hao, Grzegorz Rempala and Yangyang Wang.
The research of FA was supported by Funda\c{c}\~ao de Amparo \`a Pes\-qui\-sa do Estado de S\~ao Paulo (FAPESP) grants 2019/21181-0 and 2023/04839-7.




\bibliographystyle{cas-model2-names}

\bibliography{refs}



\end{document}